\documentclass[prx,aps,twocolumn,groupedaddress,nofootinbib,superscriptaddress,preprintnumbers]{revtex4-2}

\usepackage{graphicx}
\usepackage[nointegrals]{wasysym} %
\usepackage[export]{adjustbox}
\usepackage{amsmath,amsfonts,amssymb,latexsym}
\usepackage{hhline}
\usepackage{bm}
\usepackage{verbatim}
\usepackage{enumitem}
\usepackage{mathrsfs}
\usepackage{slashed}
\usepackage{empheq}

\newcommand{\Tr}{\text{Tr}}

\newcommand{\sgn}{{\rm sgn}}

\newcommand{\p}{\partial}

\newcommand{\lan}{\langle}
\newcommand{\ran}{\rangle}

\newcommand{\da}{{\dagger}}
\newcommand{\doa}{\downarrow}

\newcommand{\ob}[1]{\mkern 1.5mu\overline{\mkern-1.5mu#1\mkern-1.5mu}\mkern 1.5mu}

\newcommand{\ra}{\rightarrow}
\newcommand{\lra}{\leftrightarrow}

\newcommand{\wt}{\widetilde}

\newcommand{\uvx}{{\mathbf{\hat x}}}
\newcommand{\uvy}{{\mathbf{\hat y}}}
\newcommand{\uvz}{{\mathbf{\hat z}}}

\renewcommand{\(}{\left(}
\renewcommand{\)}{\right)}

\newcommand{\tp}{\otimes}

\newcommand{\twp}{{2\pi}}

\newcommand\bpm            {\begin{pmatrix}}
	\newcommand\epm           {\end{pmatrix}}

\newcommand{\ms}{\medskip}
\newcommand{\bs}{\bigskip}

\def\app#1#2{%
	\mathrel{%
		\setbox0=\hbox{$#1\sim$}%
		\setbox2=\hbox{%
			\rlap{\hbox{$#1\propto$}}%
			\lower1.1\ht0\box0%
		}%
		\raise0.25\ht2\box2%
	}%
}

\newcommand{\tw}{\textwidth}

\newcommand{\vp}{\varphi}

\newcommand{\ct}{\Theta}

\newcommand{\inv}{^{-1}}

\newcommand{\ope}\odot

\newcommand{\sqig}{\rightsquigarrow}

\usepackage{manfnt}

\newcommand{\bi}{\begin{itemize}}
	\newcommand{\ei}{\end{itemize}}

\newcommand{\igptc}[1]{\vcenter{\hbox{\includegraphics[width=.3\textwidth]{#1}}}}

\newcommand{\igpfoc}[1]{\vcenter{\hbox{\includegraphics[width=.45\textwidth]{#1}}}}

\usepackage{amsthm}

\newtheorem{corollary}{Corollary}
\newtheorem{proposition}{Proposition}
\newtheorem{fact}{Fact}
\newtheorem{claim}{Claim}
\newtheorem{lemma}{Lemma}

\theoremstyle{definition}
\newtheorem{example}{Example}
\newcommand\bpro		  {\begin{proposition}}
	\newcommand\epro 		  {\end{proposition}}
\newcommand\bproof			  {\begin{proof}}
	\newcommand\eproof 		  {\end{proof}}

\newcommand\ed            {\end{definition}}

\newcommand\be            {\begin{equation}}
\newcommand\ee            {\end{equation}}
\newcommand\ba            {\begin{aligned}}
\newcommand\ea            {\end{aligned}}
\newcommand\bea{\begin{equation}\begin{aligned}}
	\newcommand\eea{\end{aligned}\end{equation}}

\newcommand{\ethan}[1]{ { \color{blue} \footnotesize \textsf{ethan: \textsl{#1}} }}

\usepackage[dvipsnames]{xcolor}

\definecolor{nicegreen}{rgb}{0.3, 0.75, 0.2}

\usepackage{hyperref} 
\hypersetup{final}
\hypersetup{colorlinks, citecolor=red, linkcolor=red, urlcolor=red} 

\newcommand{\sss}{\subsubsection}
\renewcommand{\ss}{\subsection}

\renewcommand{\a}{\alpha}

\renewcommand{\d}{\delta}
\newcommand{\De}{\Delta}
\newcommand{\g}{\gamma}

\newcommand{\s}{\sigma}

\newcommand{\ep}{\varepsilon} %
\renewcommand{\l}{\lambda}

\renewcommand{\O}{\Omega}

\renewcommand{\r}{\rho}

\newcommand{\zt}{\mathbb{Z}_2}
\newcommand{\zn}{\mathbb{Z}_N}

\newcommand{\nn}{\mathbb{N}}

\newcommand{\rr}{\mathbb{R}}
\newcommand{\qq}{\qquad}

\newcommand{\zz}{\mathbb{Z}}

\newcommand{\mcb}{\mathcal{B}}

\newcommand{\mco}{\mathcal{O}}
\newcommand{\mcu}{\mathcal{U}}

\newcommand{\mcd}{\mathcal{D}}

\newcommand{\mch}{\mathcal{H}}
\newcommand{\mca}{\mathcal{A}}

\newcommand{\mcm}{\mathcal{M}}

\newcommand{\sfH}{\mathsf{H}}

\usepackage[mathscr]{eucal} %

\newcommand{\tta}{\mathtt{a}}
\newcommand{\ttb}{\mathtt{b}}
\newcommand{\ttc}{\mathtt{c}}
\newcommand{\ttd}{\mathtt{d}}
\newcommand{\tte}{\mathtt{e}}
\newcommand{\ttf}{\mathtt{f}}
\newcommand{\ttg}{\mathtt{g}}
\newcommand{\tth}{\mathtt{h}}

\newcommand{\ttk}{\mathtt{k}}

\newcommand{\ttn}{\mathtt{n}}
\newcommand{\tto}{\mathtt{o}}

\newcommand{\ttr}{\mathtt{r}}

\newcommand{\ttv}{\mathtt{v}}
\newcommand{\ttw}{\mathtt{w}}
\newcommand{\ttx}{\mathtt{x}}
\newcommand{\tty}{\mathtt{y}}
\newcommand{\ttz}{\mathtt{z}}
\renewcommand{\tt}[1]{\mathtt{#1}}

\newcommand{\kb}[2]{|{#1}\rangle\langle{#2}|}
\renewcommand{\k}[1]{|#1\rangle}
\newcommand{\ket}[1]{|#1\rangle}
\newcommand{\proj}[1]{|#1\rangle\langle#1|}
\newcommand{\bra}[1]{{\langle #1}|}

\usepackage{dcolumn}

\usepackage{pifont}
\usepackage{ulem}

\usepackage[caption=false]{subfig}

\usepackage{enumitem}

\usepackage{amsthm}
\usepackage{physics}

\usepackage{bm}

\renewcommand{\bot}{\bigotimes}

\DeclarePairedDelimiterX{\infdivx}[2]{(}{)}{%
	#1\;\delimsize\|\;#2%
}

\renewcommand{\bs}{{\sf BS}}
\newcommand{\typdehn}{\ob{{\sf Dehn}}}
\newcommand{\el}{{\sf EL}}
\newcommand{\dehn}{{\sf Dehn}}
\newcommand{\dyn}{{\sf Dyn}}
\newcommand{\cg}{{\mathsf{CG}}}
\newtheorem*{theorem*}{Theorem}
\newtheorem{prop}{Proposition}

\newtheorem{observation}{Observation}

\renewcommand\qq{\qquad}

\renewcommand{\bot}{\bigotimes}

\usepackage{accents}

\newcommand{\oEE}{\mathop{\mathbb{E}}}

\begin{document}

	\title{Glassy word problems: ultraslow relaxation, Hilbert space jamming, and computational complexity}
	
	\author{Shankar Balasubramanian}
	\affiliation{Department of Physics, Massachusetts Institute of Technology, Cambridge, MA 02139, USA}
	\affiliation{Center for Theoretical Physics, Massachusetts Institute of Technology, Cambridge, MA 02139, USA}
	\affiliation{Kavli Institute for Theoretical Physics, University of California, Santa Barbara, CA 93106, USA}
	
	\author{Sarang Gopalakrishnan}
	\affiliation{Department of Electrical and Computer Engineering, Princeton University, Princeton NJ 08544, USA}
	
	\author{Alexey Khudorozhkov}
	\affiliation{Department of Physics, Boston University, Boston, MA 02215, USA}
	
	\author{Ethan Lake}
	\affiliation{Department of Physics, Massachusetts Institute of Technology, Cambridge, MA 02139, USA}
	\affiliation{Department of Physics, University of California Berkeley, Berkeley, CA 94720, USA}
	\begin{abstract}
		
		We introduce a family of local models of dynamics based on ``word problems'' from computer science and group theory, for which we can place rigorous lower bounds on relaxation timescales. These models can be regarded either as random circuit or local Hamiltonian dynamics, and include many familiar examples of constrained dynamics as special cases. The configuration space of these models splits into dynamically disconnected sectors, and for initial states to relax, they must ``work out'' the other states in the sector to which they belong. When this problem has a high time complexity, relaxation is slow. In some of the cases we study, this problem also has high space complexity. When the space complexity is larger than the system size, an unconventional type of jamming transition can occur, whereby a system of a fixed size is not ergodic, but can be made ergodic by appending a large reservoir of sites in a trivial product state. This manifests itself in a new type of Hilbert space fragmentation that we call {\it fragile fragmentation}.  We present explicit examples where slow relaxation and jamming strongly modify the hydrodynamics of conserved densities. In one example, density modulations of wavevector $q$ exhibit almost no relaxation until times $O(\exp(1/q))$, at which point they abruptly collapse. We also comment on extensions of our results to higher dimensions.
	\end{abstract}
	
	\maketitle
	
	\section{Introduction}
	
	A common paradigm in quantum dynamics is that isolated quantum systems usually thermalize if one waits long enough \cite{deutsch1991quantum, srednicki1994chaos, rigol2008thermalization, nandkishore2015many-body, dalessio2016from}. Indeed, assuming that interactions are spatially local, quantum systems tend to approach a form of local equilibrium on a timescale that is independent of system size, with the late-time dynamics governed by the hydrodynamic relaxation of a small number of conserved densities. The main possible exception to this rule is the many-body localized phase in strongly disordered systems \cite{alet2018many-body, abanin2019colloquium}, or in systems which effectively self-generate strong disorder ~\cite{schiulaz2014ideal,grover2014quantum,yao2016quasi}. 
	The structures that can be rigorously shown to arrest thermalization in translation-invariant quantum systems---such as integrability \cite{retore2022introduction}, quantum scars \cite{shiraishi2017systematic, bernien2017probing, turner2018quantum, turner2018weak, moudgalya2018exact, lin2019exact, schechter2019weak, moudgalya2022quantum, chandran2023quantum}, dynamical constraints \cite{sala2020ergodicity, khemani2020localization,motrunich,pancotti2020quantum,van2015dynamics,brighi2023hilbert,lan2018quantum,garrahan2018aspects, detomasi2019dynamics,moudgalya2022thermalization}, etc.---are usually fine-tuned in some sense. However, they are still of experimental relevance since many realistic systems are near the fine-tuned points where these phenomena occur \cite{kohlert2023exploring,kinoshita2006quantum,scheie2021detection,malvania2021generalized,wei2022quantum,scherg2021observing}. Systems near these points have long relaxation timescales and approximate conservation laws that are essentially exact on the timescale of realistic experiments on noisy quantum hardware and cold atom systems. Although the algebraic structure of integrable systems and systems with many-body scars has been well studied, a general understanding of the extent to which local Hilbert space constraints can arrest thermalization is still under development. 
	
	In this work we introduce an alternative viewpoint for understanding thermalization in a large class of one-dimensional models with local Hilbert space constraints.  We begin by developing a general framework for characterizing models with constrained dynamics in terms of semigroups, algebraic structures that resemble groups but need not have inverses or an identity. This approach reproduces examples of constrained models known in the literature, but also provides us with new examples with unusual properties.  In particular, it enables us to leverage ideas from the field of geometric group theory to construct (a) models with both {\it exponentially long} relaxation times and {\it sharp} thermalization transitions, and (b) models where thermalization never occurs due to an unusual type of 
	ergodicity breaking we call ``Hilbert-space jamming'' or ``fragile fragmentation.''

	The relation between constrained one-dimensional spin chains and semigroups can be summarized as follows (see Sec.~\ref{sec:general_semigroup} for a more formal discussion). Most examples of constrained systems in the literature---and all of the examples considered in this work---have constraints which can be formulated in a local product state basis. In these systems, the constraints place certain rules on the processes that the dynamics is allowed to implement, and it is these rules which endow the dynamics with the structure of a semigroup. This identification works by viewing each basis state of the onsite Hilbert space, in the computational basis, as a generator of a semigroup. Since any element of the semigroup can be written as a product of generators, a many-body computational-basis product state is naturally associated with an element of the semigroup, obtained by taking the product of generators from left to right along the chain. We call each computational basis state a {\it word}, with each word being a presentation of a certain element in the semigroup. In this picture, the constraints are encoded by requiring that the dynamics {\it preserve the semigroup element associated with each product state}. In Sec.~\ref{sec:general_semigroup}, we show that all local dynamical constraints can be formulated in this way.  
	
	Of course, not all words represent distinct semigroup elements. For example, in the case where the semigroup is a group $G$ with identity element $\tte$ and elements $\ttg_1, \ttg_2$, three distinct words of length $4$ representing the same group element are $\tte\tte\tte\tte$, $\ttg_1 \ttg_1^{-1} \ttg_2^{-1} \ttg_2$, and $\ttg_1 \ttg_2 \ttg_2^{-1} \ttg_1^{-1}$.
	Each distinct semigroup element is thus an {\it equivalence class} of words under the application of equivalence relations like $\ttg_i \ttg_i^{-1} = \tte \tte$. The most general local dynamics that preserves semigroup elements is precisely one which locally implements these equivalence relations. %
	
	The equivalence classes so defined produce multiple sectors of the Hilbert space (``Krylov sectors'' or ``fragments'') that the dynamical rules are unable to connect, breaking ergodicity and leading to Hilbert space fragmentation~\cite{sala2020ergodicity, khemani2020localization,motrunich}. 
	Within a given fragment, thermalization of an initial basis state is a process by which the dynamics of the system ``works out'' which words represent the same semigroup element as the initial state.
	Crucially, when the problem of determining which words represent a given element is computationally hard, thermalization {\it within} each fragment is slow. 
	
	This general perspective is powerful because it maps the problem of thermalization in these models onto a well-known algorithmic problem, the word problem for semigroups (a perspective also adopted by Hastings and Freedman in Ref.~\cite{hastings} to provide examples of dynamics exhibiting `topological obstructions' which provide a separation between the performances of QMC and quantum annealing). The word problem is the problem referenced above, namely that of deciding whether two words represent the same semigroup element.
	This identification allows us to lift examples of computationally hard word problems from the mathematical literature to construct models with anomalously slow dynamics. In these models, the dynamics connects the basis states within each fragment in a manner which is much sparser than in generic systems (see Fig.~\ref{fig:connectivity_comp} for an illustration), and it is this phenomenon which leads to long thermalization times.

	The first part of this work focuses on models where the word problem takes an exponentially long time to solve. 
	We place particular focus on the ``Baumslag-Solitar model'', a spin-2 model with three-site interactions for which the relaxation time of a large class of initial states under any type of local dynamics (Hamiltonan, random unitary, etc.) is provably {\it exponentially long} in the system size.
	This model has a conserved charge, and this exponentially long timescale shows up as an exponentially slow hydrodynamic relaxation of density gradients. Not only the relaxation timescale but also the functional form of relaxation is anomalous: a state with density gradients relaxes ``gradually, then suddenly,'' with an initially prepared density wave experiencing almost no relaxation for exponentially long times, but then undergoing a sudden collapse at a sharply defined timescale. 
	Despite this extremely slow hydrodynamics, the states involved are {\it not} dynamically frozen: each configuration is rapidly locally fluctuating, and generic local autocorrelation functions decay rapidly.
	
	In the second part of this work, we turn to word problems that have high {\it spatial} complexity: i.e., solving them requires not just many steps, but also a large amount of additional space for scratchwork. Put differently, in these problems, deciding whether two words of length $L$ are equivalent requires a derivation involving intermediate words much longer than $L$. 
	
	In the corresponding dynamical systems, one has (for any fixed $L$) pairs of basis states $\ket{a}, \ket{b}$ such that: (i)~$\ket{a}$ is not connected to $\ket{b}$ by the dynamics of a chain of length $L$, but (ii)~$\ket{a} \otimes \ket{c}$ is connected to $\ket{b} \otimes \ket{c}$ by the dynamics of a longer chain, with $\k c$ a fixed ancillary product state. 
	This phenomenon can be viewed in two complementary ways: as a ``fragile'' form of Hilbert space fragmentation,\footnote{In analogy with fragile topology \cite{po2018fragile}, as the ergodicity of the dynamics is modified by the addition of trivial ancillae.} or as an unusual type of jamming which has no counterpart in known examples of jammed systems. Using constructions similar to the group model discussed above, we can construct examples where the amount of additional spatial resources grows extremely rapidly with $L$; not just exponentially but also as $e^{e^L}, e^{e^{e^L}}$, and so on.
	
	This paper is organized as follows. In Sec.~\ref{sec:general_semigroup} we introduce word problems for semigroups and groups, and relate them to fragmentation. In Sec.~\ref{sec:therm_time_bounds} we use the complexity of the word problem to derive bounds on thermalization times, and in Sec.~\ref{sec:bs} we explore an explicit example of a high-complexity group word problem which yields dynamics with exponentially slow relaxation. We present numerical evidence and analytical estimates for the anomalously slow hydrodynamics of this model. In Sec.~\ref{sec:ff} we introduce and analyze a family of group models with fragile fragmentation. Secs.~\ref{sec:nongroup} and \ref{sec:loop_models} respectively present examples based on semigroups, and generalize our one-dimensional examples to two-dimensional loop models. Finally, we conclude with a discussion of future directions in Sec.~\ref{sec:disc}.   
	
	\begin{figure}
		\centering
		\includegraphics[width=.49\tw]{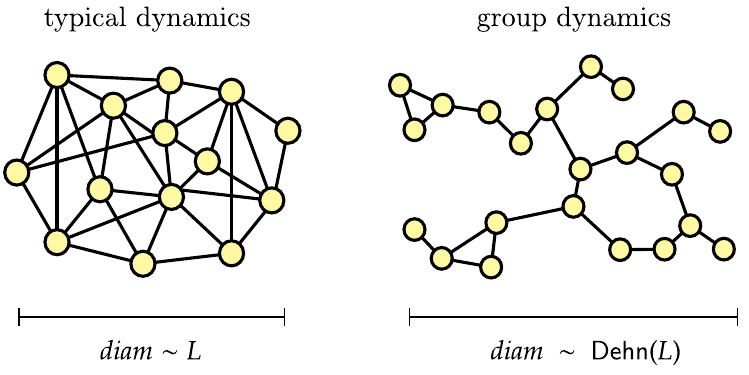}
		\caption{A schematic illustration of the difference between the Hilbert space connectivity of generic dynamics and semigroup dynamics, with each yellow dot representing a computational basis product state in a single connected sector of the dynamics. In a 1D system of length $L$, typical dynamics (left) requires at most $O(L)$ steps of the dynamics (applications of a Hamiltonian or layers of a unitary circuit) to move between any two basis states. In semigroup dynamics (right), the number of steps needed scales as the {\it Dehn function} $\dehn(L)$ of the semigroup in question, which measures the word problem's temporal complexity. When $\dehn(L)$ is large, the basis states in each sector are connected very sparsely, leading to long thermalization times.    }
		\label{fig:connectivity_comp}
	\end{figure}

	\section{Semigroup dynamics and constrained 1D systems} \label{sec:general_semigroup} 
	
	In this section we introduce the general framework used to construct the models described above. We will refer to this framework as {\it semigroup dynamics}, which encompasses a general class of constrained dynamical systems whose constraints can be derived from the presentation of underlying group or semigroup (to be defined below). These types of constraints are particularly appealing from a theoretical point of view: it turns out that we can rigorously characterize many properties---thermalization times, fragmentation, and so on---using tools from the field of {\it geometric group theory}.  Broadly speaking, geometric group theory is concerned with characterizing the complexity and geometry of discrete groups (see \cite{clay2017office} for an accessible introduction), ideas which will be made precise in the following. 
	
	As a starting point, we will describe the necessary mathematical background needed to motivate group dynamics. This discussion will center around the {\it word problem}, a century-old problem lying at a heart of results regarding the geometry and complexity of groups. We will then see how algorithms solving the word problem can naturally be encoded into the dynamics of 1D spin chains, whose thermalization dynamics is controlled by the word problem's complexity. Finally, we will see how the structure of the word problem leads to Hilbert space fragmentation, and discuss the properties of the group which control the severity of the fragmentation. 
	
	Throughout this paper, we will mostly be studying constrained dynamics on 1D spin chains whose onsite Hilbert space is finite dimensional.\footnote{Generalizations to 2D are briefly discussed in Sec.~\ref{sec:loop_models}.}  We will only consider systems whose time evolution has constraints that can be specified in a local tensor product basis (referred to throughout as {\it the computational basis}), either directly or after the application of finite-depth local unitary circuit. In the latter case we will assume the unitary transformation has been done, to avoid loss of generality.

	We will let $S$ denote the set of computational basis state labels for the onsite Hilbert space, with individual basis states being written as letters in typewriter font ($\k\tta, \k\ttb$, etc.):
	\be \mch_{\rm loc} = {\rm span}\{ \k \tta \, : \, \tta \in S\}.\ee 
	Strings of letters will be used as shorthand for tensor products, so that e.g. $\k{\tt{word}} = \k\ttw\tp\k\tto\tp\k\ttr\tp\k\ttd$. A ket with a single roman letter will denote a product state of arbitrary length, e.g. $\k{w} = \k{\tt{word}}$. 
	
	\ss{Dynamical constraints and semigroups}
	
	We will write $\dyn(t)$ to denote time evolution for time $t$ under the dynamics in question, 
	which may be performed using a set of unitary gates, a (possibly space- or time-dependent) Hamiltonian, or a bistochastic Markov chain.
	Having a dynamical constraint means that not all computational basis product states $\k w, \k{w'}$ can be connected under the dynamics. We will use $\vp(\k w)$ to denote the {\it dynamical sector} of the state $\k w$, defined as the set of all computational basis product states that $\k w$ can evolve to, thus 
	\be \label{dynconstraint}\lan w | \dyn(t) | w'\ran \propto \d_{\vp(\k w), \vp(\k{w'})}\, \, \forall \, t.\ee 
	
	The tensor product of computational basis states---that is, the stacking of one system onto the end of another---defines a binary operation on the dynamical sector labels, which we write as $\odot$: 
	\be \vp(\k w \tp \k{w'}) = \vp(\k w) \odot \vp(\k{w'}).\ee 
	Since the tensor product is associative, so too is $\odot$. This equips the set of dynamical sectors with an associative binary operation, thereby endowing it with the structure of a {\it semigroup}, a generalization of a group which needn't have inverses or an identity element.\footnote{We focus on the more general structure of a semigroup (as opposed to a group), since it provides the most general framework for discussing 1D constrained dynamics, does not complicate the bounds we will obtain on thermalization dynamics, and because most of the existing models studied in the literature are described by semigroups, rather than groups. } Since all $\k w$ are formed as tensor products of the single-site basis states $\k \tta$, the $\vp(\k{\tta})$---which we will write simply as $\tta$ to save space---constitute a {\it generating set} for this semigroup. The dynamical sector of a state $\k w = \k{\tta_1} \tp \cdots \tp \k{\tta_L}$ is then determined simply by multiplying the $\tta_i$ along the length of the chain\footnote{We will mostly be interested in chains with open boundary conditions; for PBC one may arbitrarily choose a site to serve as the ``start'' of the word.}: \be \vp(\k w) = \tta_1 \odot \cdots \odot \tta_L.\ee
	
	Following usual group theory notation, we will denote a semigroup with generating set $S$ as 
	\be G = {\sf semi} \lan S \, | \, R \ran,\ee 
	where $R$ denotes a set of {\it relations} imposed on the product states that can be formed from elements of the on-site computational basis states $S$. When writing presentations of groups we will omit the inverse generators and the identity from $S$, and will likewise omit trivial relations like $\tta\tta\inv = \tte, \tta\tte=\tte\tta$ from $R$. For the remainder of the paper, presentations of semigroups which are not groups will always be denoted by ${\sf semi}\lan S \, | R\ran$, while presentations of groups will be denoted simply by $\lan S \, | \, R\ran$. 
	
	The relations in $R$ are determined by $\vp$, namely by which product states are related to one another under $\dyn$. Consider any two states $\k{w}$, $\k{w'}$ such that $\vp(\k w) = \vp(\k{w'})$ define the same element of $G$. 
	Since we are interested in dynamics which are geometrically local, it must be possible to relate $\k w$ to $\k{w'}$ using a series of {\it local} updates to $\k w$. This means that the set $R$ must be expressible in terms of a set of equivalence relations which each involve only an $O(1)$ number of the elements of $S$---implying in particular that $|R|$ must be finite. A semigroup where both $S$ and $R$ are finite is said to be {\it finitely presented}, and all of the semigroups we consider will be of this type.

	Given a semigroup $G$, we will use the notation $\dyn_G$ to represent a general local dynamical process acting on $\mch = \mch_{\rm loc}^{\tp L}$ which satisfies the constraint \eqref{dynconstraint}, and hence preserves the dynamical sectors of all computational basis states. 	
	The locality of the dynamics means that $\dyn_G$ must be composed of elementary blocks (unitary gates or Hamiltonian terms) of constant length which, when acting on computational basis product states, implement the relations contained in $R$. 	
	Writing the relations in $R$ as $r_{l} = r_{r}$ with $r_{l/r} = \tta_{l/r,1} \cdots \tta_{l/r,n}$,
	the locality of $\dyn_G$ is determined by the maximal length of the $r_{l/r}$, which we denote as $\ell_R$: 
	\be \ell_R = \max_{r_{l/r}\in R}|r_{l/r}|.\ee 
	For us $\ell_R$ will always be $O(1)$ (and when $G$ is a group, one can show that there always exists a finite presentation of $G$ such that $\ell_R\leq 3$; see App.~\ref{app:group_props} for the proof).
	
	To be more explicit, first consider the case where $\dyn_G$ corresponds to time evolution under a (geometrically) local Hamiltonian $H$. The semigroup constraint and locality of $H$ means that $\lan w ' | H |w \ran$ can be nonzero only if the words $w,w'$ differ by the {\it local} application of a relation in $R$.  $H$ consequently assumes the general form 
	\be H = \sum_i \sum_{r \in R} (\l_{i,r} \kb{r_l}{r_r}_i + h.c.),\ee 
	where $\k{r_{l/r}}_i = \k{\tta_{l/r,1}}_i \tp \cdots \tp \k{\tta_{l/r,n}}_{i+n}$ and the $\l_{i,r}$ are arbitrary complex numbers. Note that the above Hamiltonian is only well-defined if $|r_l| = |r_r|$ for all relations $r$. In cases where this does not hold, we will rectify this problem by adding a trivial character $\texttt{e}$ to the onsite Hilbert space---with $\tte$ defined to commute with all of the other generators $\tta$---which allows us to then `pad' the relations $r$ in a way which ensures that $|r_l| = |r_r|$.
	
	As a simple example, consider the group $\zz^2 = \lan \ttx, \tty | \, \ttx\tty = \tty\ttx\ran$, which as we will see later in some sense has trivial dynamics. Since the full generating set $S = \{ \ttx, \ttx\inv, \tty,\tty\inv,\tte\}$ of this presentation has dimension $|S|=5$, a Hamiltonian $H_{\zz^2}$ with $\zz^2$-constrained dynamics thus acts most naturally on a spin-2 chain. The single nontrivial relation $\ttx\tty=\tty\ttx$ has length $\ell_R = 2$, and thus $H_{\zz^2}$ can be taken to be 2-local, assuming a form like 
	\bea H_{\zz^2} & = \sum_i \sum_{\tta \in S} \big( \l_{1,i} \kb{\tta\tta\inv}{\tte\tte}_{i,i+1} + \l_{2,i} \kb{\tta\inv\tta}{\tte\tte}_{i,i+1} \\ 
	& + \l_{3,i} \kb{\tta\tte}{\tte\tta} + h.c.\big) + \sum_i (\l_{4,i} \kb{\ttx\tty}{\tty\ttx} + h.c.) + \cdots,\eea 
	which describes two species of conserved particles, each of which defines a $U(1)$ conserved charge $n_{\texttt{a}} = \sum_{i} \ketbra{\texttt{a}}{\texttt{a}} - \ketbra{\texttt{a}^{-1}}{\texttt{a}^{-1}}$ where $\texttt{a} = \texttt{x}, \texttt{y}$. The explicit examples we consider in this work will not be more complicated than $H_{\zz^2}$ in terms of their degree of locality or the dimension of $\mch_{\rm loc}$,  but their dynamical properties will be much richer. 
	
	\begin{figure*}
		\includegraphics[width=1\tw]{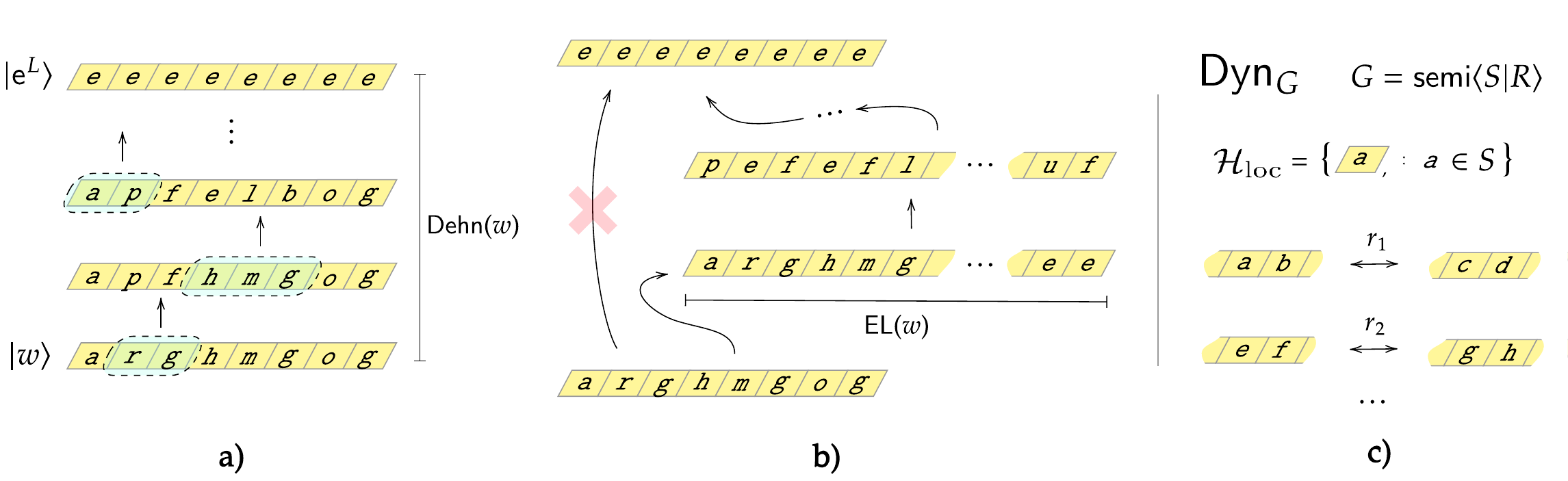}
		\caption{\label{fig:arghmgog} 
			Semigroup dynamics and the word problem. {\bf a}) In the word problem, we are given a length-$L$ word $\k{w}$ and a series of re-writing rules that let us make local updates to the characters of the word. The word problem for $\k{w}$ is the task of deciding whether or not a sequence of allowed updates can be found which transforms $\k{w}$ into a particular reference word, here chosen to be $\k{\tte^L}$. The {\it Dehn function} $\dehn(w)$ measures the time complexity of the word problem, namely the minimal number of updates needed to connect $\k{w}$ to $\k{\tte^L}$. {\bf b)} In some situations, $\k{w}$ can only be transformed into $\k{\tte^L}$ by increasing the amount of available space, done here by appending $\k{\tte\tte\cdots}$ onto the original word. The minimal amount of extra space needed defines the {\it expansion length} $\el(w)$, which captures the spatial complexity of the word problem. {\bf c)} For a semigroup $G = \lan S | R\ran$ defined by generators $\tta \in S$ and relations between the generators $r_i \in R$, our construction defines a dynamics $\dyn_G$ which acts on a 1D chain with an onsite Hilbert space $\mch_{\rm loc} = \{ \k{\tta}\, : \, \tta \in S\}$. The dynamics implimented by $\dyn_G$ is restricted to local updates which preserve the semigroup element obtained by multiplying the generators along the length of the chain; the allowed updates are consequently fixed by the relations in $R$ (with the figure drawn using the relations $r_1 : \tta\ttb = \ttc\ttd$, $r_2 : \tte\ttf = \ttg\tth$, etc.).  }
	\end{figure*}

	The construction of group-constrained random unitary dynamics is similar to the  Hamiltonian case. For random unitary dynamics, $\dyn_G$ is constructed using $\ell_R$-site unitary gates $U$ whose matrix elements $\lan w' |U|w\ran$ are nonzero only if the length-$\ell_R$ words $w,w'$ satisfy $\vp(\k w) = \vp(\k{w'})$. Such unitaries admit the decomposition 
	\begin{equation} \label{groupru}
		U = \bigoplus_{g \in G_{\ell_R}} U_g, \hspace{0.3cm} G_{\ell_R} = \{g \, | \, \exists \, \k w \, : \, \vp(\k w) = g, \,  |w| \leq \ell_R\},
	\end{equation}
	where $G_{\ell_R}$ denotes the set of elements of $G$ expressible as products of precisely $\ell_R$ generators. A particularly natural realization of $\dyn_G$ is when each $U_g$ is drawn from an appropriate-dimensional Haar ensemble, although we shall not need to specify to this case.
	
	Existing examples of constrained 1D dynamics in the literature---from multipole conserving systems to models based on cellular automata and other constrained classical systems \cite{sala2020ergodicity,khemani2020localization,motrunich,pancotti2020quantum,van2015dynamics,brighi2023hilbert,lan2018quantum,garrahan2018aspects}---are all described by $\dyn_G$ for an appropriate semigroup $G$ and an appropriate kind of dynamics (random unitary, Hamiltonian, etc.).\footnote{This is not true for models with {\it quantum} Hilbert space fragmentation \cite{motrunich,brighi2023hilbert}, which violate the above assumptions by virtue of having constraints that cannot be formulated in a local product basis.} 
	One of the main messages of this paper is that in addition to providing an organizing framework for discussing 1D constrained dynamics, approaching things from this point of view enables a large arsenal of mathematical tools from the field of geometric group theory to be brought to bear, leading to general bounds on thermalization times, precise characterizations of ergodicity breaking, and explicit examples of models with extremely slow dynamics. 
	
	
	\subsection{The word problem} \label{ssec:word_problem}
	
	We will now formulate the {\it semigroup word problem}, a concept key for determining the thermalization beahvior of models with semigroup dynamics. 
	We will refer to a computational basis product state---defined by a string of generators in $S$, e.g. $w = \tta_1 \cdots \tta_L$---as a {\it word}. In what follows we will often slightly abuse notation by letting the symbol $w$ stand for both an abstract string $\tta_1\cdots \tta_L$ of generators in $S$, as well as the associated computational-basis product state $\k{\tta_1\cdots \tta_L}$. 
	
	Words are naturally grouped into equivalence classes labeled by elements of $G$. Letting $W(S)$ denote the set of all words, we define these equivalence classes as 
	\be \label{kgdef} K_g \triangleq \{ w \in W(S) \, | \, \vp(\k w) = g\}.\ee 
	Any two words belonging to the same equivalence class can be deformed into one another by applying a sequence of relations in $R$. For any two $w, w' \in K_g$, we define a {\it derivation from $w$ to $w'$}, written $D(w \sqig w')$, as the sequence of words appearing in this deformation: 
	\begin{equation}
		D(w\sqig w') = w\to u_1 \to u_2 \to \cdots \to u_n \to w',
	\end{equation}
	where each arrow $\to$ indicates applying a {\it single} relation from $R$.

	In Sec.~\ref{sec:therm_time_bounds} we will see that the way in which $\dyn_G$ thermalizes is determined by the complexity of the {\it word problem} for $G$, a fundamental problem in the fields of abstract algebra and computability theory.  The word problem is defined by the following question:
	
	\vspace{0.25cm}
	\fbox{\begin{minipage}{22.5em}
			\textbf{Word problem: }Given two words $w,w',$
			does $\vp(\k w) = \vp(\k{w'})$? That is, does there exist a derivation $D(w \sqig w')$?
	\end{minipage}}
	\vspace{0.25cm}

	A key result is that even in the case where both $|S|$ and $|R|$ are small, answering this question can be very difficult (even undecidably so; see App.~\ref{app:undecide}). 
	For semigroups or groups which do not have an undecidable word problem, a key problem is to determine the time and space complexity of algorithms that solve it. 
	In what follows, we will introduce two functions characterizing the word problem's complexity: the {\it Dehn function}, which governs its time complexity, and the {\it expansion length}, which governs its space complexity.

	\sss{Time complexity}
	Our operational definition of the time complexity of the word problem is the minimum length of a derivation linking $w$ to $w'$, as illustrated in Fig.~\ref{fig:arghmgog}~(a).\footnote{An important caveat here is that proving $\vp(\k w) = \vp(\k{w'})$ does not need to require outputing an explicit derivation $D(w_1\sqig w_2)$, as for a {\it particular} $G$ there may exist other (faster) algorithms for constructing a proof. We therefore must stress that when we speak of the time complexity of the word problem, we are referring to its time complexity within the so-called {\it Dehn proof system}, a logical system where the {\it only} manipulations one is allowed to perform are the implementation of local relations in $R$. If one wants an algorithm which works for all choices of $G$, operating within this system is generically the best one can do. } We denote this by $\dehn(w_1, w_2)$:
	\begin{equation}\label{distdef}
		\dehn(w_1, w_2) \triangleq \min_{D} |D(w \rightsquigarrow w')|.
	\end{equation}
	$\dehn(w,w')$ measures the non-deterministic time complexity of the word problem since it is the maximum runtime of an algorithm which maps $w$ to $w'$ by blindly applying all possible relations in $R$ to $w$ in parallel and halts the first time $w'$ appears in the resulting superposition of words. Time evolution under $\dyn_G$ can be naturally regarded as a way of simulating this process, a connection enabling the derivation of the bounds arrived at in Sec.~\ref{sec:therm_time_bounds}. 
	
	For any word $w$, we define the {\it length} $|w|$ of $w$ as the number of generators appearing in $w$. To denote the subset of length-$L$ words in $K_g$, we write
	\be \label{kgldef} K_{g,L} \triangleq \{ w \in K_g \, | \, |w| = L\}\ee 
	as the set of length-$L$ words in $K_g$ (or equivalently, following our practice of letting $w$ stand for both a word and a computational-basis product state, as the collection of product states associated to such words). The (worst-case) time complexity across all words in $K_{g,L}$ defines the function 
	\be \label{dehng} \dehn_g(L) \triangleq \max_{\k w, \k{w'} \in K_{g, L}} \dehn(w,w').\ee 
	We will be particularly interested in how $\dehn_g(L)$ scales asymptotically with $L$. In App.~\ref{app:primer}, we show that this scaling is the same for any two finite presentations of the same semigroup: thus we may meaningfully speak about the Dehn function of a {\it semigroup}, rather than a particular presentation thereof.  This presentation independence imparts some degree of the robustness to the dynamical properties that we will derive. 
	
	For the case where $G$ is a {\it group} (rather than just a semigroup), a fair amount more can be said. All groups have a distinguished identity element $e$, and in App.~\ref{app:primer} we show that\footnote{Throughout this work, $\sim,\gtrsim,\lesssim$ will be used to denote (in)equality in the scaling sense. For example, when we write $f(L) \sim g(L)$, we mean that there exists a constant $C>0$ such that $f(L/C) \leq g(L) \leq f(LC)$.}
	\be \dehn_e(L) \gtrsim \dehn_g(L) \, \, \forall \, g \in G.\ee 
	We furthermore show that, as long as $|K_{g,L}|$ scales exponentially in $L$, then $\dehn_e(L) \sim \dehn_g(L)$. For groups, we thus define 
	\be \dehn(L) \triangleq \dehn_e(L)\ee 
	as a simpler characterization of the word problem's time complexity. The calculation of $\dehn_e(L)$ also simplifies further for groups, as we may fix $\k{w'} = \k{\tte^L}$ in \eqref{dehng} without changing the asymptotic scaling of $\dehn(L)$. Thus for groups, we will mostly focus on computing 
	\be \label{groupdehn} \dehn(L) = \max_{\k w \in K_{e, L}} \dehn(w,\tte^L) \qq \text{($G$ a group)}\ee 
	
	Even for groups where $|S|, |R|$ are both $O(1)$, $\dehn(L)$ can scale in many different ways. To start, it is easy to verify that $\dehn(L)\sim L$ for all finite groups, and that 
	\be G \,\, {\rm Abelian} \implies  \dehn(L) \lesssim L^2,\ee 
	with $\dehn(L) \sim L^2$ only when $G$ is infinite. Indeed, for all such groups, such as the $G = \zz^2$ example above, $\dehn(L)$ is bounded from above by the time it takes to transport the conserved charges $n_\tta \triangleq \sum_i (\proj\tta_i - \proj{\tta\inv}_i)$ across the system, which is $\sim L^2$. For our purposes, we will regard any $G$ with $\dehn(L) \lesssim L^2$ as uninteresting, as for these groups $\dyn_G$ thermalizes on timescales generically no slower than for conventional systems with conserved $U(1)$ charges. 
	
	A simple example of a group with an ``interesting'' Dehn function is the discrete Heisenberg group $\sfH_3$, which has the group presentation 
	\be \sfH_3 = \lan \ttx,\tty,\ttz\, | \, \ttx\tty=\tty\ttx\ttz, \, \ttx\ttz=\ttz\ttx,\, \tty\ttz=\ttz\tty \ran,\ee 
	and possesses a Dehn function scaling as $\dehn(L) \sim L^3$ \cite{gersten2003isoperimetric}. In Sec.~\ref{sec:bs} we will see an example of a simple group where $\dehn(L) \sim 2^L$ and in Sec.~\ref{sec:ff} one with $\dehn(L) \sim 2^{2^L}$; examples of semigroups with similarly slow dynamics are given in Sec.~\ref{sec:nongroup}. Going beyond these examples, Refs.~\cite{brady2000there,sapir2002isoperimetric} remarkably show that for any constant $\a$, almost any function with growth $L^2 \leq  f(L) < L^\a$ is the Dehn function of some finitely presented semigroup; this includes for example ``unreasonable-looking'' functions such as $L^{7\pi}, L^2 \log L,$ and $L^e\log(L)^{\log(L)}\log\log L$. 
	
	In addition to the worst case complexity of the word problem, we will also have occasion to consider its {\it average-case} complexity (both for groups and semigroups), a quantity has recieved much less attention in the math literature (the only exception the authors are aware of is \cite{young2008averaged}). To this end, we define the {\it typical Dehn function} $\ob{\dehn}_g(L)$ as the number of steps needed to map a certain constant fraction of words in $K_{g, L}$ to one another: 
	\be \label{typdehn} \ob{\dehn}_g(L) \triangleq \sup \{ t \, : \, {\rm Pr}_{w,w' \in K_{g, L}} [\dehn(w,w')  \geq t] \geq 1/2\},\ee 
	Establishing rigorous results about $\ob{\dehn}$ is unfortunately much more difficult than for $\dehn$, and the landscape of typical Dehn functions is comparatively less well explored. For the examples we focus on in this paper however, a combination of physical arguments and numerics will nevertheless suffice to understand the rough asymptotic scaling of $\ob{\dehn_g}(L)$.

	\sss{Space complexity} 
	
	Another complexity measure is the (nondeterministic) {\it space complexity} of the word problem.\footnote{The same caveat we made for the time complexity also applies here---we are referring explicitly to the space complexity within the ``Dehn proof system''.} The space complexity is nontrivial in cases where, during the course of being deformed into $w'$, a word $w$ must expand to a length $L' > L$. More formally, we define the {\it expansion length}\footnote{The expansion length function is called the {\it filling length function} in the math literature.} of a derivation $D(w\sqig w') = w \ra u_1 \ra \cdots \ra u_n \ra w'$ as the maximal length of the intermediate words $u_i$:
	\be \el(D(w \sqig w')) \triangleq  \max_{u_i \in D(w\sqig w')} |u_i|.\ee 
	The relative expansion length $\el(w,w')$ between two {\it words} is then defined as the minimal expansion length of a derivation connecting them: 
	\be \el(w,w') \triangleq \min_{D(w \sqig w')} \el(D(w\sqig w')),\ee 
	as illustrated in Fig.~\ref{fig:arghmgog}~(b). 
	The expansion length of a $K_{g, L}$ sector is likewise 
	\be \label{explengthg} \el_g(L) \triangleq \max_{w,w'\in K_{g, L}} \el(w,w').\ee 
	Similarly with the Dehn function, we show in App.~\ref{app:primer} that when $G$ is a group, $\el_g(L) \lesssim \el_e(L)$ for all $g$, so that for groups we may use 
	\be  \label{elgroup} \el(L) \triangleq \el_e(L) \qquad \text{($G$ a group)}\ee 
	as a simple metric of the space complexity. 
	
	In App.~\ref{app:primer} we show that $\el(L) \lesssim L$ for all Abelian groups, and that all finite groups have $\el(L) \leq L+C$ for some constant $C$; such scalings are ``uninteresting'' from the perspective of space complexity. Just as with the Dehn function though, there exist simple semigroups for which $\el(L)$ grows extremely fast with $L$, the consequences of which will be explored in Sec.~\ref{sec:ff}.

	\subsection{Semigroup dynamics and Hilbert space fragmentation} 
	
	The existence of multiple $K_{g, L}$ sectors means that $\dyn_G$ is not ergodic as long as $G$ is not a presentation of the trivial semigroup. 
	The simplest case is when $G$ is an Abelian group. In this case, the $K_{g, L}$ can be associated with the symmetry sectors of a global symmetry. This is true simply because the sector that a given product state lives in can be determined by computing the expectation value of the operators $n_{\ttg} = \sum_i (\proj{\ttg}_i - \proj{\ttg\inv}_i])$ for each generator $\ttg$.  Thus $\dyn_G$ for Abelian $G$ is already very well understood, given the plethora of work on thermalization in the presence of global symmetries.

	When $G$ is not an Abelian group, global symmetries may still be present, but there inevitably exist other non-local conserved quantities which distinguish different dynamical sectors. 
	Indeed, in App.~\ref{app:nosyms} we prove that the dynamical sectors of such models {\it never} be described by global symmetries alone.     
	Since $\dyn_G$ thus always has non-local conserved quantities if $G$ is not an Abelian group, the lack of ergodicity due to these quantities leads to {\it Hilbert space fragmentation} (HSF) \cite{khemani2020localization,sala2020ergodicity,moudgalya2022quantum,motrunich, MOUDGALYA2023169384}, a phenomenon whereby the dynamics of initial states becomes trapped in disconnected subspaces of $\mch$---in our notation simply the $K_{g, L}$---whose existence cannot be attributed to the presence of global symmetries alone. 
	
	The original works on HSF~\cite{khemani2020localization,sala2020ergodicity} focused mainly on fragmentation in spin systems with conserved dipole moments. While these systems fall within our semigroup framework,\footnote{Concretely, the strongly-fragmented models of~\cite{khemani2020localization,sala2020ergodicity} may be phrased as semigroup dynamics for the ``dipole semigroup'' ${\sf Dip}$ defined by the semigroup presentation 
		\be {\sf Dip} = {\sf semi}\lan \tt0, \tt{1}, \tt2 \, | \, R\,\ran, \ee
		where $R = \{ \tt0\tt2\tt0 = \tt1\tt0\tt1, \, \tt1\tt2\tt1 = \tt2\tt0\tt2, \, \tt1\tt2\tt0 = \tt2\tt0\tt1, \, \tt0\tt2\tt1 = \tt1\tt0\tt2\}$. 	
	} we will instead find it more instructive to review the pair-flip model introduced in \cite{caha2018pair}.  The pair-flip model is described by a spin-$s$ Hamiltonian of the following form:
	\be  \label{hpf} H_{PF} = \sum_i \sum_{a,b=-s}^s \l_{ab,i} (\kb{aa}{bb}_{i,i+1} + h.c.).\ee 
	For generic choices of $\l_{ab,i}$, the model is non-integrable, so any ergodicity breaking cannot be associated with integrability.
	Note that the product states $\k{w}$ where $w=a_1 \cdots a_L$ are annihilated by $H_{PF}$ if $a_i \neq a_{i+1}$ for all $i$. These product states alone provide $(2s+1)(2s)^{L-1}$ dynamically-disconnected dimension-1 sectors not attributable to any local symmetry, meaning that $H_{PF}$ exhibits HSF. 
	
	The dynamics generated by $H_{PF}$ is in fact a special case of our construction applied to the group 
	\be \zt^{\ast (2s+1)} = \lan \ttg_1 ,\dots, \ttg_{2s+1} \, | \, \ttg_i^2 = \tte \ran,\ee 
	where $\ast$ denotes the free product. $H_{PF}$ can be obtained from our general construction by considering a modified onsite Hilbert space $\mch' = {\rm span}\{ \k{\ttg}, \, \k{\ttg\inv}\, : \, \ttg \in S \}$ which contains no $\tte$ generator. The relation $\ttg_i^2 = \tte$ can be rewritten without $\tte$ as $\ttg_i^2 = \ttg_j^2$ for all $i,j$, and a general local Hamiltonian acting on $(\mch')^{\tp L}$ which obeys the $ \zt^{\ast (2s+1)}$ may accordingly be written as 
	\be H = \sum_i  \sum_{\ttg,\tth \in S} ( \l_{\ttg\ttg,i} \kb{\ttg\ttg}{\tth\tth}_{i,i+1} + h.c.) ,\ee 
	which matches the Hamiltonian in \eqref{hpf}. Unfortunately the Dehn function of $\zt^{\ast (2s+1)}$ is easily seen to scale as $\dehn(L) \sim L^2$, and thus does not provide a complexity scaling that is unusual enough to warrant further study.

	Having discussed the concept of HSF, let us further understand the structure and size of disconnected sectors under the dynamics $\dyn_G$.  The number of such disconnected sectors {\it at least} as large as the number of subspaces, each labelled by $K_{g, L}$ for some $g \in G$. Thus, the number of such subspaces---which we denote as $N_K(L)$---depends on the number of semigroup elements expressible as words of length $\leq L$. In particular, we define the {\it geodesic length} $|g|$ of an element $g\in G$ as the length of the shortest word representing $g$: 
	\be |g| = \min_{w \in K_g} |w|.\ee 
	A word $w$ with $\vp(\k w) = g$ satisfying $|w| = |g|$ is called a {\it geodesic word}.  Then 
	\be \label{groupgrowth} N_K(L) = | \{ g\, | \, |g| \leq L\} |.\ee 
	Semigroups with more elements ``close'' to the identity thus have dynamics with a larger number of Hilbert space fragments. As an example, one may readily verify that in the pair flip model, $N_K(L) \sim (2s)^L$ grows exponentially as long as $s>1/2$. 	
	
	For many models with group constraints, including the pair-flip model, $N_K(L)$ exactly determines the number of Hilbert space fragments.  However, for some semigroups, $N_K(L)$ is not the full story, with each $K_{g, L}$ further fragmenting into subspaces in a way controlled by the expansion length function defined in~\eqref{explengthg}. Understanding this phenomenon is the subject of Sec.~\ref{sec:ff}. 
	
	In the study of HSF, a distinction is often made between ``strong'' and ``weak'' HSF. This distinction was originally discussed in the context of Hamiltonian dynamics~\cite{sala2020ergodicity}, where it was defined by violations of strong and weak ETH, respectively. Since we are discussing things at a level where the nature of $\dyn_G$ may or may not involve eigenstates, we will instead adopt a slightly different definition in terms of the size of the largest $K_{g, L}$ sector, which we denote as $K_{{\rm max},L}$ (in App.~\ref{app:primer} we prove that for groups, the largest sector is in fact always the one associated with the identity, $K_{{\rm max},L} = K_{e, L}$).  
	We will say that the dynamics is
	\begin{enumerate}
		\item {\it weakly} fragmented if $|K_{{\rm max},L}|/|\mathcal{H}| \to O(1)$ as $L\ra\infty$, and 
		\item {\it strongly} fragmented if $|K_{{\rm max},L}|/|\mathcal{H}| \to 0$ as $L\ra\infty$.
	\end{enumerate}
	We will find it useful to subdivide the strongly fragmented case into additional classes according to how quickly $|K_{{\rm max}, L}|/|\mathcal{H}|$ vanishes as $L\ra\infty$. We will say that the fragmentation in the case of strong HSF is 
	\begin{enumerate}
		\item {\it polynomially strong} if $|K_{{\rm max},L}|/|\mathcal{H}| \sim 1/{\rm poly}(L)$, and 
		\item {\it exponentially strong} if $|K_{{\rm max},L}|/|\mathcal{H}| \sim \exp(-L)$.
	\end{enumerate}
	
	Note that the above definitions are made without reference to any global symmetry sectors. This means that if global symmetries happen to be present, the quantum numbers associated with them will constitute part of the elements $g$ labeling the different dynamical sectors, and the above definition of strong/weak HSF will simply single out the largest, regardless of its symmetry quantum number(s). When symmetries are present one could also quantify the degree of fragmentation by first fixing a quantum number, changing $K_{\rm max}(L)$ to be the largest sector having that quantum number, and replacing $|\mch|$ by the dimension of the chosen symmetry sector. However, since generic $\dyn_G$ dynamics needn't have any global symmetries, and since the result of the above procedure can depend sensitively on the chosen quantum number \cite{morningstar2020kinetically,pozderac2023exact,wang2023freezing}, we will focus only on the above (simpler) definition, which maximizes over all symmetry sectors.

	Our models provide a way of addressing two questions raised in Ref.~\cite{sala2020ergodicity}. The first question was whether or not 1D models exist with $0<|K_{{\rm max},L}| / |\mch| < 1$ in the thermodynamic limit (known examples with weak fragmentation all have $|K_{{\rm max},L}|/|\mch| \ra 1$ in the thermodynamic limit).  Our construction answers this question in the affirmative, with examples provided by $\dyn_G$ for any finite non-Abelian $G$ (e.g. $G = S_3$). 
	
	The second question concerned the existence of models where $|K_{{\rm max},L}| / |\mch|$ vanishes more slowly than exponentially as $L\ra\infty$ after specifying to a fixed symmetry sector (in fact an affirmative answer to this question was already provided by the spin-1 Motzkin chain introduced in \cite{Bravyi}, where $|K_{{\rm max},L}| / |\mch| \sim L^{-3/2}$\footnote{This follows from the fact that the probability of a length-$L$ random walk on $\zz$ being a returning Motzkin walk scales as $\sim L^{-3/2}$.}). 
	Our models provide (many) more examples of this phenomenon, as one need only let $G$ be a group with $L^2 \lesssim \dehn(L) \lesssim L^\infty$, a simple example being the discrete Heisenberg group $\sfH_3$.\footnote{If $\dyn_G$ possesses global symmetries and one modifies the above definitions to operate within a fixed symmetry sector, then various models are known where a {\it specific} symmetry sector with polynomially strong fragmentation exists \cite{morningstar2020kinetically,pozderac2023exact,wang2023freezing}. }
	One may further ask whether there are strongly fragmented systems where $|K_{{\rm max},L}|/|\mathcal{H}|$ decays at a rate in between $1/{\rm poly}(L)$ and $\exp(-L)$. The answer to this is again affirmative, with one such example being the focus of Sec.~\ref{sec:bs}.

	\section{Slow thermalization and the time complexity of the word problem} \label{sec:therm_time_bounds} 
	
	From the discussion of the previous section, it is natural to expect that systems with $\dyn_G$ dynamics will have thermalization times controlled by the time complexity of the word problem, as diagnosed by the Dehn functions $\dehn_g(L)$ defined in \eqref{dehng}. In this section we make this relationship precise by using the functions $\dehn_g(L)$ to place lower bounds on various thermalization timescales. In the case of random unitary or classical stochastic dynamics, $\dehn_g(L)$ will be used to bound relaxation and mixing times; for Hamiltonian dynamics it will appear in bounds for hitting times. 
	
	A common feature these timescales have is that they quantify when $\dyn_G$ is able to spread initial product states across Hilbert space. This is often not something that can be probed by looking at correlation functions of local operators, which would be the preferred method for thinking about thermalization. 
	While the bounds derived in this section do not mandate that the relaxation of {\it local operators} also be controlled by $\dehn_g(L)$, in Sec.~\ref{sec:bs} we will explore an explicit example in which they are. Until then, we will focus on the more ``global'' diagnostics of relaxation and hitting times.
	
	We note in passing that it is impossible to use $\dehn_g(L)$---or any other quantity---to place {\it upper} bounds on thermalization timescales without making any additional assumptions about the details of $\dyn_G$ (as without additional assumptions we could always choose $\dyn_G$ to be evolution with a many body localized Hamiltonian, and the relevant timescales would all diverge). Even in the case where $\dyn_G$ is an appropriately constrained form of random unitary evolution, upper bounding the relaxation timescales requires techniques beyond those employed below, and constitutes an interesting direction for future work.

	\subsection{Circuit dynamics}

	We first discuss the case where $\dyn_G$ is generated by a $G$-constrained random unitary circuit, which will we see can be mapped to the case where $\dyn_G$ is a classical Markov process. We assume that $\dyn_G$ is expressed as a brickwork circuit whose gates act on $\ell_R$ sites, with $\ell_R$ the maximum size of a relation in $R$. $\dyn_G(t)$ will be used to denote $t \in \nn$ timesteps of this dynamics, with each timestep consisting of $\ell_R$ staggered layers of gates.
	
	Let us first look at how operators evolve under $\dyn_G(t)$. We use overbars to denote averages over circuit realizations, so that acting on an operator $\mco$ with one layer of the brickwork (corresponding to a time of $t = 1/\ell_R$) gives
	\bea \ob{\mco(t=\ell_R\inv)} & = \ob{\dyn^\da_G(\ell_R\inv) \, \mco\, \dyn_G(\ell_R\inv)} \\ 
	& = \oEE_{\{ U_{j,g},U_{j',g'}\}} \bot_j \bigoplus_{g \in G_{\ell_R}} U_{j,g}^\da \, \mco \,\bot_{j'} \bigoplus_{g' \in G_{\ell_R}} U_{g',j'}  \eea 
	where each $U_{j,g}$ acts on a length-$\ell_R$ block of sites and $G_{\ell_R}$ as before denotes the set of group elements expressible as words of length $\leq \ell_R$. Assuming the $U_{j,g}$ are drawn uniformly from the $|K_{g, \ell_R}|$-dimensional Haar ensemble,\footnote{	Note that unlike in Ref.~\cite{Singh}, here we take the $U_{j,g}$ to be nontrivial in {\it every} $K_{g, \ell_R}$ sector, even the one-dimensional ones (where $U_{j,g}$ acts as multiplication by a random phase). This leads to a comparatively simpler expression for the Markov generator to appear below.  } performing the average gives 
	\bea \label{mcoevolnmat} \ob{\mco(t=\ell_R\inv)} & = \bot_j \( \sum_{g \in G_{\ell_R}} \frac{\Tr [\mco_j \Pi_{K_{g, \ell_R}}]}{|K_{g, \ell_R}|} \Pi_{K_{g, \ell_R}}  \), \eea 
	where we have defined 
	\be \Pi_{K_{g, \ell_R}} \triangleq \sum_{w \in K_{g, \ell_R}} \proj w\ee 
	as the projector onto the space of length-$\ell_R$ words $w$ with $\vp(\k w) = g$, as well as---without loss of generality---taken $\mco$ to factorize as $\mco = \bot_j \mco_j$.
	
	Thus after a single step of the dynamics, all operators completely dephase and become diagonal in the computational basis, operators violating the dynamical constraint evaluate to zero under Haar averaging. We may thus focus on diagonal operators without loss of generality, which we will indicate using vector notation as $\k{\mco(t)}$.  With this notation, \eqref{mcoevolnmat} becomes 
	\be \k{\ob{\mco(t=\ell_R\inv)}} = \mcm_1 \k{\mco}\ee 
	with the matrix 
	\be \label{mcm1} \mcm_1 = 
	\( \sum_{g \in G_{\ell_R}} \wt\Pi_{K_{g, \ell_R}}\)^{\tp L/\ell_R},\ee 
	where 
	\be \wt \Pi_{K_{g, \ell_R}} \triangleq \frac1{|K_{g, \ell_R}|} \sum_{w,w' \in K_{g, \ell_R}} \kb{w}{w'}\ee 
	projects onto the uniform superposition of states within $K_{g, \ell_R}$. 
	
	Different layers of the brickwork likewise define matrices $\mcm_i$ with $i = 1,\dots,\ell_R$ where $i$ denotes the staggering. Defining 
	\be \label{mcm2} \mcm \triangleq \prod_{i=1}^{\ell_R} \mcm_i,\ee diagonal operators evolve as 
	\be \k{\ob{\mco(t)}} = \mcm^t \k{\mco}.\ee 
	$\mcm$ is a symmetric\footnote{Strictly speaking $\mcm = \mcm^T$ only if we double each step of the evolution by including a reflection-related pattern of gates, namely $\mcm \ra \prod_{i=1}^{\ell_R} \mcm_i \prod_{j=1}^{\ell_R} \mcm_{\ell_R+1-j}$; we will tacitly assume this has always been done.} doubly-stochastic matrix, with the smallest eigenvalue of 0 and the largest eigenvalue of 1. Due to the group constraint, $\mcm$ does not have a unique steady state following from that fact that the Markovian dynamics is reducible as  
	\be \mcm = \bigoplus_{g\in G_L} \mcm_g.\ee 
	Each $\mcm_g$ however defines a irreducible aperiodic Markov chain, whose unique steady state is the uniform distribution $\k{\pi_g} \triangleq \frac{1}{|K_{g, L}|} \sum_{w \in K_{g, L}} |w\ran$. Infinite temperature circuit-averaged correlation functions are determined by the mixing time and spectral gap of the $\mcm_g$, which we now relate to the Dehn function. 
	
	Consider first the mixing time of $\mcm_g$, which we may view as a characterization of the thermalization timescale within $K_{g, L}$:
	\be t_{\rm mix}(g) \triangleq \min \{ t \, : \, \max_{\k w \in K_{g, L}} ||\mcm_g^t \k w - \k{\pi_g} ||_1 \leq 1/2 \}.\ee
	A basic result \cite{levin2017markov} about $t_{\rm mix}(g)$ is that it is lower bounded by half the diameter of the state space that $\mcm_g$ acts on, simply because in order to mix, the system must at the minimum be able to traverse across most of state space. This gives the bound 
	\be \label{naivemix} t_{\rm mix}(g) \geq \frac{\dehn_g(L)}{2}.\ee   
	However, this bound is in fact too lose.  Intuitively, this is because to saturate the bound, $\dyn_G(t)$ would need to immediately find the optimal path between any two nodes in configuration space. 
	Since $\mcm$ generates a random walk on state space, the dynamics will instead diffusively explore state space in a less efficient manner. 
	On general grounds one might therefore expect a bound on $t_{\rm mix}(g)$ which is the square of the RHS above. This guess in fact turns out to be essentially correct, with  
	\be \label{dehnsqbound} t_{\rm mix}(g) \geq \frac{\dehn_g^2(L)}{16 \ln |K_{g,L}|},\ee 
	which follows from Prop.~13.7 of \cite{lyons2017probability} after using that the equilibrium distribution of $\mcm_g$ is always uniform on $K_g$ by virtue of $\mcm_g$ being doubly stochastic.
	Since $|K_{g, L}|$ is at most exponential in $L$, we thus have 
	\be t_{\rm mix}(g) \geq C_g\frac{\dehn_g^2(L)}L \ee 
	for some $g$-dependent $O(1)$ constant $C_g$. 
	We generically expect the random walk generated by $\mcm$ to be the fastest mixing local Markov process, and hence expect it to saturate the above bound on $t_{\rm mix}(g)$, a prediction which we will confirm numerically for the example in Sec.~\ref{sec:bs}.

	Correlation functions for operators computed in states in $K_{g, L}$ are controlled by the {\it relaxation time} $t_{\rm rel}$ of $\mcm_g$, defined by the inverse gap of $\mcm_g$:
	\be t_{\rm rel}(g) \triangleq \frac1{1-\l_2},\ee 
	where $\l_2$ is the second-largest eigenvalue of $\mcm_g$. 
	This quantity admits a similar bound to $t_{\rm mix}(g)$, as one can show that \cite{levin2017markov}
	\be t_{\rm rel}(g) \geq \frac{t_{\rm mix}(g)}{\ln (2|K_{g,L}|)},\ee 
	and hence 
	\be \label{trelbound} t_{\rm rel}(g) \geq C'_g\frac{\dehn_g^2(L)}{L^2},\ee 
	with $C'_g$ another $O(1)$ constant. Thus $\dehn_g(L)$ places lower bounds on both mixing and the decay of correlation functions. Note that the operators whose correlators decay as $t_{\rm rel}(g)$ need not be local; indeed the obvious ones to consider are projectors like $\proj w$. In Sec.~\ref{sec:bs} we will however explore an example which possesses local operators that relax according to \eqref{trelbound}.  
	
	Note that the mixing and relaxation times are worst case measures of the thermalization time.
	We can additionally define a ``typical'' mixing time $\ob{t_{\rm mix}}(g)$ as 
	\be \ob{t_{\rm mix}}(g) \triangleq \min\{ t\, : \, {\rm Pr}_{\k w}( ||\mcm^t_g|w\ran - \k{\pi_g} ||_1 \leq 1/2) \geq 1/2\},\ee 
	where ${\rm Pr}_{\k w}$ indicates the probability over words uniformly drawn from $K_{g, L}$. By following similar logic as above, one can derive similar bounds on $\ob{t_{\rm mix}}(g)$ in terms of the typical Dehn function.

	\subsection{Hamiltonian dynamics}
	
	We have thus far discussed mixing times of the Markov chains which arise from random unitary dynamics. However, purely Hamiltonian dynamics does not mimic that of a Markov chain, and being reversible it possesses no direct analogues of mixing and relaxation times.  Nevertheless, the Dehn function can still be used to place bounds on the timescales taken for time evolution to spread wavefunctions across in Hilbert space. To be more quantitative, we will focus on the
	{\it hitting time} $t_{\rm hit}(g)$ of $\dyn_G$, which we define as the minimum time needed for product states in $K_{g, L}$ to ``reach'' all other product states in $K_{g, L}$. Since $\lan w' | e^{-iHt} | w\ran$ will generically be nonzero for all $\k w,\k{w'}\in K_{g, L}$ as soon as $t>0$, we will need a slightly different definition of $t_{\rm hit}(g)$ as compared with the Markov chain case. 
	
	To define $t_{\rm hit}(g)$ more precisely, we first define the hitting amplitude between two computational-basis product states $\k w, \k{w'} \in K_{g, L}$ as
	\be h_{ww'}(t;g) \triangleq |\lan w' | \dyn_G(t)| w\ran|^2.\ee 
	Note that $h_{w,\cdot}(t;g)$ is a probability distribution on $K_{g, L}$, with $\sum_{w'\in K_{g, L}} h_{ww'}(t;g) = 1$ for all $ \k w\in K_{g, L}$.
	We define $t_{\rm hit}(g, w, w')$ between two words $w$ and $w'$ as the minimum time for which $h_{ww'}(t;g)$ reaches a fixed fraction of its infinite-time average $\ob{h_{ww'}}(g)$, defined by
	\be \ob{h_{ww'}}(g) \triangleq \lim_{T\ra\infty} \frac1T\int_0^Tdt \, h_{ww'}(t;g).\ee
	The hitting time is the maximum over all pairs of words $\k w, \k{w'}\in K_{g, L}$ of $t_{\rm hit}(g, w, w')$, which is written as
	\be t_{\rm hit}(g) \triangleq \max_{\k w, \k{w'} \in K_{g, L}} \min\{ t \, : \, h_{ww'}(t;g) > \frac1{2} \ob{h_{ww'}}(g) \} \ee 
	
	To bound $t_{\rm hit}(g)$, we first evaluate the hitting amplitude as  
	\begin{equation}
		h_{ww'}(t;g) = |\lan w' |e^{-iHt} |w\ran|^2 = \abs{\sum_{k = 0}^{\infty} \frac{t^k}{k\,!} \mel{w'}{(-iH)^k}{w}}^2 
	\end{equation}
	We know that the minimum $k$ such that $\lan w' | H^k |w\ran \neq 0$  is given by $\dehn_g(w,w')$, which for brevity we denote by $d_{ww'}$ in the following. The first $d_{ww'}$ terms in the above sum will thus vanish, and so truncating the above Taylor series at the leading nonzero term, we apply the remainder theorem to find
	\begin{equation}
		\label{hwwtbound}	h_{ww'}(t;g) \leq  \(\frac{ t^{d_{ww'}} |\lan w'| H^{d_{ww'}} | w\ran|}{d_{ww'}\,!}\)^2. 
	\end{equation}
	To diagnose the hitting time we need to compare $h_{ww'}(t;g)$ with its long time average $\ob{h_{ww'}}(g)$. We do this by relating the above matrix element $|\lan w'| H^{d_{ww'}} | w\ran|^2$ to $\ob{h_{ww'}}(g)$ as follows. 	
	Writing $\{ \k E\}$ for $H$'s eigenbasis, 
	\bea \label{matelementbound}
	|\lan w | H^{d_{ww'}}&| w'\ran |^2 = \abs{\sum_E \braket{w}{E} \braket{E}{w'} E^{d_{ww'}}}^2 \\
	& \leq 	\left(\sum_E |\braket{w}{E}|^2 |\braket{E}{w'}|^2\right) \left(\sum_E |E|^{2d_{ww'}}\right) 
	\eea 
	where the second line follows from the Cauchy-Schwarz inequality. The first factor in parenthesis is simply $\ob{h_{ww'}}(g)$. This can be readily verified: 
	\bea 
	\ob{h_{ww'}}(g) & \triangleq \lim_{T \to \infty} \frac{1}{T}\int_0^T  h_{ww'}(t;g) \,dt \\ & = \lim_{T \to \infty} \frac{1}{T} \sum_{E,E'} \braket{w}{E} \braket{E}{w'} \braket{w'}{E'} \braket{E'}{w} \\ & \qq \times \int_0^T e^{-i(E-E') t}\,dt
	\eea 
	which, assuming that the spectrum of $H|_{K_{g, L}}$ is non-degenerate (which we assume to be the case throughout this paper), gives
	\begin{equation}
		\ob{h_{ww'}}(g) = 
		\sum_{E} \abs{\braket{w}{E}}^2 \abs{\braket{E}{w'}}^2.
	\end{equation}
	Inserting this in \eqref{matelementbound}, 
	\bea 
	|\lan w | H^{d_{ww'}} | w'\ran |^2 & \leq |\mch| \ob{h_{ww'}}(g)   \times ||H||_\infty^{2d_{ww'}}. \eea 
	Since $H$ is local, $||H||_\infty = \l L$ for some $O(1)$ constant $\l$. We may thus write \eqref{hwwtbound} as 
	\be h_{ww'}(t;g) \leq \(\frac{(\l L t)^{d_{ww'}}}{d_{ww'}!}\)^2 |\mch| \ob{h_{ww'}}(g).\ee 
	
	As long as $d_{ww'} = \Omega(L)$, $h_{ww'}(t;g)$ is thus much smaller than its equilibrium value when $t = \O(d_{ww'}/L)$. Indeed, from Stirling's approximation we have 
	\be h_{ww'}(\eta d_{ww'}/L;g) \leq \frac{(\l\eta e)^{2d_{ww'}}}{\twp d_{ww'}} |\mch| \ob{h_{ww'}}(g),\ee
	which can always be made exponentially small in $L$ by choosing $\eta$ appropriately if $d_{ww'} = \O(L)$. We conclude that 
	\be t_{\rm hit}(g) \geq \frac{\dehn_g(L)}L,\ee 
	which is essentially the same bound as our naive result \eqref{naivemix} for $t_{\rm rel}(g)$ in the case of random circuit evolution. It would be interesting to see if this bound could be improved to the square of the RHS, as in \eqref{trelbound}.
	
	Finally, we note that just as with the mixing time, a {\it typical} hitting time may also be defined as 
	\be \ob{t_{\rm hit}}(g) \triangleq \sup\{ t\, : \, {\rm Pr}_{w,w'\in K_{g, L}}(t_{{\rm hit}, ww'} \geq t) \geq 1/2\}.\ee
	Arguments similar to the above then show that $\ob{t_{\rm hit}}(g)$ admits a similar bound in terms of $\ob{\dehn}_g(L)$.

	\section{The Baumslag-Solitar group: exponentially slow hydrodynamics} \label{sec:bs}
	
	The previous two sections have focused on developing parts of the general theory of semigroup dynamics. In this section we take an in-depth look at a particular example which exhibits anomalously long relaxation times.  
	
	Our example will come from a family of groups whose Dehn functions scale exponentially with $L$, yielding word problems with large time complexity.\footnote{The {\it spatial} complexity of these examples---which will be discussed later in Sec.~\ref{sec:ffbs}---is however trivial.} When $\dehn(L) \sim \exp(L)$, the hitting and mixing times discussed in the previous section are exponentially long. It is perhaps already rather surprising that a translation invariant local Markov process / unitary circuit can have mixing times that are this long, but our example is most interesting for another reason: it {\it also} possesses a global symmetry whose conserved charge takes $\dehn(L)$ time to relax. This allows the slow time complexity of the word problem to be manifested in the expectation values of simple local operators, rather than being hidden in hitting times between different product states. 
	
	Our example comes from a family of groups known in the math literature as the {\it Baumslag-Solitar} groups \cite{baumslag1969non}, which are parameterized by two integers $n,m \in \nn$. Each group $\bs(n,m)$ in this family is generated by two elements $\tta$ and $\ttb$ obeying a single nontrivial relation, with the group presentation 
	\be \bs(n,m) = \lan \tta,\, \ttb\, | \, \texttt{b} \texttt{a}^m \texttt{b}^{-1} = \texttt{a}^n\ran.\ee
	Models with $\dyn_{\bs(n,m)}$ dynamics are thus most naturally realized in spin-2 systems with $\max(n,m)$-local dynamics and onsite Hilbert space 
	\be \mch_{\rm loc} = {\rm span}\{ \k\tta,\, \k{\tta\inv},\, \k \ttb ,\, \k{\ttb\inv}, \, \k\tte \}.\ee 
	The simplest Baumslag-Solitar group which is interesting for our purposes is $\bs(1,2)$, and the rest of our discussion will focus on this case. For convenience, in what follows we will write $\bs(1,2)$ simply as $\bs$. 
	
	The nontrivial relation in $\bs$ reads
	\be \label{bs12reln} \tta \ttb = \ttb \tta^2, \ee 
	so that $\tta$ generators duplicate when moved to the right of $\ttb$ generators. By taking inverses, the $\tta$ generators are also seen to duplicate when moved to the left of $\ttb\inv$ generators:
	\be \ttb\inv \tta = \tta^{2} \ttb\inv.\ee 
	This duplication property means that an $O(n)$ number of $\ttb$ and $\ttb\inv$ generators can be used to grow an $\tta$ generator by an amount of order $O(2^n)$, as 
	\be \ttb^{-n} \tta \ttb^n = \tta^{2^n}.\ee 
	We will see momentarily how this exponential growth can be linked to the Dehn function of $\bs$, which also scales exponentially. 
	
	Another key property of \eqref{bs12reln} is that it preserves the number of $\ttb$'s. This means that models with $\dyn_G$ dynamics possess a $U(1)$ symmetry generated by the $n_\ttb$, defined as 
	\be \label{nbdef} n_\ttb \triangleq \sum_i n_{\ttb,i} \triangleq \sum_i (\proj \ttb_i - \proj{\ttb\inv}_i).\ee 
	It is this conserved quantity which will display the exponentially slow hydrodynamics advertised above.

	\ss{A geometric perspective on group complexity}
	
	To understand how dynamics in $\dyn_\bs$ works, we will find it helpful to introduce a {\it geometric} way of thinking about the group word problem.\footnote{This picture is most useful for groups, rather than semigroups; see Sec.~\ref{sec:nongroup} and App.~\ref{app:group_props} for discussions on this point.} Given a general discrete group $G$, we will let $\cg_G$ denote the {\it Cayley graph} of $G$. Recall that $\cg_G$ is a graph with vertices labeled by elements of $G$ and edges labeled by generators of $G$ and their inverses, with two vertices $\tth,\ttk$ being connected by an edge $\ttg$ iff $\tth = \ttg \ttk$. As simple examples, $\cg_{\zn}$ for $\zn = \lan \ttg \, | \, \ttg^N = \tte\ran $ is a length-$N$ closed cycle; $\cg_{\zz^2}$ for $\zz^2 = \lan \ttx, \, \tty \, | \, \ttx\tty = \tty \ttx\ran$ is a 2D square lattice; and $\cg_{\zz\ast\zz}$ for the free group $\zz\ast \zz = \lan \ttx , \, \tty \, | \, \ran$ is a Bethe lattice with coordination number 4.\footnote{While the exact structure of $\cg_G$ is different for different presentations of $G$, the {\it large scale structure} of $\cg_G$ (scalings of geodesic lengths, graph expansion, and so on) is presentation-independent (see e.g. Ref.~\cite{clay2017office} for background). Therefore we will mostly refer to the Cayley graph of a {\it group}, rather than a particular presentation thereof. }
	
	\begin{figure*}
		\includegraphics[width=\linewidth]{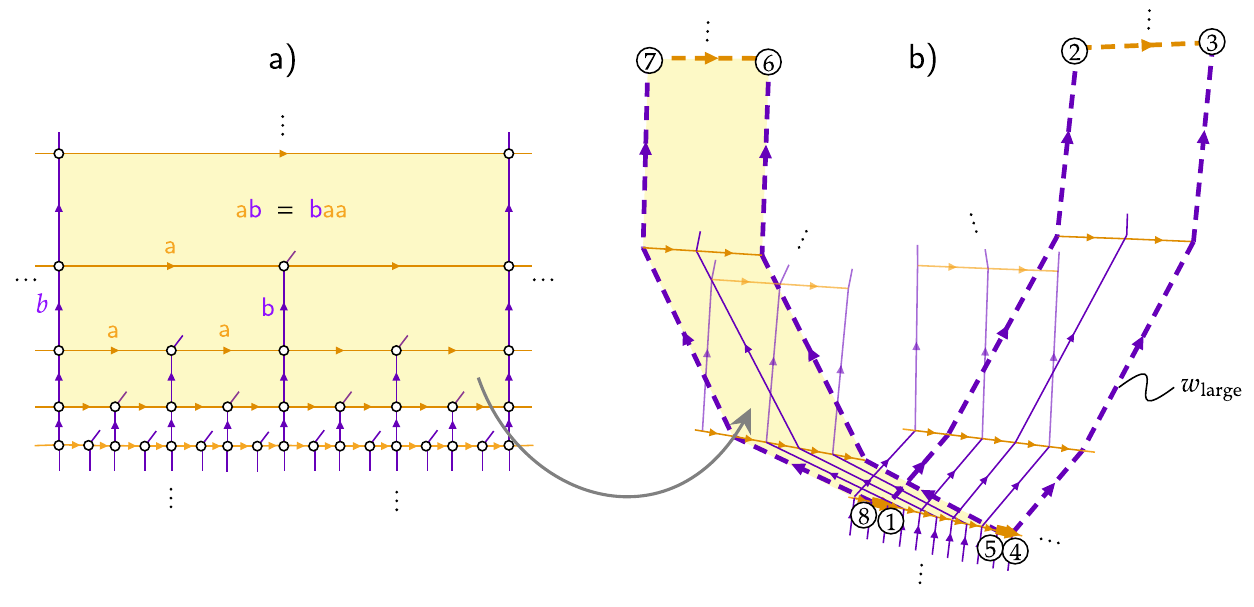}%
		\caption{The group geometry of $\bs(1,2)$. The Cayley graph of $\bs(1,2)$ is constructed from an infinite number of two-dimensional sheets glued together in a tree-like fashion. \textsf{a)} A single sheet of the Cayley graph, which resembles a hyberbolic plane. Orange arrows denote multiplication by $\tta$ and purple arrows multiplication by $\ttb$, with the boundary of each plaquette being the defining relation $\tta\ttb\tta^{-1}\tta^{-1}\ttb^{-1}=e$. The diagonal purple lines denote edges which connect to different sheets; all diagonal edges connected to vertices at the same "height" connect to the same sheet. \textsf{b)} How different sheets are glued together to form the full Cayley graph. The shaded yellow region denotes how a section of a single sheet is embedded in the full Cayley graph. The bold dashed path denotes the path traced out by a ``large'' word $w_\mathrm{large}(n) =\tta \ttb^{-n}\tta\inv \ttb^n \tta\inv\ttb^{-n} \tta \ttb^n$, which is homotopic to the identity ($\vp(\k{w_{\rm large}}) = e$) and possesses an area scaling exponentially with its length. In the figure $n=3$, and the path on the Cayley graph is ordered by points 1 to 8, according to the 8 different components of $w_{\rm large}$ (read from right to left). Thus the point 2 corresponds to the word $\ttb^3$, the point 5 corresponds to $\tta\inv\ttb^{-3}\tta\ttb^3$, and so on. }
		\label{fig:bs_geometry}
	\end{figure*}
	
	It is useful to realize that from any Cayley graph $\cg_G$, we can always construct a related 2-complex by associating oriented faces (or 2-cells) to each of $\cg_G$'s elementary closed loops; the 2-cells have the property that the product of generators around their boundary is a relation in $R$. For the just-mentioned examples, $\zn$ would have one $N$-sided face, $\zz^2$ would have a face for each plaquette of the square lattice (with the generators around the plaquette boundaries reading $\ttx\tty\ttx\inv\tty\inv$), and the free product $\zz\ast\zz$ would have no faces at all (due to its lack of non-trivial relations). The 2-complex thus constructed is known as the {\it Cayley 2-complex} of $G$, and we will abuse notation by also referring to it as $\cg_G$. 
	
	Any group word $w\in W(G)$ naturally defines an oriented path in $\cg_G$ obtained by starting at an (arbitrarily chosen) origin of $\cg_G$ and moving along edges based on the characters in $w$. The endpoint of this path on the Cayley graph is thus the group element $\vp(\k w)$. Additionally, applying local relations in $R$ to $w$ deforms this path while keeping its endpoints fixed. This gives a geometric interpretation of the subspaces $K_{g, L}$: 
	\bea K_{g, L} = \Big\{P| P =\text{length-$L$ path from origin to $g$ in $\cg_G$}\Big\}.\eea
	In particular, the (largest) subspace $K_{e, L}$ is identified with the set of all length-$L$ closed loops in $\cg_G$. The number $N_K(L)$ of Krylov sectors is the number of vertices of $\cg_G$ located a distance $\leq L$ from the origin. 
	
	Having provided a relationship between Krylov subspaces and the Cayley 2-complex, what is the geometric interpretation of the Dehn function? Based on the definition \eqref{groupdehn}, we can restrict our attention to deformations $D(w \sqig \tte^L)$ for $w \in K_{e, L}$, which are simply homotopies that shrink the loop defined by $w$ down to a point. The loop passes across one cell at each step, and the number of steps in $D(w \sqig \tte^L)$ is the area enclosed by the surface swept out by the homotopy. In particular, the Dehn function of a $w$ is 
	\be \dehn(w,\tte^L) = \min_{S\, : \, \p S = w} {\rm Area}(S),\ee 
	where the minimum is over surfaces in $\cg_G$ with boundary $w$; we will thus refer to 
	\be {\rm Area}(w) \triangleq \dehn(w,\tte^L) \ee 
	as the {\it area} of $w$. The Dehn function of the group is then the largest area of a word in $K_{e, L}$:
	\be \dehn(L) = \max_{\text{loops $w$ of length $L$}} \dehn(w,\tte^L).\ee 
	
	This perspective is important as it gives the algorithmic definition of $\dehn(L)$ a geometric meaning.  The large-scale geometry of $\cg_G$ thus directly affects the complexity of the word problem, and consequently the thermalization dynamics of $\dyn_G$. 
	
	\ss{The geometry of $\bs$ and the fragmentation of $\dyn_\bs$}\label{sec:bsel}
	
	The simple examples (discrete groups, abelian groups, free groups) discussed above are all geometrically uninteresting, but $\bs$ group is a notable exception. $\cg_\bs$ has the structure of an infinite branching tree of hyperbolically-tiled planes, and is illustrated in Fig.~\ref{fig:bs_geometry}. 
	To understand how this comes about, recall that $\ttb^{-n}\tta \ttb^n = \tta^{2^n}$. 
	This means that for all $n$, $\ttb^{-n}\tta \ttb^n \tta^{-2^n}$ forms a closed loop in $\cg_\bs$. These closed loops give rise to a tiling of the hyperbolic plane, as shown in Fig.~\ref{fig:bs_geometry}~(a). Letting multiplication by $\tta$ correspond to the motion along the $\uvx$ direction and multiplication by $\ttb$ to the motion along $\uvy$, the hyperbolic structure comes from the fact that to moving by $2^n$ sites along the $\uvx$ direction of the Cayley graph can either be accomplished by direct path (multiplication by $\tta^{2^n}$) or a path that requires exponentially fewer steps which first traverses $n$ steps along the $\uvy$ direction before moving along $\uvx$ (multiplication by $\ttb^{-n}\tta \ttb^n$). 
	
	The full geometry of $\cg_\bs$ is more complicated than a single hyperbolic plane, however. As shown in Fig.~\ref{fig:bs_geometry}~(a), this can be seen from the fact that the word $\ttb\tta\ttb$ cannot be embedded into the hyperbolic plane. Instead, this word must ``branch out'' into a new sheet, which also forms a hyperbolic plane. Fig.~\ref{fig:bs_geometry}~(b) illustrates the consequences of this for the Cayley graph, whose full structure consists of an infinite number of hyperbolic planes glued together in the fashion of a binary tree (formally, the presentation complex of $\bs$ is homeomorphic to $\rr \times T_3$, where $T_3$ is a 3-regular tree). 
	
	The locally hyperbolic structure of $\cg_\bs$ means that the number of vertices within a distance $L$ of the origin---and hence the number of Krylov sectors $N_K(L)$---grows exponentially with $L$. In App.~\ref{app:bs_details} we argue that the base of the exponent is very nearly $\sqrt3$:
	\be N_K(L) \approx \sqrt3^L.\ee 
	Interestingly, despite having exponentially many Hilbert space fragments, it is known that the largest sector $K_{e, L}$ occupies a fraction of the full Hilbert space $\mch$ that decreases {\it sub-exponentially} with $L$: 
	\be \label{id_sector_scaling} \frac{|K_e|}{|\mch|} \sim e^{-\a L^{1/3}},\ee 
	where the numerical results in Fig.~\ref{fig:id_sector_scaling}---obtained by sampling random words and postselecting on membership in $K_{e, L}$---indicate that $\a \approx 1.84$. Models with $\dyn_\bs$ dynamics provide an example (indeed the first that we are aware of) of a strongly fragmented model where the largest sector occupies neither an exponentially nor polynomially small fraction of Hilbert space.
	
	\begin{figure}
		\centering
		\includegraphics[width=.39\textwidth]{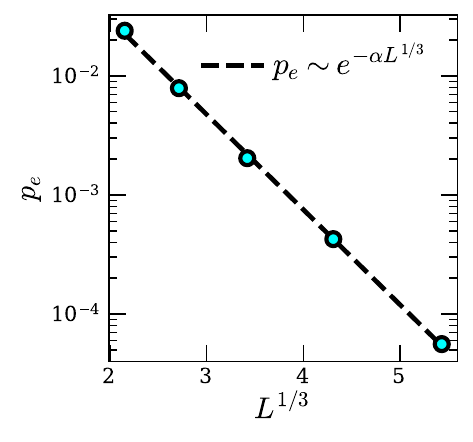}
		\caption{The probability $p_e$ for a randomly chosen word to lie in the identity sector, determined according to the procedure described in App.~\ref{app:bs_details}. The dashed line is a fit to $p_e \propto e^{-\a L^{1/3}}$ with $\a \approx 1.84$. }
		\label{fig:id_sector_scaling}
	\end{figure}
	
	\ss{The Dehn function}
	We finally have the necessary tools explain why the Dehn function of $\bs$ grows exponentially with $L$. 	
	Define the word
	\begin{equation}
		w_\mathrm{large}(n) \triangleq \tta \ttb^{-n}\tta\inv \ttb^n \tta\inv\ttb^{-n} \tta \ttb^n
	\end{equation}
	which is of length $|w_\mathrm{large}(n)| = 4n + 4$, belongs to $K_{e, L}$, and is shown for $n=3$ in Fig.~\ref{fig:bs_geometry} as the thick dashed line in panel~(b).   As a path, this word can be broken into two legs. The first leg moves to the vertex labelled by $\tta^{2^n}$ by first going ``up'' the hyperbolic plane of a given sheet, traversing one step along $\uvx$, and then coming back ``down''. The second leg moves back to the origin, but does so by passing along a {\it different} sheet of the tree. From Fig.~\ref{fig:bs_geometry}, it is clear that the area of the minimal surface bounding $w_\mathrm{large}(n)$ is exponentially large; more precisely it is 
	\be {\rm Area}(w_\mathrm{large}(n)) = 2 \sum_{i=0}^{n-1} 2^i = 2^{n+1}-2,\ee  
	which scales exponentially with $n$. 
	This construction thus shows that $\dehn(L) = \O(2^L)$ \cite{gersten1992dehn}. This bound is in fact tight (see App.~\ref{app:bs_details}), and so  
	\be \dehn(L) = \ct(2^L).\ee 
	From the results of Sec.~\ref{sec:therm_time_bounds}, this implies that $\dyn_G$ has exponentially long mixing, relaxation, and hitting times.

	A more sophisticated treatment needs to be given in order to understand the scaling of the {\it typical} Dehn function $\ob{\dehn}(L)$. 
	To our knowledge, this question has not been answered in the math literature. While we will not provide a completely rigorous proof of $\ob{\dehn}(L)$'s scaling, a combination of physical arguments and numerics---which we relegate in their entirety to App.~\ref{app:bs_details}---indicate that 
	\be \label{avgdehnbs} \ob{\dehn}(L) \sim 2^{\sqrt L}.\ee 
	The rough intuition behind this result is that a typical closed-loop path in $K_{e, L}$ will roughly execute a random walk along the tree part of $\cg_\bs$, reaching a ``depth'' of $\sqrt L$; from this point it is then able to enclose an area exponential in this depth. Giving a rigorous proof of~\eqref{avgdehnbs} could be an interesting direction for future research. 
	
	\ss{Exponentially slow hydrodynamics}
	
	An observation about the word $w_\mathrm{large}(n)$ is that it contains a density wave of $\ttb$s, with the profile of $n_{\ttb,i}$ looking like two periods of a square wave as a function of $i$. Utilizing the fact that ${\rm Area}(w_{\rm large}(n)) \sim 2^n$, 
	one sees that this density wave takes an {\it exponentially long time to relax}, with a thermalization time $t_{\rm th}(n) \sim [{\rm Area}(w_{\rm large}(n))]^2 = \O(2^{2n})$ (where the square of the area comes from the square in \eqref{trelbound}).
	Furthermore,  almost any word with an $n_\ttb$ density wave will admit a similar bound on $t_{\rm th}$. Indeed, consider words ${w_{\rm large}(n,L)}$ obtained by extending $w_{\rm large}(n)$ to length $L > 4n+4$ by inserting $L-(4n+4)$ random characters at random points of $w_{\rm large}(n)$ (we will assume for simplicity that in fact $ w_{\rm large}(n,L) \in K_{e, L}$, but this assumption is not crucial). The only way for ${\rm Area}({w_{\rm large}(n,L)})$ to be significantly smaller than ${\rm Area}(w_\mathrm{large}(n))$ is if the added characters cancel out almost the entirety of the $\ttb$ density wave, or cancel a large number of $\tta$'s located near the peaks and troughs of the density wave. For a random choice of ${w_{\rm large}(n,L)}$ these situations are exponentially unlikely to occur, and we expect 
	\be \label{wlargearea} {\rm Area}(w_{\rm large}(n,L)) = \O(2^n)\ee 
	with probability 1 in the large $n$ limit. 
	The long relaxation time of the density wave is attributed to the long mixing time of $\bs$ because of its large Dehn function (see Sec.~\ref{sec:therm_time_bounds}).  This is remarkable because probing the large scale complexity of $\bs$ only requires studying the dynamics of {\it local} operators, namely those that overlap with $n_\ttb$.  We henceforth will denote the relaxation and mixing times by $t_{th}$.
	
	Of course since $2^n$ is the {\it smallest} possible time needed to map $w_{\rm large}(n)$ to $\tte^{4n+4}$, it gives only a lower bound on $t_{\rm th}$. In the case where $\dyn_\bs(t)$ is generated by a classical Markov process or random unitary circuit dynamics, $\dyn_G(t)$ effectively leads to words executing a random walk on configuration space of loops. Due to the diffusive nature of this random walk, this leads us to expect that the true thermalization time instead scales as the square of its lower bound, namely 
	\be \label{tthwlarge} t_{\rm th}(w_{\rm large}(n)) = O(2^{2n}),\ee
	an estimate that we will confirm shortly in Sec.~\ref{sec:numerics}. 
	
	To examine the relaxation in more detail, consider a state $\k{w_{A,q}}$ which contains a $n_\ttb$ density wave of momentum $q$ and amplitude $A$, but which is otherwise random. By this, we mean that the expectation value of $n_{\ttb,i}$ in $\k{w_{A,q}}$ is (switching to schematic continuum notation) 
	\be n_{\ttb}(x) = A\sin(qx),\ee
	and that $\k{w_{A,q}}$ is chosen randomly from the set of all states in $K_{e, L}$ that have this expectation value (with the restriction to $K_{e, L}$ done purely for notational simplicity). 
	
	The ``depth'' $n_A$ that $w_{A,q}$ proceeds into $\cg_\bs$ is equal to the contrast of the density wave, which we define as the integral of the $\ttb$ density over a quarter period of the density wave: 
	\be n_A \triangleq \int_0^{\pi/2q} n_\ttb(x) = \frac{A}q.\ee 
	We thus expect the density wave defined by $\k{w_{A,q}}$ to have a thermalization timescale of 
	\be t_{\rm th} = O(2^{2n_A}) = O(2^{2A/q}).\ee 
	Note that this exponential timescale is not visible in the standard linear-response limit, in which one takes $A \to 0$ before taking $q \to 0$: the two limits lead to qualitatively different relaxation. In the standard linear-response limit, fluctuations at scale $q$ will decorrelate on the {\it typical} relaxation timescale $\sim 2^{\sqrt{L}}$.

	Beyond $t_{\rm th}$, we would also like to know how the amplitude $A$---or equivalently the contrast $n_A$---behaves as a function of time. To estimate this, consider first how $w_{\rm large}(n)$ relaxes. By the geometric considerations of Sec.~\ref{sec:bsel}, the shortest homotopy reducing $w_{\rm large}(n)$ to the identity word must first map $w_{\rm large}(n) \ra w_{\rm large}(n-1)$, then $w_{\rm large}(n-1) \ra w_{\rm large}(n-2)$, and so on; see Fig.~\ref{fig:word_homotopy} for an illustration. The self-similarity of this process means that the time needed to map $w_{\rm large}(m) \ra w_{\rm large}(m-1)$, namely $t_{\rm th}(w_{\rm large}(m-1))$, is equal to the combined time needed to perform all of the maps $w_{\rm large}(k) \ra w_{\rm large}(k-1)$ with $k<m$. Applying this observation to the density wave under consideration and using \eqref{tthwlarge}, we estimate 
	\be n_A(q,t_{\rm th}(1-2^{-k})) \approx n_A(q,0) - k\ee
	for $k \leq n_A(0)$. Writing the $k$ on the RHS in terms of the time $t = t_{\rm th}(1-2^{-k})$, these arguments suggest that the density relaxes as 
	\be \label{eq:nbt} n_\ttb(q,t) \approx \ct(t_{\rm th}- t) \( n_{\ttb,0} + \log_2 \( 1 - t/t_{\rm th} + 2^{-n_{\ttb,0}}t/t_{\rm th}\)\),  \ee
	where $t_{\rm th} = C2^{2A/q}$ with $C$ an $O(1)$ constant, $n_{\ttb,0}$ is shorthand for $n_\ttb(q,0)$, and the $2^{-n_{\ttb}(q,0)} t/t_{\rm th}$ inside the logarithm ensures that $n_{\ttb}(q,t_{\rm th}) = 0$. 
	
	An interesting consequence of \eqref{eq:nbt} is that the relaxation of the density wave happens ``all at once'' in the large contrast limit (e.g. small $q$ at fixed $A$), meaning that the density profile remains almost completely unchanged until a time very close to $t_{\rm th}$, at which point the density wave is abruptly destroyed. Indeed, define the {\it collapse timescale} 
	\be t_{\rm col}(\ep) \triangleq \min \{ t \, : \, n_A(q,t) < (1-\ep) n_A(q,0)\}\ee 
	as the time at which the collapse of the density wave becomes noticeable within a precision controlled by the constant $0<\ep <1$. Then the form \eqref{eq:nbt} implies that 
	\be \label{suddencollapse} \lim_{\ep \ra 0}\lim_{A/q \ra \infty} \frac{t_{\rm col}(\ep)}{t_{\rm th}} = 1,\ee
	implying that the collapse is instantaneous in the large-contrast limit. 
	
	These arguments show that density waves collapse abruptly, but do not tell us about the statistical distribution of collapse times obtained when considering an ensemble of realizations of the dynamics (specific disorder realizations of random unitary circuits, specific choices of Hamiltonian, etc.). In systems which relax suddenly, the width $\s_{t_{\rm th}}$ of this distribution can be either parametrically smaller (in system size) than the mean collapse time $t_{\rm th}$, or can scale as $\O(t_{\rm th})$. The former case occurs in Markov chains displaying a ``cutoff'' (see e.g. \cite{aldous1986shuffling,diaconis1996cutoff}), which can occur when the chain describes random motion on a highly-connected graph, or when it describes relaxation in a metastable system (where relaxation is caused by rare events that yield a Poissonian scaling $\s_{t_{\rm th}} \sim  \sqrt{t_{\rm th}} = o(t_{\rm th})$). The latter case is typical in systems which relax diffusively, with an unbiased random walk on $\zz^d$ being a typical example. 
	
	From the geometric considerations of this section, it is perhaps not clear which scenario should apply a priori. The numerics of the following section tentatively point to the latter scenario, with the distribution of thermalization times roughly obeying $\s_{t_{\rm th}} \sim t_{\rm th}$. Giving a more precise characterization of this scaling would be interesting to explore in future work. 

	\ss{Numerics} \label{sec:numerics} 
	
	\begin{figure*}
		\centering
		\includegraphics[scale=0.285]{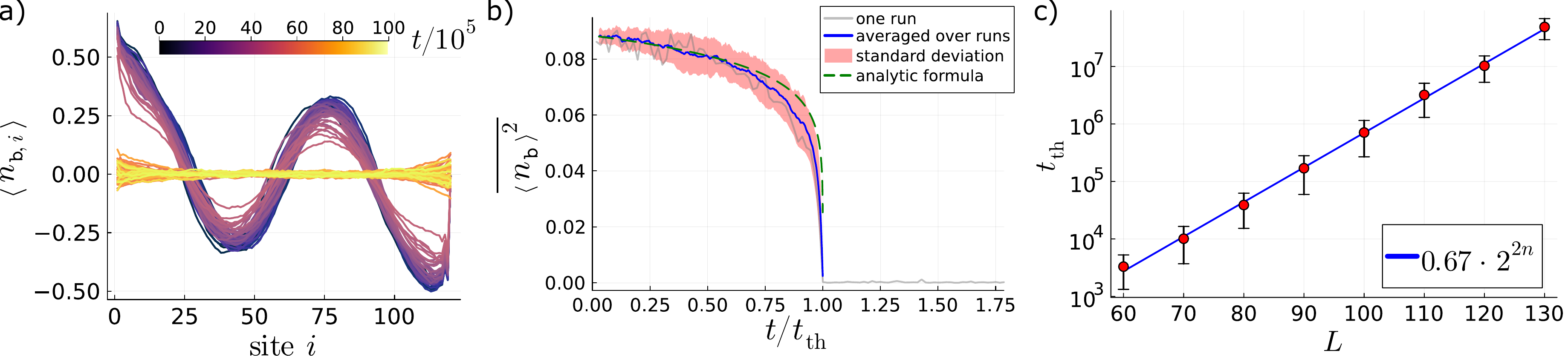}
		\caption{
			Stochastic circuit time evolution of observables associated with the $\ttb$-charge for the system initialized in a state $\k{ w_\mathrm{large}(n=L/10,L)}$, which is a large-area word $w_\mathrm{large}(n=L/10)$ with identities inserted at random places. (a) Time evolution of the spatial profile of $\ttb$-charge, $\expval{n_{\ttb, i}}$, where averaging is performed over a time window of $T = 10^5$ brickwork layers. (b) Time evolution of observable $\overline{\expval{n_{\ttb}}^2}$ from Eq. (\ref{eq:observable}). Thermalization occurs at time $t_\mathrm{th}$, when this observable drops to zero, which corresponds to the $\ttb$-charge density wave collapsing to a flat profile. The run corresponding to panel (a) is shown in grey, while the blue curve corresponds to an average over several independent runs, with the red shade showing the standard deviation. Time for each run has been rescaled by the respective thermalization time. The analytic formula from Eq.\eqref{eq:nbt} is shown in dashed green (rescaled by the value at $t=0$). (c) Thermalization time scales exponentially with the system size $L$ if the density of $\ttb$'s in the initial word is kept fixed, $n = L/10$.}
		\label{fig:bs_t_th}
	\end{figure*}
	
	We now present the results of numerical simulations that let us take a more detailed look into the relaxation of $n_{\ttb}$ and confirm the predictions made in the previous subsection. 
	
	\sss{Stochastic circuits}
	
	Our simulations all treat $\dyn_\bs$ as time evolution under $\bs$-constrained random unitary circuits. Because off-diagonal operators are rendered diagonal after a single step of random unitary dynamics (as was shown in \eqref{mcoevolnmat}), we can without loss of generality focus on the evolution of diagonal operators. 
	In Sec.~\ref{sec:therm_time_bounds}, we showed that the product states associated with diagonal operators evolve in time according to the stochastic matrix $\mcm$ derived in \eqref{mcm1}. 
	Explicitly, since the maximal size of an elementary relation in $\bs$ is $\ell_R= 3$ (e.g. $\tta\ttb\tte = \ttb\tta\tta$), $\mcm$ is most naturally constructed using 3 brickwork layers of 3-site gates: 
	\be \mcm \triangleq \mcm_1 \mcm_2 \mcm_3,\qq \mcm_a \triangleq \bigotimes_{i=0+a}^{\lfloor L/3\rfloor-1} M_{3i + a},\ee 
	where, as in \eqref{mcm1}, the matrices $M_i$ induce equal-weight transitions among all dynamically equivalent 3-letter words: 
	\bea \label{stochasticgate} M_i & \triangleq \sum_{\ttg_1,\dots,\ttg_6 \in S \cup S\inv \cup \{\tte\}} \d_{\ttg_1\ttg_2\ttg_3,\,\ttg_4\ttg_5\ttg_6} \\ & \qq \frac1{|K_{\ttg_1\ttg_2\ttg_3}(3)|}\kb{\ttg_1,\ttg_2,\ttg_3}{\ttg_4,\ttg_5,\ttg_6}_{i,i+1,i+2},\eea 
	where the $1/|K_{\ttg_1\ttg_2\ttg_3}(3)|$ factor ensures that $M_i$ is stochastic, in fact doubly so on account of $M_i^T = M_i$. 
	The steady-state distribution $\k{\pi_g}$ of $\mcm$ within $K_{g, L}$ is accordingly given by the uniform distribution on $K_{g, L}$: 
	\be \k{\pi_g} \triangleq \frac1{|K_{g, L}|} \sum_{w \in K_{g, L}} |w\ran.\ee 
	
	In practice we do not actually diagonalize $\mcm$, but rather use the matrix elements of $\mcm$ to randomly sample updates that may be applied to a computational basis state $\k w$, with the system thus remaining in a product state at all times. A single time step in our simulations corresponds to a single brickwork layer of $\mcm$, i.e. to the application of a single relation at each 3-site block of sites. Since $\lfloor L/3 \rfloor$ relations can be applied at each time step, in these units it is $\dehn(L) / L$---rather than $\dehn(L)$---which lower bounds the mixing and relaxation times of $\mcm$.  
	
	\sss{Slow relaxation}
	
	We start by exploring how long it takes the state $\k{ w_\mathrm{large}(n,L)}$ to relax under $\mcm$'s dynamics, where as above ${ w_\mathrm{large}(n,L)}$ is obtained from $\k{w_{\rm large}(n)}$ by padding it with $L - (4n+4)$ identity characters inserted at random positions. This is done by initializing the system in $\k{w_{\rm large}(n,L)}$ and tracking the local $n_{\ttb,i}$ density over time.
	We focus in particular on the moving average of the $n_\ttb$ density and the fluctuations thereof, defined as 
	\bea \label{eq:observable}
	\lan n_{\ttb,i}\ran(t) & \triangleq \frac{1}{T} \sum_{t'=t}^{t+T} n_{\ttb, i} (t'), \\ 
	\overline{\expval{n_{\ttb}}^2} (t) & \triangleq \frac{1}{L} \sum_{i=1}^L  \lan n_{\ttb,i}\ran^2\eea
	where $n_{\ttb,i}(t) = \delta_{\ttg_i(t), \ttb} - \delta_{\ttg_i(t), \ttb\inv}$ is the $\ttb$-charge on site $i$ at time step $t$ (with $\ttg_i(t)$ the $i$th entry of the state at time $t$), and $T$ is a time window that is small relative to the thermalization timescales of interest, but large enough to suppress short-time statistical fluctuations $n_{\ttb,i}(t)$ (here $\lan\cdot\ran$ denotes averaging over this time window, while $\ob{\,\,\cdot\,\,}$ denotes averaging over space). The equilibrium distribution $\k{\pi_g}\propto \sum_{w \in K_{g,L}} \k{w}$ for $g=e$ satisfies $\lan \pi_e| n_{\ttb,i}|\pi_e\ran \approx 0$ for all $i$ (see App.~\ref{app:bs_details}), and so for initial states $\k w \in K_{e, L}$, any nonzero value of $\lan n_{\ttb,i}\ran(t)$ indicates a lack of equilibration (for $g \neq e$ the average $\lan \pi_g| n_{\ttb,i}|\pi_g\ran$ can be nonzero, and must first be computed in order to diagnose equilibration).

	The evolution of $\lan n_{\ttb,i}\ran(t)$ for a single realization of the dynamics initialized in $\k{w_{\rm large}(n,L)}$ is shown in Fig.~\ref{fig:bs_t_th}~(a) for $n = L/10$ and $L=120$. From this we can clearly see the sudden collapse phenomena predicted above in \eqref{suddencollapse}: the $n_\ttb$ density wave hardly decays at all until very close to $t_{\rm th} \sim 5 \times 10^6$, and which point $\lan n_{\ttn,i}\ran(t)$ rapidly becomes very close to zero. This is quantified in Fig.~\ref{fig:bs_t_th}~(b), which plots the spatially-averaged $n_\ttb$ fluctuations $\ob{\lan n_{\ttb}\ran^2}$ for the same realization. A relatively good fit (dashed line) is obtained using the square of the ``sudden collapse'' function defined in \eqref{eq:nbt}. 
	
	\begin{figure}
		\centering
		\includegraphics[width=.48\tw]{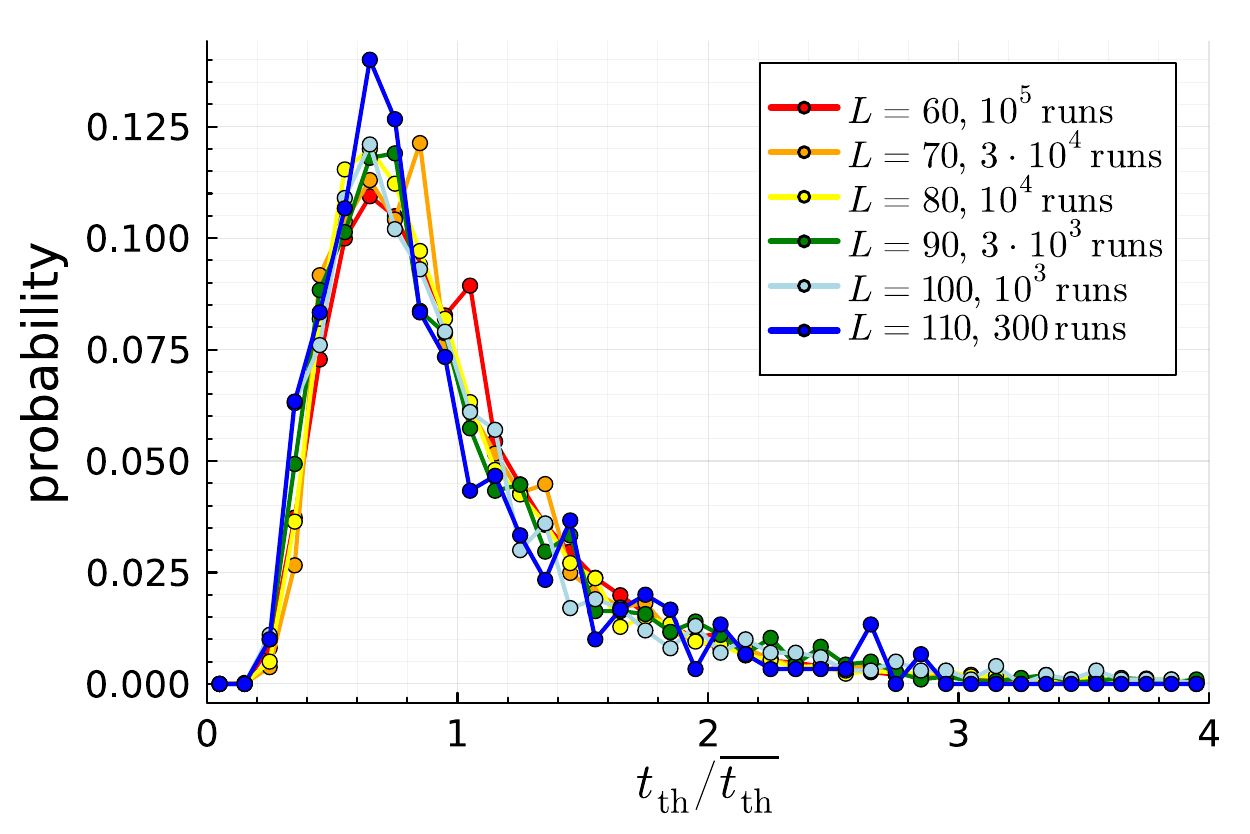 }
		\caption{Distribution of thermalization times $t_{\rm th}$ of the density wave defined by $\k{w_{\rm large}(n,L)}$ with $n=L/10$, shown for different system sizes with 300 circuit realizations. A rather broad distribution is observed.
		}
		\label{fig:tth_dist}
	\end{figure}
	
	We now investigate how the thermalization time of $\k{w_{\rm large}(n,L)}$ scales with $n$, continuing to fix $n/L = 10$ so as to keep both $n,L$ extensive. Operationally, we define the thermalization time as the first time when the magnitude of fluctuations in $n_\ttb$ drop below a fixed fraction of their initial value: 
	\be t_{\rm th} \triangleq \min \{ t\, : \, \ob{\lan n_\ttb\ran^2 }(t) < \frac1{10}\ob{\lan n_\ttb\ran^2 }(0)\}.\ee 
	Fig.~\ref{fig:bs_t_th}~(c) shows the scaling of $t_{\rm th}$ with $L$, which is observed to admit an excellent fit to the predicted scaling of $\sim 2^{2n}$. As shown in Fig.~\ref{fig:tth_dist}, the {\it distribution} of thermalization times across different runs is additionally observed to be rather broad, with a standard deviation that scales approximately in the same way as $t_{\rm th}$.

	While this confirms that the density wave present in $\k{w_{\rm large}(n,L)}$ relaxes on a timescale of $t_{\rm th} \sim 2^{2n}$, it does not show that {\it all} states with $\ttb$ density waves of amplitude $A$ and wavelength $q$ relax on times of order $2^{2A/q}$. This in fact cannot be true, and fast-relaxing density waves always exist. Indeed, consider the word $w_\mathrm{small}(n,L)$ obtained from $w_\mathrm{large}(n,L)$ by replacing all $\tta$s and $\tta\inv$s with $\tte$s. The Dehn time of this word is merely $\dehn(w_\mathrm{small}(n,L)) \sim n^2$, since there are no $\tta$s present to slow down the dynamics of the $\ttb$s (any $\tta\tta\inv$ pairs created in between the segments of the density waves have a net zero number of $\tta$s, and are thus ineffectual at providing a slowdown). 
	This means that the relaxation time of an initial state $\k w$ carrying a $n_\ttb$ density wave cannot be predicted from knowledge of the conserved density alone---one must also have some knowledge about the distribution of $\tta$'s in $\k w$. 
	
	We expect however that $t_{\rm th} \sim 2^{2A/q}$ for {\it generic} states containing an amplitude-$A$, wavelength-$q$ density wave. Indeed, as argued above near \eqref{wlargearea}, this is simply because as long as the value of the $\tta$-charge $\sum_i (\proj \tta_i - \proj{\tta\inv}_i)$ is nonzero in the regions between the segments of the density wave---which will almost always be true for a random density wave state in the thermodynamic limit---the $\tta$s trapped ``inside'' the density wave will take an exponentially long time to escape. Note that this remark applies to a generic density wave state $\k w$, regardless of whether or not $\k w\in K_{e, L}$. 
	
	In Fig.~\ref{fig:bs_b_wave_a_random}, we numerically investigate the relaxation of generic density waves by considering initial states which host density waves of momentum $q = 4\pi/3L$ and amplitude $A = nq$ (with $n$ fixed at $n=L/10$), but which are otherwise random. Compared with $\k{w_{\rm large}(n,L)}$, these random density waves (which generically lie in different $K_{g,L}$ sectors) exhibit a much broader range of thermalization times; some thermalize very quickly, while some never thermalize over our chosen simulation time window (Fig.~\ref{fig:bs_b_wave_a_random}~a). Nevertheless, we still observe the average thermalization timescale to scale exponentially with $n$ (Fig.~\ref{fig:bs_b_wave_a_random}~b). The large sample-to-sample fluctuations of $t_{\rm th}$ make it difficult to reliably extract the exact scaling behavior, but the above reasoning suggests that $t_{\rm th}$ continues to scale as $2^{2n}$ for typical initial states. 
	
	\begin{figure}
		\centering
		\includegraphics[scale=0.4]{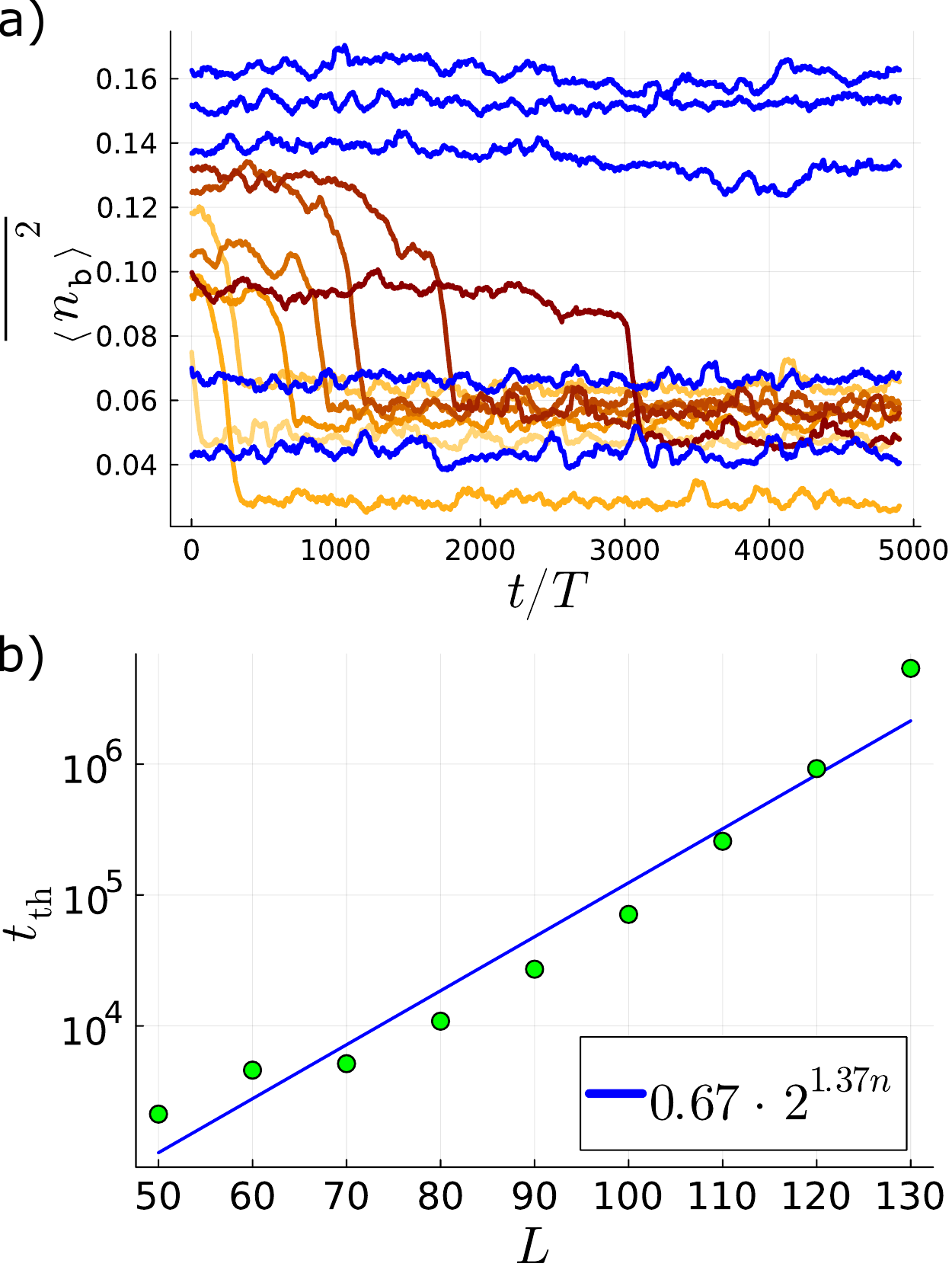}
		\caption{Relaxation of $n_\ttb$ under stochastic circuit dynamics for initial states with random $\ttb$ density waves. The initial states are chosen to be words of the form $w_1\ttb^nw_2 \ttb^{-n}w_3 \ttb^nw_4 \ttb^{-n}w_5$, where $n=L/10$ and the $w_{1,\dots,5}$ are random words containing only the characters $\{\tta, \tta^{-1}, \tte\}$. (a) Time evolution of observable $\overline{\expval{n_{\ttb}}^2}$ from Eq. (\ref{eq:observable}) for several independent runs with random initial states described above and the time averaging window $T=100$. The time at which the density wave collapses and the final equilibrium value of $\ob{\lan n_\ttb\ran^2}$ are seen to change significantly for different choices of initial state, with some runs thermalizing faster than $T$ (low-lying blue lines) and some runs not thermalizing within the whole displayed time window of $5000T$ (high-lying blue lines). 
			(b) Thermalization times $t_{\rm th}$ of random density waves, post-selected on initial states that exhibit a drop in $\ob{\lan n_\ttb\ran^2}$ of at least $75\%$ during the displayed time window. $t_{\rm th}$ is determined as the median across all runs (as opposed to the mean), due to the presence of a long tail in the distribution of $t_{\rm th}$, the statistics of which are shown for $L=80$ in the inset. $t_{\rm th}$ determined in this way is seen to scale exponentially or faster with $n$.}
		\label{fig:bs_b_wave_a_random}
	\end{figure}
	
	\section{Fragile fragmentation and the space complexity of the word problem}\label{sec:ff}
	
	Our discussion in the past few sections has focused on the way in which the time complexity of the word problem enters in the thermalization times of $\dyn_G$. In this section, we turn to the {\it space complexity} of the word problem. 
	As discussed in Sec.~\ref{ssec:word_problem}, the space complexity of the word problem is determined by the maximal amount of space required to map between two words $w,w'\in K_{g, L}$, as diagnosed by the expansion length function $\el_g(L)$ \eqref{explengthg}.  When the expansion length is large, transitioning from $w$ to $w'$ necessarily requires that $w$ first grow to be much larger than its original size before shrinking down to $w'$. When $\el_g(L) > L$, the dynamics thus lacks the spatial resources needed to connect all states which describe the same group element. In this situation, $\dyn_G$ cannot act ergodically in $K_{g,L}$, and thus the $K_{g,L}$ become {\it further} fragmented. Each fragment now contains words that can be reached from some reference word $w$ by derivations that do not involve intermediate words longer than $L$. We call this phenomenon {\it fragile fragmentation} in analogy with the notion of fragile topology in band theory \cite{po2018fragile}: $\dyn_G$ is said to exhibit fragile fragmentation if there are pairs of words $w, w'$ of length $L$ such that $\dyn_G$ on a system of length $L$ does not connect $\ket{w}$ and $\ket{w'}$, but $\dyn_G$ on a larger system of length $L'>L$ connects the ``padded'' words $\ket{w} \otimes \ket{\texttt{e}^{L'-L}}$ and $\ket{w'} \otimes \ket{\texttt{e}^{L'-L}}$.\footnote{For semigroups, we would replace $\texttt{e}$ with a character $0$ which freely moves past other characters.} A schematic illustration of this definition is given in Fig.~\ref{fig:ff_and_jamming}. 
	
	A simple example of this phenomenon that exists in higher dimensions is the jamming transition.  In jammed systems, an ensemble of particles with hard core repulsive interactions can exhibit a phase transition from a low density mobile phase to a high density jammed phase. The analog of fragmentation is the limited configuration space that particles can explore in the jammed phase. When the jammed particles are given more space, their density decreases, and when it drops below a critical value, the dynamics becomes ergodic. If the extra space is subsequently removed ergodicity may again be broken, but the system may find itself in a previously inaccessible microstate.  
	The models we discuss in this section are more drastic examples of this phenomenon: unlike the examples with hardcore particle models, the models we study exihibit jamming even in one dimension, where the analog of the critical jamming density can be polynomially or exponentially small in system size. 
	
	To understand when $\dyn_G$ exhibits fragile fragmentation, we need to compute the expansion length $\el(L)$.\footnote{Our notation in this discussion will assume that $G$ is a group, for which $\el(L)\triangleq \el_e(L) \gtrsim \el_g(L)$ for all $g$. For general semigroup dynamics, one would need to consider the $\el_g(L)$'s separately. } Doing so for a general group can be rather difficult, as $\el(L)$'s definition involves a rather complicated minimization problem. If however we already know the time complexity of the word problem---i.e. if we know the scaling of $\dehn(L)$---it is possible to place a lower bound on $\el(L)$ \cite{gersten2002filling}. Indeed, suppose a word $w\in K_{e, L}$ has an expansion length $\el(w)$, implying that $|w|$ is at most $\el(w)$ during any homotopy from $w$ to $\tte^L$. Then the expansion length of any derivation $D(w\sqig\tte)$ cannot exceed the total number of length-$\el(w)$ words; if it did, there would be at least one state which appeared multiple times in $D(w\sqig\tte)$---which implies that such a derivation cannot be of minimal length. Since the number of length $\el(w)$ words is $(2|S|+1)^{\el(w)}$, we thus have $\dehn(w) \leq (2|S|+1)^{\el(w)}$. By maximizing over all possible $w\in K_{e, L}$ and taking a logarithm, we obtain the general bound 
	\be\label{mtspacetimebound} \el(L) \geq \frac{\ln\(\dehn(L)\)}{\ln(2|S|+1)}.\ee 
	This bound is interesting in that it connects the spatial and temporal complexities of the word problem. 
	It also has the consequence that to find examples with additional ergodicity breaking, we need only find a group with super-exponential Dehn function. We will do this in Sec.~\ref{sec:iteratedbs}, but before doing so, we will warm up by understanding fragile fragmentation and the spatial complexity of the word problem in the simpler case of $\dyn_\bs$ dynamics. A general discussion of fragile fragmentation and its repercussions for thermalization will be given in Sec.~\ref{sec:generalff}.

	\begin{figure*}
		\centering
		\includegraphics[width=.9\tw]{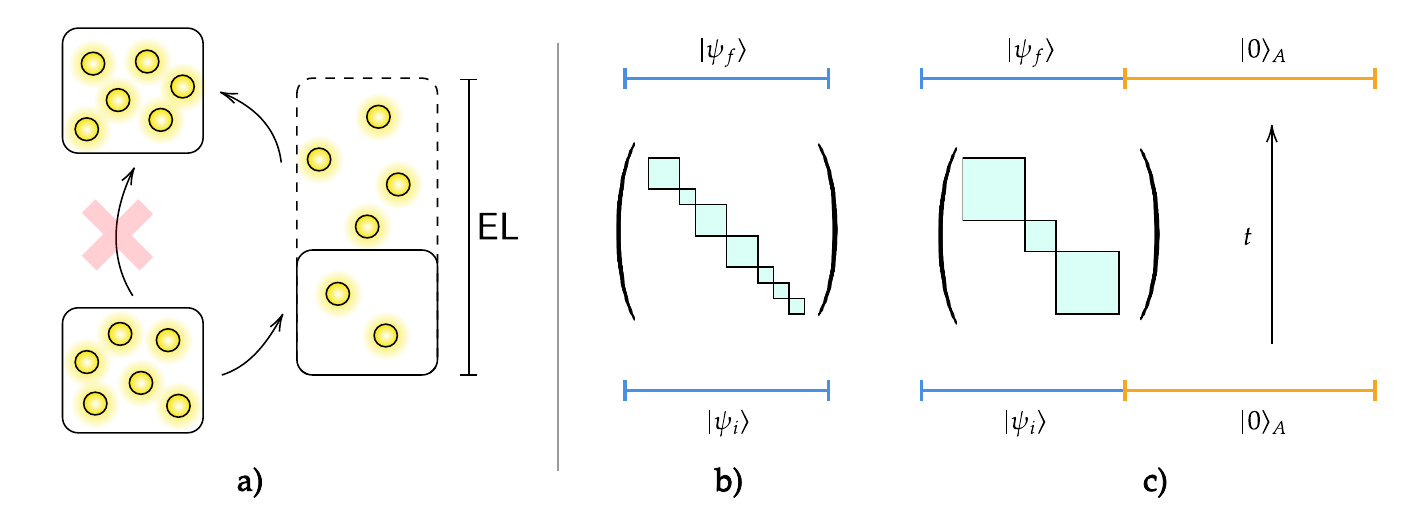} 
		\caption{\label{fig:ff_and_jamming} {\bf a)} A schematic of a jammed system, illustrated by densely-packed repulsively-interacting particles. The interactions and high density prevent the system from rearranging itself with the space available to it. If however the system size is increased for an intermediate period of time---allowing the particles to intermittently occupy a larger region of space before returning to their original volume---all different particle configurations can be reached. The amount by which the system size must be increased for ergodicity to be restored defines the expansion length $\el$. {\bf b)} Fragile fragmentation is the analogue of jamming in our dynamics. For a fixed system size, the dynamics is non-ergodic, but {\bf c)} (some degree of) ergodicity is restored when a reservoir of trivial ancillae $|0\ran_A$ is appended to the system. The analogue of returning to the original system size in the jamming example is played by projecting the ancillae onto their original state $\k0_A$ at the end of the time evolution.    } 
	\end{figure*}

	\subsection{Fragile fragmentation in $\bs$ dynamics}\label{sec:ffbs}
	
	\begin{figure*}
		\centering
		\includegraphics[width=\linewidth]{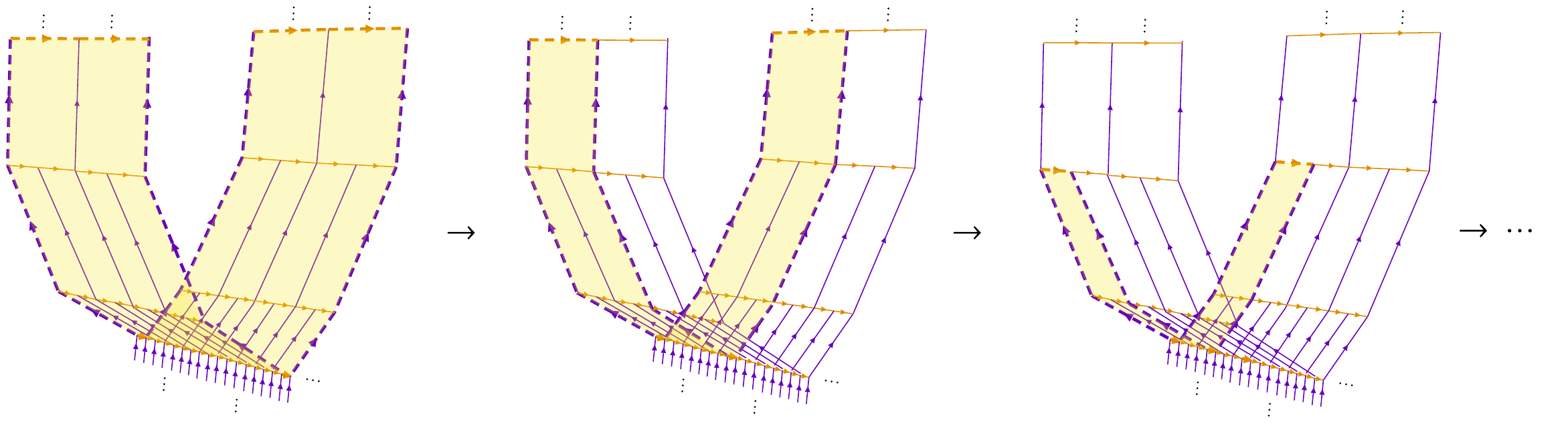}%
		\caption{\label{fig:word_homotopy} How a word $w_\mathrm{large}$ enclosing an exponentially large area can be deformed to the identity word without incurring an exponential amount of expansion. The deformation proceeds by making the loop defined by $w_\mathrm{large}$ narrower along the $\tta$ direction before shrinking the loop in the $\ttb$ direction.  }
	\end{figure*}
	
	We saw in Sec.~\ref{sec:bs} that the word $w_{\rm large}(n)$ encloses an area of $O(2^n)$ on the $\bs$ Cayley graph, leading to a word problem with exponentially large time complexity. We now address the space complexity of the word problem for the $\bs$ group. At first pass, it may seem that the spatial complexity is also exponentially large. Indeed, the naive homotopy mapping $w_\mathrm{large}(n)$ to the trivial word is to bring the two excursions $w_{\rm large}(n)$ makes along the $\ttb$ axis back ``down'' onto the $\tta$ axis. Doing this would cause $w_\mathrm{large}(n)$ to grow to a length of $\sim 2^n$ over the course of the homotopy. 
	
	However, it turns out that $w_{\rm large}(n)$ can be deformed in a way that does not require its size to significantly increase, via the process illustrated in Fig.~\ref{fig:word_homotopy}. As shown in the figure, instead of collapsing the excursions ``down'', we instead first make the loop ``skinnier'' by narrowing its extent along the $\tta$ axis before bringing the excursions ``down'' after they have become low-area enough. 
	
	It is then clear that the length of $w_\mathrm{large}(n)$ does not grow by too much---at least, not by more than a factor linear in $L$---during this homotopy. In App.~\ref{app:bs_details} we prove that at large $L$, 
	\be \label{elbs} \el(L) \sim (1+\a)L\ee 
	where $\a > 0$ is an $O(1)$ constant. Importantly, the fact that $\a>0$ means that 
	as $L$ approaches $4(n+1)$ from above, there will be an $L_* > 4(n+1)$ at which $|w_\mathrm{large}(n)| < L_*$ (so that $w_\mathrm{large}(n)$ can still be fit inside of the system), but $\el(w_\mathrm{large}(n)) > L_*$; in this regime, the density wave defined by $w_\mathrm{large}(n,L_*)$ cannot relax even at infinite times. 
	Therefore, there is a transition at a finite density of $\tte$s in the initial product state between a jammed regime (where homotopies are unable to contract) and an ergodic regime, leading to fragile fragmentation. Since the expansion length $\el(L)\propto L$, the severity of this jamming is comparable to that of conventional jammed systems in higher dimensions. 
	
	\begin{figure}
		\centering
		\includegraphics[scale=0.4]{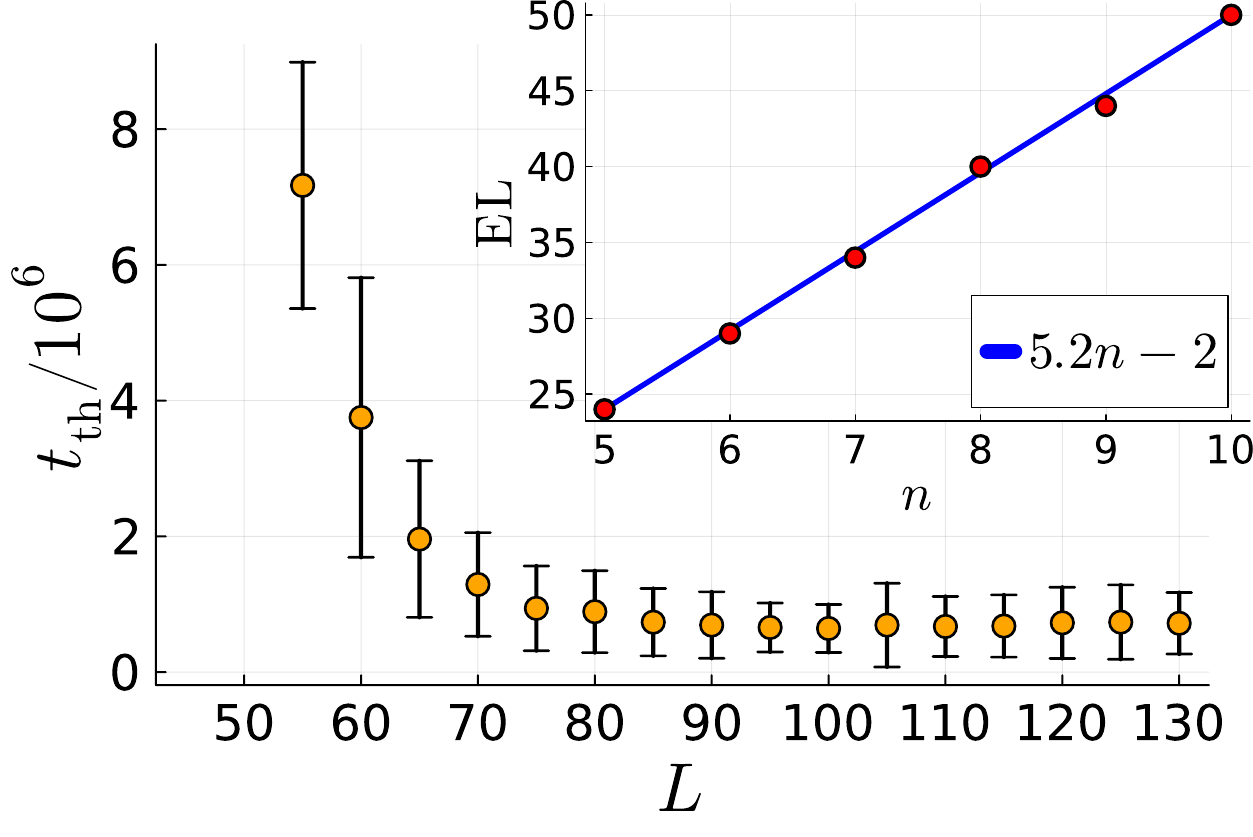}
		\caption{Thermalization time for stochastic $\bs$ dynamics as a function of $L$ for a system initialized in a product state $\k{w_\mathrm{large}(n,L)}$ with $n=10$ (so that the initial states for different $L$ differ only by the density of $\tte$s). The thermalization time diverges at some finite $L$, which defines the expansion length of the word $w_\mathrm{large}(n)$. Each thermalization time is obtained by averaging over several independent runs with the same parameters. For $L=50$, we observed a single thermalizing run, with thermalization time $t_{\rm th}\approx 5\cdot10^8$ (not shown in the plot). No thermalization was observed for $L<50$. (Inset) Expansion length of $w_\mathrm{large}(n)$ for different $n$. For the $\bs$ group, $\el$ is a linear function of $n$ (and therefore, of $L$).
		}
		\label{fig:bs_EL}
	\end{figure}
	
	We can identify $L_*$ numerically simply by decreasing $L$ until $|w_\mathrm{large}(n,L)\ran$ ceases to thermalize. The results of doing this are shown in Fig.~\ref{fig:bs_EL}. In this example, the length of the initial word (without identities) is $\abs{w_\mathrm{large}} = 4(n+1)=44$. We observe that thermalization time diverges as $L$ approaches $L_* \lesssim 50$ (for $L = 50$, we have observed only a single instance of thermalization at a very long time, $t_\mathrm{th} \approx 5\cdot 10^8$), confirming a nontrivial expansion length.

	\subsection{Exponential spatial complexity: Iterated Baumslag-Solitar}\label{sec:iteratedbs}
	
	We now present an example of group dynamics that exhibits fragile fragmentation with a zero density jamming transition: the iterated Baumslag-Solitar group \cite{clay2017office}.  This group is (loosely speaking) constructed by embedding a $\bs$ group inside of itself. We refer to it as $\bs^{(2)}$, and define it via the presentation
	\be \bs^{(2)} = \lan \tta,\ttb,\ttc  \, | \,  \tta\ttb = \ttb\tta^2 , \, \ttb\ttc = \ttc\ttb^2 \ran.\ee 
	$\dyn_{\bs^{(2)}}$ dynamics are thus most naturally realized in spin-3 chains with 3-local dynamics, whose local Hilbert space is obtained from that of the $\bs$ model by appending the two states $\{\k\ttc, \, \k{\ttc\inv}\}$. Note that like $\bs$, $\bs^{(2)}$ has a single conserved $U(1)$ charge given by the density of $\ttc$ generators, $n_\ttc \triangleq \sum_i n_{\ttc, i} = \sum_i (\proj\ttc_i - \proj{\ttc\inv}_i)$. 
	
	Several facts about $\bs^{(2)}$ (and related generalizations thereof) are proven in App.~\ref{app:iterated_bs}. The most important result is that the Dehn function of $\bs^{(2)}$ is a {\it super-exponential} function of $L$:
	\be \dehn(L) \sim 2^{2^{L}}.\ee 
	This can be intuited from the fact that the word $v(n) \equiv (\ttc^{-n} \ttb^{-1} \ttc^n) \tta (\ttc^{-n} \ttb \ttc^n)$ is equivalent to a doubly-exponentially long string of $\tta$s: 
	\bea  v(n) & = (\ttc^{-n} \ttb^{-1} \ttc^n) \tta (\ttc^{-n} \ttb \ttc^n) \\ & = \ttb^{-2^n} \tta \ttb^{2^n} \\ 
	& = \tta^{2^{2^n}},\eea 
	with the length of the RHS being doubly-exponentially larger than that of $|v(n)|$ (and where ``$=$'' in the above denotes equality as elements of $\bs^{(2)}$. By following the same strategy as in the construction of $w_{\rm large}(n)$, we can construct a word $w_{\rm huge}(n)\in K_{e, L}$ whose area grows doubly-exponentially with its length, namely 
	\be w_{\rm huge}(n) = \tta\inv v(n)\inv \tta v(n).\ee 
	Like with $\bs$, the slow dynamics of $\bs^{(2)}$ is manifest in the relaxation of the conserved charge $n_\ttc$, which from the scaling of $\dehn(L)$ we expect to relax with an effective momentum-dependent diffusion constant which is {\it doubly}-exponentially suppressed with $q$.

	The general bound \eqref{mtspacetimebound} implies that $\el(L) \gtrsim 2^L$. In App.~\ref{app:iterated_bs} we prove that this bound is in fact tight, so that  
	\be \el(L) \sim 2^L.\ee 
	Thus unlike $\bs$, there is no way to contract $w_{huge}(n)$ to the identity without it taking up an exponentially larger amount of space.\footnote{Essentially, the best one can do is to first eliminate the $\ttc$s and expand out $w_{huge}(n)$ into $w_{large}(2^n)$, at which point the word lies entirely in the $\bs$ subgroup generated by $\tta,\ttb$ and can be contracted without a further large expansion in length.} This means that the $n_\ttc$ density wave pattern present in the state $\k{w_{huge}(n)}\tp \k{\tte^m}$ will remain present even at {\it infinite} times unless $m \gtrsim 2^n$, i.e. unless the $n_\ttc$ density wave is exponentially dilute. 
	
	We now demonstrate this phenomenon numerically, using an extension of the analysis presented for $\bs$. A direct implementation of the bistochastic circuit \eqref{stochasticgate} is numerically rather expensive when investigating how $n_\ttc$ relaxes, due to the requirement of needing simulations to be run for times doubly exponential in characteristic scale of the $n_\ttc$ fluctuations under study. For this reason we will instead consider a {\it irreversible} modification of \eqref{stochasticgate}. We modify the dynamics so that $\ttc\ttc\inv$ and $\ttc\inv\ttc$ pairs can be annihilated, but not created. This means that our elementary stochastic gates $M_i$ contain terms like $\kb{\tte,\tte,\tte}{\ttc,\ttc\inv,\tte}_{i,i+1,i+2}$ but {\it not} the transpose thereof, with the quantity 
	\be n_{|\ttc|} \equiv \sum_i (\proj\ttc_i + \proj{\ttc\inv}_i)\ee  decreasing monotonically with time. 
	
	The merit of taking this approach is that the irreversible setting allows us to extract lower bounds on the relaxation time of $n_\ttc$ for the reversible setting. Our simulations are run by initializing the system in $|w_{\rm huge}(n,L)\ran$, a version of $w_{\rm huge}(n)$ padded with $L - |w_{\rm huge}(n)|$ $\tte$ characters at random locations, so that $|w_{\rm huge}(n,L)| = L$, and then tracking the time evolution of $n_{|\ttc|}$. The results are shown in Fig.~\ref{fig:itbs_EL} for different values of $n$, which are necessarily very small on account of the doubly-exponential growth of the Dehn function. In the main panel, we show the relaxation time of $|w_{\rm huge}(3,L)\ran$ as a function of $L$, defined by the time at which no $\ttc,\ttc\inv$ characters remain in the evolved word. The inset shows $\el(n)$, defined as the minimal value of $L$ for which relaxation (namely the reaching of a state containing no $\ttc,\ttc\inv$ characters) was observed to occur over 1000 runs of the dynamics. The extracted $\el(n)$ roughly conforms to our expectation of $\el(n) \sim 2^{\a n}$ for an $O(1)$ constant $\a$, although the long timescales required to observe thermalization mean that statistical errors are rather large, and with our current data we should not expect to obtain a perfectly exponential scaling. 
	
	\begin{figure}
		\centering
		\includegraphics[scale=0.4]{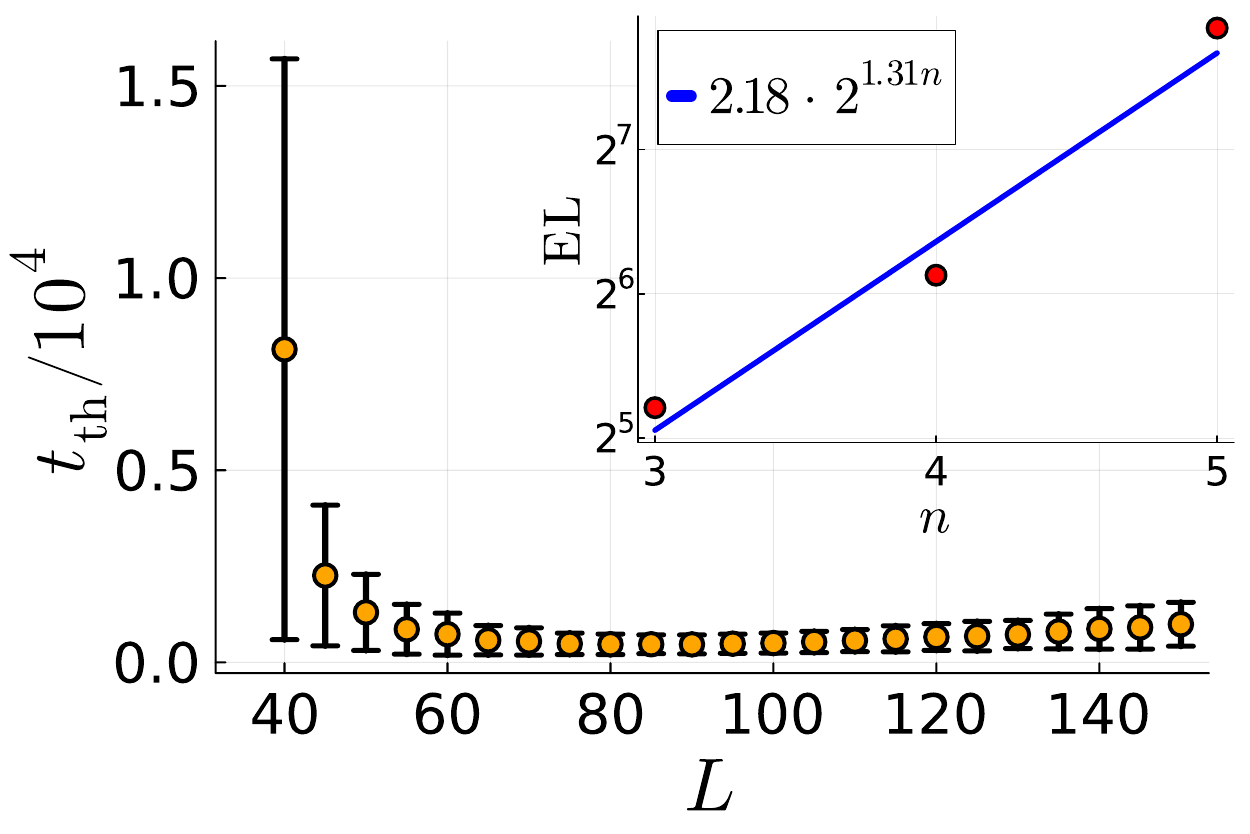}
		\caption{\label{fig:itbs_EL} Thermalization time as a function of $L$ for irreducible $\bs^{(2)}$ dynamics, where $\ttc \ttc^{-1}$, $\ttc^{-1} \ttc$ pairs can be annihilated but not created. The dynamics is initialized in the state $\ket{w_{huge}(n,L)}$ with $n=3$, and the system is considered to have thermalized when no $\ttc,\ttc\inv$ characters remain. {\it Inset:} Expansion length of $w_\mathrm{huge}(n)$ for different $n$, defined as the minimal system size for which thermalization was observed to occur over 1000 runs of the dynamics. An approximately exponential dependence on $n$ is observed.}
	\end{figure}

	\subsection{Fragile fragmentation: generalities} \label{sec:generalff}
	
	The $\bs^{(2)}$ example has a conserved density, so its failure to thermalize manifests itself as the freezing of a conserved density. In general, however, models exhibiting fragile fragmentation need not have conserved densities. Defining fragile fragmentation and finding reliable diagnostics for it in the general case are nontrivial tasks, which we address below. First, we provide a more precise definition of fragile fragmentation, to distinguish it from what we call {\it intrinsic fragmentation}. Second, we present a physical `decoupling' algorithm for detecting whether a system exhibits fragile fragmentation, given access to a large enough reservoir of ancillas. Third, we comment on the ways in which fragile fragmentation manifests itself in the dynamics of a thermalizing system. 
	
	\subsubsection{Defining fragile fragmentation}
	
	We begin by precisely defining the notions of intrinsic and fragile fragmentation in a general context (i.e., without reference to the word problem). For concreteness we specialize to quantum systems evolving under unitary dynamics, specified by a sequence of evolution operators $U_i(\mathcal{A}), i \in \mathbb{Z}_+$, acting on the system plus a collection of ancillas $\mca$, with each ancilla assumed for simplicity to have the same onsite Hilbert space as the system itself. We assume that $U_i$ form a uniform family of time evolution operators that can be defined for any number of ancillas.  Given any initial state $\ket{\psi}$ of the system, and a fixed reference state $\ket{0}_{\mathcal{A}}$ of the ancillas, we define an ensemble of states on the system as 
	\be
	\tau_{\mathcal{A}}(\psi) \equiv \{ \langle 0|_\mathcal{A} \left (U_i \ket{\psi} \otimes \ket{0}_\mathcal{A} \right) \}_{i = 1}^\infty.
	\ee
	In other words, $\tau_\mathcal{A}$ is the ensemble of pure states one gets by evolving the initial state $\ket{\psi} \otimes \ket{0}_\mathcal{A}$ for an arbitrary time, and postselecting on the ancillas being in the final state $\ket{0}_\mathcal{A}$. Note that in principle we could make $\tau_{\mathcal{A}}$ depend explicitly on the state of the ancillas.  However, for group dynamics the most natural choice for this state is $\k0_{\mathcal{A}} = \k\tte_{\mathcal{A}}$, and for semigroups an analogous state can be defined by augmenting the local Hilbert space with a character $\ket{\tte}$ which commutes with all characters.  We therefore content ourselves with studying fragmentation for this particular choice of ancilla state.  We define the Krylov sector of $\ket{\psi}$ extended to $\mathcal{A}$ as 
	\be K_\psi(\mathcal{A}) \equiv \mathrm{span}(\tau_\mathcal{A}(\psi)) \subset \mch_{\rm sys},\ee 
	where $\mch_{\rm sys}$ is the Hilbert space of the system (without ancillas). We furthermore define the {\it intrinsic Krylov sector} associated to $\k{\psi}$ as the limit
	\begin{equation}
		K_{\text{in}, \psi} = \lim_{|\mathcal{A}| \to \infty} K_{\psi}(\mathcal{A}).
	\end{equation}
	Under generic thermalizing Hamiltonian or unitary dynamics, the intrinsic Krylov sector of any $\ket{\psi}$ will be the entire Hilbert space, $K_{{\rm in},\psi} = \mch_{\rm sys}$. Intrinsic fragmentation occurs whenever this is not true, i.e. when there exist distinct initial states that do not mix under the dynamical rules even when the system is attached to an infinitely large bath (undergoing the same dynamics as the system). 
	
	When $|\mathcal{A}|$ is finite, each intrinsic Krylov sector may further shatter into many subsectors. This phenomenon (for $|\mathcal{A}|$ larger than the system) is what we have referred to above as {\it fragile fragmentation}. The expansion length associated with the dynamics is the minimal size of $\mathcal{A}$ below which additional subsectors form.
	
	\subsubsection{Probing fragile fragmentation}
	
	Our definition of fragile fragmentation above makes reference to postselection on the final state of the ancillas being $\ket{0}_\mathcal{A}$. The probability of postselection succeeding is clearly exponentially small in $|\mathcal{A}|$. We now present a more efficient algorithm for (i)~identifying whether a given system exhibits fragile fragmentation, and (ii)~constructing the subspace $K_\psi(\mathcal{A})$ associated with an initial state $\ket{\psi}$ given a maximum expansion length $L+|\mathcal{A}|$.  This procedure is more efficient than naively postselecting on the state of the ancillas in various regimes which we discuss below.
	
	The general algorithm proceeds as follows. We start with the state $\ket{\psi} \otimes \ket{0}_\mathcal{A}$, and evolve it under $\dyn$ acting on the system plus ancillas for some time $t_{\mathrm{th}}$. 
	After time $t_{\mathrm{th}}$, we repeat the following steps many times:
	\begin{enumerate} 
		\item Measure the last site of the system plus ancillas in the computational basis. 
		\item If the outcome is $\tte$, decouple this site from the rest of the system. Otherwise, leave the site coupled.
		\item Run the dynamics for a time $t_{\rm retherm}$ on the system plus remaining ancilla, and go to step 1. 
	\end{enumerate} 
	The iteration stops when all ancilla sites have been decoupled: we know this is always possible since the initial state was originally decoupled from the ancillas. On physical grounds we expect that the probability of getting outcome $0$ in step 1 is $O(1)$ at all times, but for our purposes it suffices for it to scale as  $1/{\rm poly}(L + |\mathcal{A}|)$. 
	
	We first discuss how this algorithm can be used to construct the subspace $K_\psi(\mathcal{A})$. To accomplish this, one sets $t_{\mathrm{th}}$ to be the maximum possible thermalization time for the system plus ancillas, i.e., $t_{\mathrm{th}} \sim \exp(L + |\mathcal{A}|)$. Any fragmentation that persists after $t_{\mathrm{th}}$ will persist to infinite time for the given spatial resources. One can take the re-thermalization time after a measurement, $t_{\rm retherm}$, to be much shorter (i.e., as a low-order polynomial in $L + |\mathcal{A}|$), as the measurement is a single-site perturbation to the equilibrated state, and it is not expected to take more than polynomial time to have a significant amplitude to be in $\k{\tte}$. 
	When the procedure terminates it yields a state $\ket{\psi'}$ that (by hypothesis) is in $K_\psi(\mathcal{A})$. After many runs, the ensemble of generated states spans $K_\psi(\mathcal{A})$.
	
	The procedure we described avoids the exponential overhead of postselection, but still incurs the exponential overhead of {\it mixing}. If we want to reconstruct a state with overlap on all states in $K_{\text{in}, \psi}$, this overhead cannot be avoided. Suppose, however, that we are not interested in full reconstruction of $K_{\text{in}, \psi}$, but just in the simpler task of showing that adding ancillas and removing them (as above) partially lifts the fragmentation of the original system. More generally, suppose that we have constructed $K_\psi(\mathcal{A}_0)$, and want to know if enlarging $\mathcal{A}_0$ to $\mathcal{A}_1$ enlarges the sector. We can run the initial equilibration step to a much shorter time than the full equilibration time. We stop when we have compressed back down to $\mathcal{A}_0$, and check if the resulting state is in $K_\psi(\mathcal{A}_0)$~\footnote{For word problems this can be done via a computational basis measurement; for systems with fragmentation in a non-product state basis, this measurement requires more care.}. Finding a single state that lies outside $K_\psi(\mathcal{A}_0)$ suffices to establish fragile fragmentation.  Thus, detecting fragile fragmentation can be accomplished without requiring the system to fully thermalize or requiring a large bath.
	
	Can this procedure be used to study other properties of the fragmentation?  One additional quantity that can be computed using this method are the geodesic lengths of group elements.  Recall that the geodesic length $|g|$ of an element $g\in G$ is the length of the shortest word which represents $g$, namely $|g| = \min\{ |w| \, : \, \vp(\k w) = g\}$.  To compute this, we repeat this sequential length reducing procedure until the system freezes.  Formally speaking, the system will freeze when $\mel{\texttt{e}}{\rho_{\text{end}}}{\texttt{e}} = 0$ for quantum dynamics (where $\r_{\rm end}$ is the reduced density matrix of the spin at the end of the system), or $p(\texttt{e}_{\text{end}}) = 0$ for classical dynamics, where $p(\cdot)$ denotes a marginal distribution for the last site of the chain. When this occurs, the system size has reached the minimum length needed to support a word in the Krylov sector, therefore yielding the geodesic length of the word.
	
	\subsubsection{Fragile fragmentation and thermalization} \label{ffandthm}
	
	We now discuss the unexpectedly subtle consequences of fragile fragmentation for the evolution of generic states under unitary dynamics that need not be time-independent or have any local conserved densities. To keep our discussion concrete, we will focus on $\dyn_G$ dynamics for some group $G$, although the diagnostics we will arrive at are much more general. As we will see, $G$ being a group (rather than just a semigroup) in fact makes fragile fragmentation particularly hard to detect locally; when $G$ is instead a semigroup, simpler diagnostics exist (see Sec.~\ref{sec:nongroup} for an example). 
	
	In a system exhibiting fragile fragmentation, a random word $w$ of length $L$ will contain many substrings $s$ that have expansion length greater than $L$; in particular, treating a particular substring $s$ as subsystem $A$ and $B = A^c$ as a bath of $\tte$'s, $\ket{s}_A \otimes \ket{\tte}_B$ exhibits fragile fragmentation so long as $|B| \ll \el(|s|)$. 
	In reality, the substring is nested in the system as $\ket{w} = \ket{w_L s w_R}$, but we still claim that the action of $\dyn_G$ can never map $\k w$ to $\k{w'} = \k{w_L s' w_R}$, where $s'$ is a word with expansion length asymptotically greater than $L$ and in a different fragment to $s$, but satisfies $\vp(s) = \vp(s')$. In particular, one might worry that the presence of the environment words $w_{L/R}$ can `catalyze' transitions of $s$, thereby sending $\k w$ to $\k {w'}$ despite $s$ and $s'$ living in different fragile sectors.
	
	A simple argument shows that such catalysis cannot parametrically change the expansion length of a word. Indeed, suppose that by contradiction catalysis can occur. We can append $w_L^{-1}$ to the left and $w_R^{-1}$ to the right, increasing the length of the system by less than $L$. %
	Then the sequence 
	\begin{align}
		\ket{\tte^{2|w_L|} s \tte^{2|w_R|}} &\leftrightarrow \ket{w_L^{-1} w_L s w_R w_R^{-1}} \nonumber\\
		&\leftrightarrow \ket{w_L^{-1} w_L s' w_R w_R^{-1}} \nonumber\\ &\leftrightarrow \ket{\tte^{2|w_L|} s'\tte^{2|w_R}}
	\end{align}
	is allowed by $\dyn_G$. Therefore the space complexity of the transition $s \leftrightarrow s'$ is at most $2L$. For groups with asymptotically superlinear expansion lengths this is a contradiction, and thus such a catalysis cannot occur. 
	
	To summarize, a random word contains large substrings that are frozen in some sense: if the initial state can be written as $\ket{w_L s w_R}$, time evolution under $\dyn_G$ will never produce $\ket{w_L s' w_R}$. An immediate consequence of this fact is that the time-evolved reduced density matrix for a region $A$ has $\langle s | \rho_A(t) | s' \rangle = 0$ at all times if $s,s'$ are in distinct fragile fragments. Furthermore, since both $s$ and $s'$ can be locally generated, through the transition 
	\begin{equation}
		\ket{\tte^{2|s|}} \leftrightarrow \ket {ss^{-1}},
	\end{equation}
	which only requires $2|s| \ll L$ of space, we in general expect $\langle s | \rho_A | s \rangle, \langle s' | \rho_A | s' \rangle$ to be both nonzero. 
	
	In conventional systems that exhibit the jamming transition, one can easily diagnose the jammed phase by computing local autocorrelation functions.  However, in contrast, fragile fragmentation is hard to detect in this way because the frozen substrings can slide around in the system and locally change their configuration (while remaining in the same fragile sector).  Thus, while one can write down a conserved quantity describing the frozen strings, such a quantity will generically be very nonlocal.  Nevertheless, the observation that the dynamics does not connect pairs of states like $s w_R$ and $s' w_R$ regardless of $w_R$ still allows one to construct a reasonable dynamical probe of fragile fragmentation. For any two words $w,w'$ in the same Krylov sector, consider the two-point correlation function 
	\begin{align} 
		C_{ww'}(t) &\equiv \mathrm{Tr}(W X_{ww'}(t)) \\
		&= 2\cdot \text{Re}\sum_{\alpha, \beta} \bra{w \alpha} \dyn_G^{\dagger}(t) \ketbra{w \beta}{w' \beta} \dyn_G(t) \ket{w \alpha} \nonumber
	\end{align} 
	where $W = \ket{w}\bra{w}$ and $X_{ww'}= \ket{w} \bra{w'} + \mathrm{h.c.}$
	are operators which can be fully supported in a subsystem of size $\max(|w|,|w'|)$. Suppose that $\el(w,w')$ is large, so that all derivations between $w$ and $w'$ require large spatial resources. Then, if $\dyn_G$ describes circuit dynamics, $C_{ww'}(t)$ is zero for small $t$ since if both $\bra{w \alpha} \dyn_G^{\dagger}(t) \ket{w \beta}$ and $\bra{w' \beta} \dyn_G(t) \ket{w \alpha}$ are nonzero, then one can transition from $\k{w \beta}$ to $\k{w' \beta}$ in time $2 t$.  Thus, the quantity $\dehn(w \beta, w' \beta)$ is expected to control when we expect this quantity to be nonzero\footnote{For Hamiltonian dynamics, this quantity is never strictly zero, but using similar arguments from Section~\ref{sec:therm_time_bounds} we expect this correlator to take a long time to reach its equilibrium value.}.  Furthermore, for systems where the spacetime bound is saturated (like iterated BS) or close to being saturated, this timescale is dictated by $\el(w,w')$ and can be at most $\sim \exp(\el(w,w'))$.

	Beyond this timescale, the system is able to undergo a large-scale rearrangement that connects $w$ and $w'$ (i.e. $w$ and $w'$ will no longer appear to live in disconnected fragile sectors). By measuring the onset timescale for $C_{ww'}(t)$ to become nonzero as a function of $|w|,|w'|$, one can diagnose the existence of and place bounds on large expansion lengths.  Alternatively, computing $C_{w,w'}(t)$ can be thought of in the following way: prepare an initial state $\rho_0$ where a subregion $R$ is in the pure state $\ket{w}$, time evolve this state to get $\rho$ and measure the expectation value of $X_{w,w'}$ in the reduced density matrix $\rho_R$.
	
	The above prescription suggests a heuristic way to determine the expansion length as follows (we leave various technical details of this proposal to future work).  Define $\Phi_R$ to be a dephasing channel acting on region $R$ of the system (decomposing the system to the form $ABR$ for convenience):
	\begin{equation}
		\Phi_R [\rho] \triangleq \sum_{\alpha_
			{A,B}, w_R, \beta_{A,B}} \rho_{\alpha_A w_R \alpha_B, \beta_A w_R \beta_B} \ketbra{\alpha_A w_R \alpha_B}{\beta_A w_R \beta_B}.
	\end{equation}
	Starting in the state $\rho_0$ described above, one can alternate between time evolving under $\dyn_G$ and applying $\Phi_R$ before measuring the expectation value of $X_{ww'}$.  Call $A$ the subsystem in which the state $\ket{w}$ is initially present.  If $\text{dist}(A, R) \gg \el(w,w')$, then repeatedly dephasing the system should not change the value of $C_{ww'}(t)$ by much.  However, if the dephasing channel is applied within a distance of $\el(w,w')$ from $A$, then we would expect a further suppression of $C_{ww'}(t)$ given that the dephasing eliminates many trajectories mapping $w$ to $w'$.  Finding the location of $R$ where one crosses over between these two behaviors would thus provide an estimate of $\el(w,w')$.
	
	The protocol discussed above is general but somewhat indirect. As we saw in Sec.~\ref{sec:iteratedbs}, in specific examples fragile fragmentation can have more direct and dramatic manifestations. In the next section we show that when the group property is violated, fragile fragmentation generally has easier-to-detect physical consequences: in these cases, there can sometimes be a transition where at small subsystem sizes reduced density matrices are generically full rank, while above a threshold size reduced density matrices have nontrivial kernels.

	\section{Semigroup examples}\label{sec:nongroup}
	
	In the explicit examples of $\dyn_G$ dynamics studied above, $G$ has been taken to have the structure of a group. The phenomena discussed so far are however not limited to models with a group structure; indeed all of the general results obtained in Secs.~\ref{sec:general_semigroup} and \ref{sec:therm_time_bounds} were are valid for any finitely presented semigroup. 
	For semigroups, however, the geometric perspective
	adopted above in the discussion of $\dyn_\bs$ is less useful\footnote{To understand why, one may envision some notion of a Cayley 2-complex which labels elements of the semigroup.  However, since there is no concept of an identity element, loops on the graph do not carry the same meaning that they do for groups.  Furthermore, paths or loops on the Cayley 2-complex are no longer freely deformable because the semigroup relations are not invariant under cyclic conjugation, due to the lack of inverses.}. 
	In this section, we introduce two new semigroup models which do not admit a group structure but which nevertheless have word problems exhibiting large time and space complexity, which we establish using combinatorial rather than geometric arguments. The first example has large time and small space complexity, and shares similarity with the $\bs$ model. The second example has both large time and large space complexity, is qualitatively unique to non-group based dynamics, and leads to a more direct characterization of fragile fragmentation than that provided by the general criterion of Sec.~\ref{ffandthm}. 
	
	These models are inspired by the Motzkin spin chain and, more broadly, Motzkin dynamics (see Refs.~\cite{Bravyi, movassagh}).  For the readers' convenience, we briefly summarize Motzkin dynamics.  The local state space includes an identity character $|0\ran$ (which plays the role of $\k\tte$ in our group-based models) as well as left and right parentheses $\k(, \k)$. The dynamics is engineered so that the ``nestedness'' of the parenthesis remains preserved, where nestedness is defined by the number of left parenthesis located to the left of matching right parenthesis; for example, under the dynamics the word `$()$' may evolve to `$()()$', `$(())$' or `$0$', but not to `$)($', `$(($', or `$))$'. These rules can be summarized formally by defining the {\it Motzkin semigroup} as
	\be {\sf Motz} = {\sf semi}\lan 0, (, )\, | \, (0 = 0(, \, )0 = 0), \, () = 00 \ran.\ee 
	Alternatively, we may define a height field $h_i$ which keeps track of the level of nestedness at site $i$, via 
	\be h_i = \sum_{j<i} (\proj{(}_j - \proj{)}_j).\ee 
	The dynamics then preserves both $h_L$ (the net difference between the number of $($ and $)$ parenthesis) and $\min_i h_i$ (which measures the extent of the nestedness). 
	
	\subsection{Star-Motzkin model: slow thermalization} \label{sec:starmotzkin}
	
	Our first example has a local state space which we will label by $\{(, ), 0, \ast\}$. As usual, all of what follows can be applied to Hamiltonian, random unitary, or classical stochastic dynamics. 
	The purpose of the extra character $\ast$ is to slow down the dynamics of the parenthesis; this is done by adding to the relations of ${\sf Motz}$ the relations 
	\begin{equation}\label{eq:starrule}
		( \ast \, 0 = \ast \ast (\hspace{0.5cm} 0 \, \ast ) =\, ) \ast \ast \hspace{0.5cm} 0 \, \ast =\ast \, 0.
	\end{equation}
	Thus when a parenthesis moves past a $\ast$ character (in a certain direction), the $\ast$ character is duplicated.
	The combination of these set of rules results in the dynamics which we will refer to as $\dyn_{\ast M}$.  One can readily see that $\dyn_{\ast M}$ exhibits Hilbert space fragmentation.  Indeed, if we ignore the $\ast$ character that was added, the Hamiltonian describes Motzkin dynamics, and thus already possesses fragmentation which cannot be described solely by a conserved parenthesis density.  The sectors of these dynamics are labelled by a sequence of closed parentheses followed by a sequence of open parentheses taking the form $)^m (^n$: this corresponds to the total parenthesis imbalance in the configuration.  When the $\ast$ character is added, a label for the Krylov sectors becomes
	\begin{equation}
		K_{\ell, m, n} = \, \, )^m \ast^{\ell} (^n
	\end{equation}
	where $\ell = 0,1,\cdots,O(L\cdot 2^{\max(m,n)})$.  It is clear that there are exponentially more sectors with the addition of the $\ast$ character, and the structure of fragmentation is thus richer.
	
	At some level $\dyn_{\ast M}$ resembles $\dyn_{\bs}$, as the $\ast$ duplicate every time they are moved past a parenthesis, in a way similar to the duplication of $\tta$s that occurs as they move past $\ttb$s in $\bs$.  However, there are two differences.  The first is that $\ast$ does not have a natural inverse.  The second is that the underlying Motzkin dynamics does not satisfy properties of a group: if the characters $($ and $)$ are identified with a generator and its inverse, then we necessarily must also allow the rule $() \leftrightarrow )($, which is absent in $\dyn_{\ast M}$.  Nevertheless, in App.~\ref{app:nongroup} we show that the word problems for these models exhibit the same scaling of spatial and temporal complexities as in $\dyn_\bs$: 
	\begin{enumerate}
		\item Within the sector $K_{q,0,0}$, the expansion length is linear, $\el(w,w') = O(L)$.
		\item Within the sector $K_{q,0,0}$, there exists words $w, w'$ such that $\dehn(w,w') = O(q)$.
	\end{enumerate}
	Since $q$ can grow to be exponentially large in $L$, this last fact implies slow dynamics.  In fact, one way to study this slow dynamics is to observe that atypical configurations of the $U(1)$ charge corresponding to the parentheses takes a very long time to thermalize.  Regarding the second result, we also provide a crisper characterization for the circumstances under which it takes a long time to transition between two words (see App.~\ref{app:nongroup}).
	
	\subsection{Chiral star-Motzkin model: fragile fragmentation}\label{sec:chstarmotzkin}
	
	We now present an example where the expansion length is exponentially large, implying fragile fragmentation.  In this example, we simply replace the $\ast$ character with a `chiral' version of the character, which we denote as $\triangleright$.  The rules for these new characters are similar to that of $\ast$, except that the rules are only activated when a $\triangleright$ character is adjacent to $)$.  More specifically, we replace the rules in~\eqref{eq:starrule} with
	\begin{equation}
		0 \, \triangleright \, ) = \,)\triangleright  \triangleright \hspace{0.75cm} 0 \, \triangleright = \triangleright \, 0
	\end{equation}
	Note that one could also add another chiral character $\triangleleft$ which only interacts with $)$ and commutes with $\triangleright$, but the necessary physics is already illustrated for $\triangleright$.  
	
	We first discuss the intrinsic Krylov sectors of the dynamics.  In particular, if we append the system is a large bath of $0$'s, then one can show that any configuration can be reduced to the canonical form:
	\begin{equation}
		R_{\vec{k}, \ell, m, n} = \,\,\, )^m \triangleright^{\ell} ( \, \triangleright^{k_1} \, ( \, \triangleright^{k_2}\, ( \, \triangleright^{k_3} \cdots ( \, \triangleright^{k_n}.
	\end{equation}
	Note that we can have a large number of $\triangleright$ characters locked between adjacent `(' because $\triangleright$ characters cannot tunnel past `(' characters.  As a result, we will label the Krylov sectors by the four indices $(\vec{k}, \ell, m, n)$ where $\dim \vec{k} = n$.
	
	Chirality of the $\triangleright$ character plays a crucial role in the large space complexity.  Define a {\it nest} as a collection of parentheses in the form $(((\cdots)))$.  In particular, $\triangleright$ characters embedded in a nest can exit the nest but cannot enter an adjacent one, due to the chirality constraint.  As a result, the transition from $\triangleright^k (((\cdots)))$ to $(((\cdots)))\triangleright^k$ (which subsequently allows $\triangleright^k$ to enter the nest) is not possible unless $(((\cdots)))$ is collapsed.  Collapsing $(((\cdots)))$ can however require exponentially large space if the nest $(((\cdots)))$ contains a large number of $\triangleright$ characters.  Thus, intuitively, an exponentially large bath is needed to unfreeze the system.
	
	In App.~\ref{app:nongroup}, we provide a more rigorous argument for when the expansion length connecting two words $w$ and $w'$ can be exponentially large, and also discuss an interesting consequence of the  fragile fragmentation in this model, which does not have an analog in group dynamics.  In particular, we argue that under unitary time evolution $\rho(t) = e^{-iHt} \rho(0) e^{iHt}$ where $\rho(0)$ is a product state, the subsystem density matrix $\rho_A(t) = \Tr_{A^c}(\rho(t))$ exhibit a {\it transition in its rank} as $|A|$ is increased.  When $|A| \ll \log L$ then $\rho_A(t)$ is of full rank and when $|A| \gg \log L$, then $\rho_A(t)$ is no longer of full rank.  Probing this property in a physically reasonable way is further discussed in App.~\ref{app:nongroup}
	
	\section{Generalizations to two dimensions: group loop models} \label{sec:loop_models}
	
	The discussion thus far has been restricted to 1D models of group dynamics.
	It is natural to wonder whether or not higher-dimensional models with similar behavior can be constructed, especially since the aforementioned phenomena are more prevalent in higher dimensions. 
	
	In this section, we discuss one 2D generalization of our group-based models that in some sense is the most faithful way of embedding the 1D group constraint in a two dimensional system, and which leads to a qualitatively new way of producing glassy dynamics and jamming in 2D. This proceeds by fixing a group $G$\footnote{As we will see shortly, our 2D models are simplest to write down when $G$ is a group rather than a semigroup, although this is not a fundamental restriction.} and considering loop models that possess one flavor of loop for each generator of $G$. Along any one-dimensional reference loop, one can associate a group element corresponding to the product of all of the generators corresponding to loops that the reference loop intersects. The dynamics is engineered so that this group element remains invariant under the dynamics. 
	
	This class of models can be viewed as a broad generalization of the construction in Refs.~\cite{stephen2022ergodicity, stahl2023topologically} (which studied dynamics) and Ref.~\cite{Balasubramanian2023} (which studied ground state properties), and we expect similar robustness of the Hilbert space fragmentation in these models. For brevity's sake we will discuss this construction at a high level---e.g. we will largely use continuum language in order to avoid the notational burden incurred by an explicit lattice description---and will defer a more comprehensive analysis to future work.

	Given a discrete group $G$ and a presentation thereof, the degrees of freedom in our 2D model are associated with directed loops labeled by generators of $G$. In the following, we describe how to place constraints on the dynamics of these loops to produce phenomenology similar to that present in our 1D models. 
	
	We start by observing that each directed reference path $\g$ through space (microscopically, along the lattice) can be associated with a group word $w(\g)$. This word is determined by the ordered labels of the loops that $\g$ intersects. Specifically, when proceeding along $\g$, each time $\g$ intersects a loop labeled by the generator $\ttg_i$, $w(\g)$ is i) multiplied by $\ttg_i$ if the local tangent vector $\vec{r}$ of $\g$ and the local tangent vector $\vec{\ell}$ of the loop can have a cross product $\vec{r} \times \vec{\ell}$ is parallel to $\uvz$; and ii) is multiplied by $\ttg_i\inv$ if $\vec r \times \vec\ell$ is anti-parallel to $\uvz$. 
	
	This observation means that isolated closed loops can be regarded as implementing trivial relations in the words associated to the paths which pass through them, a fact we illustrate pictorially as
	\be \igpfoc{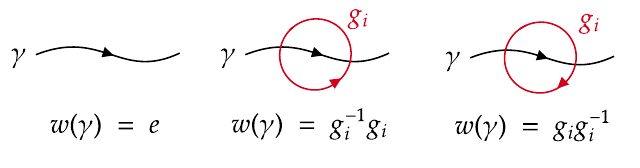} \ee 
	Therefore, we associate processes nucleating a loop with a free expansion ($\tte\tte \ra \ttg \ttg\inv$) and annihilating a loop with a free reduction ($\ttg\ttg\inv \ra \tte\tte$).  With more loops, a more general situation might look like the following: 
	\be \igptc{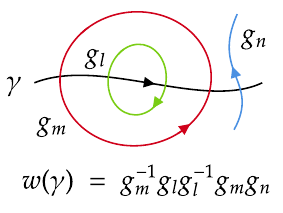}\ee 
	
	We now discuss how to implement relations of the group in terms of the loops.  Suppose the group presentation is indicated by 
	\begin{equation}
		G = \langle \ttg_1, \ttg_2, \cdots, \ttg_n | r_1, r_2, \cdots r_m\rangle,
	\end{equation}
	where each of the $r_i$ are words to be identified with the identity in $G$, and $|r_i| \leq 3$ for all $i$; this restriction on the length of the relations follows from the fact that any group (but {\it not} any semigroup) exhibits a finite presentation satisfying $|r_i| \leq 3$ (see App.~\ref{app:group_props} for the proof).  Writing $r_i = \ttg_m \ttg_n \ttg_{\ell}$, this corresponds to an object which we refer to as a {\it net}:
	\begin{equation}
		r_i = \vcenter{\hbox{\includegraphics[scale = 0.6]{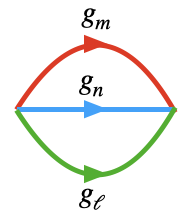}}}
	\end{equation}
	As another example, if $r_i = \ttg_m \ttg_n^{-1} \ttg_{\ell}$ (i.e. one of the generators is replaced with its inverse), then the net looks like
	\begin{equation}
		r_i = \vcenter{\hbox{\includegraphics[scale = 0.6]{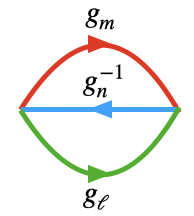}}}
	\end{equation}
	and so on.  Any loop configuration corresponding to one of the above nets can be created or destroyed without changing $\vp(w(\g))$ (the group element associated to $w(\g)$), since for any curve $\g$ that cuts across the net, creating or destroying the net simply corresponds to applying the appropriate relation $r_i$ at some point in the word $w(\gamma)$. 
	
	The dynamics we consider are generic dynamical processes which preserve $\vp(w(\g))$ for all closed curves $\g$, which may be viewed as an unusual type of gauge constraint. Thus the dynamics will include processes which nucleate and destroy nets associated to each $r_i$, and will also include processes where lines are moved, stretched, contracted, and where intersections of lines are moved. It will also include processes which attach a loop with an intersection point of other loops, like so:
	
	\begin{equation}
		\vcenter{\hbox{\includegraphics[scale = 0.55]{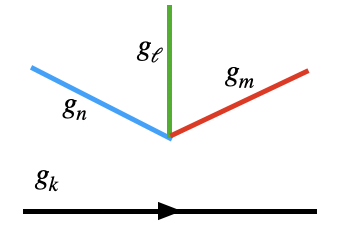}}} \longleftrightarrow \vcenter{\hbox{\includegraphics[scale = 0.55]{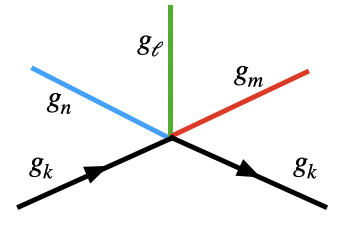}}}
	\end{equation}
	where an analogous deformation occurs for $\ttg_k$ replaced with $\ttg_k^{-1}$ (in which case the arrow is reversed).  
	To avoid problems on the lattice where an unbounded number of joins can occur (requiring unphysically large degrees), we only allow for a join if the degree of the intersection point is below a certain threshold, set by $\max_i |r_i|$.
	
	The dynamical processes described above  
	are sufficient to simulate the full group dynamics.  For example, suppose we want to determine whether the processes we wrote down suffice to simulate $\ttg_m \ttg_n = \ttg_{\ell}^{-1}$ (assuming $r_i = \ttg_m \ttg_n \ttg_{l} = \tte$ is a relation the group).  We can show that this is indeed the case by applying the following sequence of relations:
	\begin{widetext}
		\begin{equation}
			\vcenter{\hbox{\includegraphics[scale = 0.6]{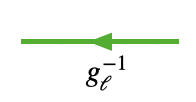}}} \to \vcenter{\hbox{\includegraphics[scale = 0.6]{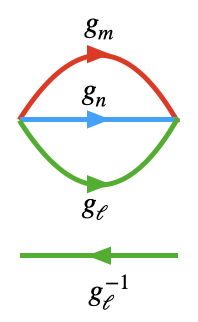}}} \to \vcenter{\hbox{\includegraphics[scale = 0.6]{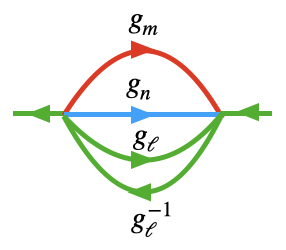}}} \to \vcenter{\hbox{\includegraphics[scale = 0.6]{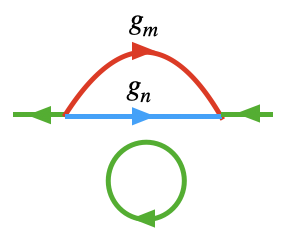}}} \to 
			\vcenter{\hbox{\includegraphics[scale = 0.6]{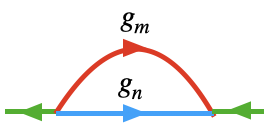}}}
		\end{equation}
	\end{widetext}
	where the first relation creates a net, the second relation corresponds to two joins, the third relation corresponds to two un-joins, and the last relation corresponds to a free reduction.  
	A similar derivation shows that cyclic conjugates of relations (such as $\ttg_{l}^{-1} r_i \ttg_{l} = \ttg_{l}^{-1} \ttg_m \ttg_n = \tte$) can similarly be simulated by the rules we have already discussed.
	
	To summarize, our 2D dynamics contains the following processes: 
	\begin{enumerate}
		\item Processes which deform loops in ways which do not create or destroy loop crossings,
		\item 
		Processes which nucleate and annihilate closed loops (performing free reductions and expansions),
		\item Processes which join a free loop with an intersection of loops\footnote{For dynamics suitably defined on the lattice, processes 1 and 3 are governed by essentially the same process.}, 
		\item For each relation $r_i$, a process which nucleates or annihilates an $r_i$-net. 
	\end{enumerate}
	
	The dynamical processes above were formulated for the case where $G$ is a group, but it is also possible to generalize to the case when $G$ is merely a semigroup, producing models reminiscent of the 2D generalizations of the Motzkin chain studied in Ref.~\cite{Balasubramanian2023}.
	Obtaining a more systematic understanding of the temporal and spatial complexity of the dynamics of these models would constitute an interesting avenue for future work.
	
	One final observation we make is that there is a simple way to map group loop models to what we call {\it tile models}.  Observe that a configuration of loops or nets splits the plane into a set of disjoint tiles (two adjacent tiles are joined by an edge which corresponds to a generator of the group).  Then, it turns out to be possible to label tiles with an element $g \in G$ such that if two tiles have labels $g$ and $h$ and share an edge corresponding to generator $\mathtt{k}$, then $g \mathtt{k} = h$.  For example, if $G = \mathbb{Z}$, each tile will be labelled by an integer and two adjacent tiles have labels which differ by $1$ (which is the generator of $\mathbb{Z}$).  This corresponds to a mapping to a height model, of which the tile model is a generalization of in the case of arbitrary $G$.  This mapping may play a role in understanding the nature of fragmentation in these models, which we leave to future work.

	\section{Discussion}\label{sec:disc}
	
	In this work, we have constructed a number of natural dynamical systems---with local few-body interactions---in which relaxation places anomalously expensive demands on a system's temporal and/or spatial resources. When the models have local conserved densities, the hydrodynamics of these densities is anomalous or frozen; even when conserved densities are absent, we have presented diagnostics for nonergodic behavior. 
	
	Our examples were all constructed in the context of models with intrinsic Hilbert space fragmentation. A natural question is whether the intrinsic fragmentation is essential to their physics.
	In our framework, dynamics without fragmentation is generated by finite presentations of the trivial group, which cannot have a superlinear Dehn function (see App.~\ref{app:primer}).
	This of course does not mean that the dynamics of models without fragmentation cannot be slow, but it does mean that any mechanism for slow thermalization must originate from something other than the Dehn function.

	Our results lend themselves to several natural extensions. Most naturally, the anomalous hydrodynamic relaxation we saw in the $\bs$ model can be extended to other groups with presentations that manifest a conservation law. Whether these groups give rise to new classes of hydrodynamic relaxation is an interesting question for future work (a family of such examples will be presented in~\cite{BLM2024}). Another interesting direction is to investigate the {\it ground states} of Hamiltonians that implement group dynamics. A natural class of frustration-free Hamiltonians can be read off from the transfer matrices of bistochastic Markov processes~\cite{RK}; their ground states are equal-weight superpositions of all the configurations in a sector, and their spectral gaps can be bounded by the Markov-chain gap. The tools developed here may be useful for undertaking a more detailed study of properties of these states.
	
	The family of models we considered is restricted in the sense that the dynamical constraints can be expressed in the computational basis, so that every computational-basis product state is in a definite dynamical sector of Hilbert space. 
	More generally, one can consider constraints 
	associated with a commuting set of projectors with entangled eigenstates. 
	In one-dimensional spin chains, such commuting projectors can be deformed into unentangled projectors by a short-depth unitary circuit. However, this process just corresponds to conjugating the Hamiltonians/unitaries we have explored with short-depth unitary circuits, and yields nothing qualitatively new. Getting something new in one dimension thus requires defining constraints using {\it non}-commuting projectors---which occurs in the (quantum-fragmented) Temperley-Lieb model~\cite{batchelor1991temperley,aufgebauer2010quantum,motrunich}---and hence pursuing this direction requires systematic characterization of such constraints.
	In higher dimensions however, there are sets of commuting projectors (like those of the toric code) whose simultaneous eigenstates are inherently long-range entangled. An interesting task for future work would be to extend our group-theoretic dynamical systems to these models, and explore the resulting entanglement dynamics. 
	
	Finally, it would be interesting to explore the stability of our results with respect to weak violations of the constraints. One way to do this is by breaking the constraints in an isolated region of space, or by bringing a thermal bath in contact with the system at its boundary. In the case of group dynamics, doing so destroys fragmentation, making the dynamics fully ergodic. It furthermore leads to all states being connected by at most $O(L^2)$ steps of the resulting dynamics, since the $O(L)$ elements of a given product state's geodesic word need to be moved at most an $O(L)$ distance to the constraint-free region, at which point they can be changed into the elements of any other geodesic word. It is however still possible for the dynamics in this case to have long (even exponentially long in $L$) thermalization times, due to bottlenecks in Fock space that arise from the finite spatial extent of the constraint-free region. A specific example of this is presented in \cite{hsf_impurities}.

	\section{Acknowledgements}
	
	S.B. and E.L. thank Soonwon Choi and Tim Riley for useful discussions.  S.B. thanks Aram Harrow for helpful discussions on the cutoff phenomenon.  A.K. is grateful to Claudio Chamon and Andrei Ruckenstein for related discussions. S.B. was supported by the National Science Foundation Graduate Research Fellowship under Grant No.~1745302. E.L. was supported by a Miller Research Fellowship. A.K. was supported in part by DOE Grant DE-FG02-06ER46316 and the
	Quantum Convergence Focused Research Program, funded by the Rafik B.
	Hariri Institute at Boston University. S.B., S.G., and E.L. performed some of this work at the Boulder Summer School, supported by NSF DMR-1560837.  This research was also supported in part by the National Science Foundation under Grant No. NSF PHY-1748958 and NSF PHY-2309135, the Heising-Simons Foundation, and the Simons Foundation (216179, LB). 

	\bibliography{refs}

\begin{thebibliography}{89}%
\makeatletter
\providecommand \@ifxundefined [1]{%
 \@ifx{#1\undefined}
}%
\providecommand \@ifnum [1]{%
 \ifnum #1\expandafter \@firstoftwo
 \else \expandafter \@secondoftwo
 \fi
}%
\providecommand \@ifx [1]{%
 \ifx #1\expandafter \@firstoftwo
 \else \expandafter \@secondoftwo
 \fi
}%
\providecommand \natexlab [1]{#1}%
\providecommand \enquote  [1]{``#1''}%
\providecommand \bibnamefont  [1]{#1}%
\providecommand \bibfnamefont [1]{#1}%
\providecommand \citenamefont [1]{#1}%
\providecommand \href@noop [0]{\@secondoftwo}%
\providecommand \href [0]{\begingroup \@sanitize@url \@href}%
\providecommand \@href[1]{\@@startlink{#1}\@@href}%
\providecommand \@@href[1]{\endgroup#1\@@endlink}%
\providecommand \@sanitize@url [0]{\catcode `\\12\catcode `\$12\catcode
  `\&12\catcode `\#12\catcode `\^12\catcode `\_12\catcode `\%12\relax}%
\providecommand \@@startlink[1]{}%
\providecommand \@@endlink[0]{}%
\providecommand \url  [0]{\begingroup\@sanitize@url \@url }%
\providecommand \@url [1]{\endgroup\@href {#1}{\urlprefix }}%
\providecommand \urlprefix  [0]{URL }%
\providecommand \Eprint [0]{\href }%
\providecommand \doibase [0]{https://doi.org/}%
\providecommand \selectlanguage [0]{\@gobble}%
\providecommand \bibinfo  [0]{\@secondoftwo}%
\providecommand \bibfield  [0]{\@secondoftwo}%
\providecommand \translation [1]{[#1]}%
\providecommand \BibitemOpen [0]{}%
\providecommand \bibitemStop [0]{}%
\providecommand \bibitemNoStop [0]{.\EOS\space}%
\providecommand \EOS [0]{\spacefactor3000\relax}%
\providecommand \BibitemShut  [1]{\csname bibitem#1\endcsname}%
\let\auto@bib@innerbib\@empty
\bibitem [{\citenamefont {Deutsch}(1991)}]{deutsch1991quantum}%
  \BibitemOpen
  \bibfield  {author} {\bibinfo {author} {\bibfnamefont {J.~M.}\ \bibnamefont
  {Deutsch}},\ }\bibfield  {title} {\bibinfo {title} {Quantum statistical
  mechanics in a closed system},\ }\href
  {https://doi.org/10.1103/PhysRevA.43.2046} {\bibfield  {journal} {\bibinfo
  {journal} {Phys. Rev. A}\ }\textbf {\bibinfo {volume} {43}},\ \bibinfo
  {pages} {2046} (\bibinfo {year} {1991})}\BibitemShut {NoStop}%
\bibitem [{\citenamefont {Srednicki}(1994)}]{srednicki1994chaos}%
  \BibitemOpen
  \bibfield  {author} {\bibinfo {author} {\bibfnamefont {M.}~\bibnamefont
  {Srednicki}},\ }\bibfield  {title} {\bibinfo {title} {Chaos and quantum
  thermalization},\ }\href {https://doi.org/10.1103/PhysRevE.50.888} {\bibfield
   {journal} {\bibinfo  {journal} {Phys. Rev. E}\ }\textbf {\bibinfo {volume}
  {50}},\ \bibinfo {pages} {888} (\bibinfo {year} {1994})}\BibitemShut
  {NoStop}%
\bibitem [{\citenamefont {Rigol}\ \emph {et~al.}(2008)\citenamefont {Rigol},
  \citenamefont {Dunjko},\ and\ \citenamefont
  {Olshanii}}]{rigol2008thermalization}%
  \BibitemOpen
  \bibfield  {author} {\bibinfo {author} {\bibfnamefont {M.}~\bibnamefont
  {Rigol}}, \bibinfo {author} {\bibfnamefont {V.}~\bibnamefont {Dunjko}},\ and\
  \bibinfo {author} {\bibfnamefont {M.}~\bibnamefont {Olshanii}},\ }\bibfield
  {title} {\bibinfo {title} {Thermalization and its mechanism for generic
  isolated quantum systems},\ }\href@noop {} {\bibfield  {journal} {\bibinfo
  {journal} {Nature}\ }\textbf {\bibinfo {volume} {452}},\ \bibinfo {pages}
  {854} (\bibinfo {year} {2008})}\BibitemShut {NoStop}%
\bibitem [{\citenamefont {Nandkishore}\ and\ \citenamefont
  {Huse}(2015)}]{nandkishore2015many-body}%
  \BibitemOpen
  \bibfield  {author} {\bibinfo {author} {\bibfnamefont {R.}~\bibnamefont
  {Nandkishore}}\ and\ \bibinfo {author} {\bibfnamefont {D.~A.}\ \bibnamefont
  {Huse}},\ }\bibfield  {title} {\bibinfo {title} {Many-body localization and
  thermalization in quantum statistical mechanics},\ }\href
  {https://doi.org/10.1146/annurev-conmatphys-031214-014726} {\bibfield
  {journal} {\bibinfo  {journal} {Annual Review of Condensed Matter Physics}\
  }\textbf {\bibinfo {volume} {6}},\ \bibinfo {pages} {15} (\bibinfo {year}
  {2015})}\BibitemShut {NoStop}%
\bibitem [{\citenamefont {D'Alessio}\ \emph {et~al.}(2016)\citenamefont
  {D'Alessio}, \citenamefont {Kafri}, \citenamefont {Polkovnikov},\ and\
  \citenamefont {Rigol}}]{dalessio2016from}%
  \BibitemOpen
  \bibfield  {author} {\bibinfo {author} {\bibfnamefont {L.}~\bibnamefont
  {D'Alessio}}, \bibinfo {author} {\bibfnamefont {Y.}~\bibnamefont {Kafri}},
  \bibinfo {author} {\bibfnamefont {A.}~\bibnamefont {Polkovnikov}},\ and\
  \bibinfo {author} {\bibfnamefont {M.}~\bibnamefont {Rigol}},\ }\bibfield
  {title} {\bibinfo {title} {From quantum chaos and eigenstate thermalization
  to statistical mechanics and thermodynamics},\ }\href
  {https://doi.org/10.1080/00018732.2016.1198134} {\bibfield  {journal}
  {\bibinfo  {journal} {Advances in Physics}\ }\textbf {\bibinfo {volume}
  {65}},\ \bibinfo {pages} {239} (\bibinfo {year} {2016})}\BibitemShut
  {NoStop}%
\bibitem [{\citenamefont {Alet}\ and\ \citenamefont
  {Laflorencie}(2018)}]{alet2018many-body}%
  \BibitemOpen
  \bibfield  {author} {\bibinfo {author} {\bibfnamefont {F.}~\bibnamefont
  {Alet}}\ and\ \bibinfo {author} {\bibfnamefont {N.}~\bibnamefont
  {Laflorencie}},\ }\bibfield  {title} {\bibinfo {title} {Many-body
  localization: An introduction and selected topics},\ }\href
  {https://doi.org/https://doi.org/10.1016/j.crhy.2018.03.003} {\bibfield
  {journal} {\bibinfo  {journal} {Comptes Rendus Physique}\ }\textbf {\bibinfo
  {volume} {19}},\ \bibinfo {pages} {498} (\bibinfo {year} {2018})},\ \bibinfo
  {note} {quantum simulation / Simulation quantique}\BibitemShut {NoStop}%
\bibitem [{\citenamefont {Abanin}\ \emph {et~al.}(2019)\citenamefont {Abanin},
  \citenamefont {Altman}, \citenamefont {Bloch},\ and\ \citenamefont
  {Serbyn}}]{abanin2019colloquium}%
  \BibitemOpen
  \bibfield  {author} {\bibinfo {author} {\bibfnamefont {D.~A.}\ \bibnamefont
  {Abanin}}, \bibinfo {author} {\bibfnamefont {E.}~\bibnamefont {Altman}},
  \bibinfo {author} {\bibfnamefont {I.}~\bibnamefont {Bloch}},\ and\ \bibinfo
  {author} {\bibfnamefont {M.}~\bibnamefont {Serbyn}},\ }\bibfield  {title}
  {\bibinfo {title} {Colloquium: Many-body localization, thermalization, and
  entanglement},\ }\href {https://doi.org/10.1103/RevModPhys.91.021001}
  {\bibfield  {journal} {\bibinfo  {journal} {Rev. Mod. Phys.}\ }\textbf
  {\bibinfo {volume} {91}},\ \bibinfo {pages} {021001} (\bibinfo {year}
  {2019})}\BibitemShut {NoStop}%
\bibitem [{\citenamefont {Schiulaz}\ and\ \citenamefont
  {M{\"u}ller}(2014)}]{schiulaz2014ideal}%
  \BibitemOpen
  \bibfield  {author} {\bibinfo {author} {\bibfnamefont {M.}~\bibnamefont
  {Schiulaz}}\ and\ \bibinfo {author} {\bibfnamefont {M.}~\bibnamefont
  {M{\"u}ller}},\ }\bibfield  {title} {\bibinfo {title} {Ideal quantum glass
  transitions: Many-body localization without quenched disorder},\ }in\
  \href@noop {} {\emph {\bibinfo {booktitle} {AIP Conference Proceedings}}},\
  Vol.\ \bibinfo {volume} {1610}\ (\bibinfo {organization} {American Institute
  of Physics},\ \bibinfo {year} {2014})\ pp.\ \bibinfo {pages}
  {11--23}\BibitemShut {NoStop}%
\bibitem [{\citenamefont {Grover}\ and\ \citenamefont
  {Fisher}(2014)}]{grover2014quantum}%
  \BibitemOpen
  \bibfield  {author} {\bibinfo {author} {\bibfnamefont {T.}~\bibnamefont
  {Grover}}\ and\ \bibinfo {author} {\bibfnamefont {M.~P.}\ \bibnamefont
  {Fisher}},\ }\bibfield  {title} {\bibinfo {title} {Quantum disentangled
  liquids},\ }\href@noop {} {\bibfield  {journal} {\bibinfo  {journal} {Journal
  of Statistical Mechanics: Theory and Experiment}\ }\textbf {\bibinfo {volume}
  {2014}},\ \bibinfo {pages} {P10010} (\bibinfo {year} {2014})}\BibitemShut
  {NoStop}%
\bibitem [{\citenamefont {Yao}\ \emph {et~al.}(2016)\citenamefont {Yao},
  \citenamefont {Laumann}, \citenamefont {Cirac}, \citenamefont {Lukin},\ and\
  \citenamefont {Moore}}]{yao2016quasi}%
  \BibitemOpen
  \bibfield  {author} {\bibinfo {author} {\bibfnamefont {N.}~\bibnamefont
  {Yao}}, \bibinfo {author} {\bibfnamefont {C.}~\bibnamefont {Laumann}},
  \bibinfo {author} {\bibfnamefont {J.~I.}\ \bibnamefont {Cirac}}, \bibinfo
  {author} {\bibfnamefont {M.~D.}\ \bibnamefont {Lukin}},\ and\ \bibinfo
  {author} {\bibfnamefont {J.}~\bibnamefont {Moore}},\ }\bibfield  {title}
  {\bibinfo {title} {Quasi-many-body localization in translation-invariant
  systems},\ }\href@noop {} {\bibfield  {journal} {\bibinfo  {journal}
  {Physical review letters}\ }\textbf {\bibinfo {volume} {117}},\ \bibinfo
  {pages} {240601} (\bibinfo {year} {2016})}\BibitemShut {NoStop}%
\bibitem [{\citenamefont {Retore}(2022)}]{retore2022introduction}%
  \BibitemOpen
  \bibfield  {author} {\bibinfo {author} {\bibfnamefont {A.~L.}\ \bibnamefont
  {Retore}},\ }\bibfield  {title} {\bibinfo {title} {Introduction to classical
  and quantum integrability},\ }\href
  {https://doi.org/10.1088/1751-8121/ac5a8e} {\bibfield  {journal} {\bibinfo
  {journal} {Journal of Physics A: Mathematical and Theoretical}\ }\textbf
  {\bibinfo {volume} {55}},\ \bibinfo {pages} {173001} (\bibinfo {year}
  {2022})}\BibitemShut {NoStop}%
\bibitem [{\citenamefont {Shiraishi}\ and\ \citenamefont
  {Mori}(2017)}]{shiraishi2017systematic}%
  \BibitemOpen
  \bibfield  {author} {\bibinfo {author} {\bibfnamefont {N.}~\bibnamefont
  {Shiraishi}}\ and\ \bibinfo {author} {\bibfnamefont {T.}~\bibnamefont
  {Mori}},\ }\bibfield  {title} {\bibinfo {title} {Systematic construction of
  counterexamples to the eigenstate thermalization hypothesis},\ }\href
  {https://doi.org/10.1103/PhysRevLett.119.030601} {\bibfield  {journal}
  {\bibinfo  {journal} {Phys. Rev. Lett.}\ }\textbf {\bibinfo {volume} {119}},\
  \bibinfo {pages} {030601} (\bibinfo {year} {2017})}\BibitemShut {NoStop}%
\bibitem [{\citenamefont {Bernien}\ \emph {et~al.}(2017)\citenamefont
  {Bernien}, \citenamefont {Schwartz}, \citenamefont {Keesling}, \citenamefont
  {Levine}, \citenamefont {Omran}, \citenamefont {Pichler}, \citenamefont
  {Choi}, \citenamefont {Zibrov}, \citenamefont {Endres}, \citenamefont
  {Greiner} \emph {et~al.}}]{bernien2017probing}%
  \BibitemOpen
  \bibfield  {author} {\bibinfo {author} {\bibfnamefont {H.}~\bibnamefont
  {Bernien}}, \bibinfo {author} {\bibfnamefont {S.}~\bibnamefont {Schwartz}},
  \bibinfo {author} {\bibfnamefont {A.}~\bibnamefont {Keesling}}, \bibinfo
  {author} {\bibfnamefont {H.}~\bibnamefont {Levine}}, \bibinfo {author}
  {\bibfnamefont {A.}~\bibnamefont {Omran}}, \bibinfo {author} {\bibfnamefont
  {H.}~\bibnamefont {Pichler}}, \bibinfo {author} {\bibfnamefont
  {S.}~\bibnamefont {Choi}}, \bibinfo {author} {\bibfnamefont {A.~S.}\
  \bibnamefont {Zibrov}}, \bibinfo {author} {\bibfnamefont {M.}~\bibnamefont
  {Endres}}, \bibinfo {author} {\bibfnamefont {M.}~\bibnamefont {Greiner}},
  \emph {et~al.},\ }\bibfield  {title} {\bibinfo {title} {Probing many-body
  dynamics on a 51-atom quantum simulator},\ }\href@noop {} {\bibfield
  {journal} {\bibinfo  {journal} {Nature}\ }\textbf {\bibinfo {volume} {551}},\
  \bibinfo {pages} {579} (\bibinfo {year} {2017})}\BibitemShut {NoStop}%
\bibitem [{\citenamefont {Turner}\ \emph
  {et~al.}(2018{\natexlab{a}})\citenamefont {Turner}, \citenamefont
  {Michailidis}, \citenamefont {Abanin}, \citenamefont {Serbyn},\ and\
  \citenamefont {Papi\ifmmode~\acute{c}\else \'{c}\fi{}}}]{turner2018quantum}%
  \BibitemOpen
  \bibfield  {author} {\bibinfo {author} {\bibfnamefont {C.~J.}\ \bibnamefont
  {Turner}}, \bibinfo {author} {\bibfnamefont {A.~A.}\ \bibnamefont
  {Michailidis}}, \bibinfo {author} {\bibfnamefont {D.~A.}\ \bibnamefont
  {Abanin}}, \bibinfo {author} {\bibfnamefont {M.}~\bibnamefont {Serbyn}},\
  and\ \bibinfo {author} {\bibfnamefont {Z.}~\bibnamefont
  {Papi\ifmmode~\acute{c}\else \'{c}\fi{}}},\ }\bibfield  {title} {\bibinfo
  {title} {Quantum scarred eigenstates in a rydberg atom chain: Entanglement,
  breakdown of thermalization, and stability to perturbations},\ }\href
  {https://doi.org/10.1103/PhysRevB.98.155134} {\bibfield  {journal} {\bibinfo
  {journal} {Phys. Rev. B}\ }\textbf {\bibinfo {volume} {98}},\ \bibinfo
  {pages} {155134} (\bibinfo {year} {2018}{\natexlab{a}})}\BibitemShut
  {NoStop}%
\bibitem [{\citenamefont {Turner}\ \emph
  {et~al.}(2018{\natexlab{b}})\citenamefont {Turner}, \citenamefont
  {Michailidis}, \citenamefont {Abanin}, \citenamefont {Serbyn},\ and\
  \citenamefont {Papi{\'c}}}]{turner2018weak}%
  \BibitemOpen
  \bibfield  {author} {\bibinfo {author} {\bibfnamefont {C.~J.}\ \bibnamefont
  {Turner}}, \bibinfo {author} {\bibfnamefont {A.~A.}\ \bibnamefont
  {Michailidis}}, \bibinfo {author} {\bibfnamefont {D.~A.}\ \bibnamefont
  {Abanin}}, \bibinfo {author} {\bibfnamefont {M.}~\bibnamefont {Serbyn}},\
  and\ \bibinfo {author} {\bibfnamefont {Z.}~\bibnamefont {Papi{\'c}}},\
  }\bibfield  {title} {\bibinfo {title} {Weak ergodicity breaking from quantum
  many-body scars},\ }\href@noop {} {\bibfield  {journal} {\bibinfo  {journal}
  {Nature Physics}\ }\textbf {\bibinfo {volume} {14}},\ \bibinfo {pages} {745}
  (\bibinfo {year} {2018}{\natexlab{b}})}\BibitemShut {NoStop}%
\bibitem [{\citenamefont {Moudgalya}\ \emph {et~al.}(2018)\citenamefont
  {Moudgalya}, \citenamefont {Rachel}, \citenamefont {Bernevig},\ and\
  \citenamefont {Regnault}}]{moudgalya2018exact}%
  \BibitemOpen
  \bibfield  {author} {\bibinfo {author} {\bibfnamefont {S.}~\bibnamefont
  {Moudgalya}}, \bibinfo {author} {\bibfnamefont {S.}~\bibnamefont {Rachel}},
  \bibinfo {author} {\bibfnamefont {B.~A.}\ \bibnamefont {Bernevig}},\ and\
  \bibinfo {author} {\bibfnamefont {N.}~\bibnamefont {Regnault}},\ }\bibfield
  {title} {\bibinfo {title} {Exact excited states of nonintegrable models},\
  }\href {https://doi.org/10.1103/PhysRevB.98.235155} {\bibfield  {journal}
  {\bibinfo  {journal} {Phys. Rev. B}\ }\textbf {\bibinfo {volume} {98}},\
  \bibinfo {pages} {235155} (\bibinfo {year} {2018})}\BibitemShut {NoStop}%
\bibitem [{\citenamefont {Lin}\ and\ \citenamefont
  {Motrunich}(2019)}]{lin2019exact}%
  \BibitemOpen
  \bibfield  {author} {\bibinfo {author} {\bibfnamefont {C.-J.}\ \bibnamefont
  {Lin}}\ and\ \bibinfo {author} {\bibfnamefont {O.~I.}\ \bibnamefont
  {Motrunich}},\ }\bibfield  {title} {\bibinfo {title} {Exact quantum many-body
  scar states in the rydberg-blockaded atom chain},\ }\href
  {https://doi.org/10.1103/PhysRevLett.122.173401} {\bibfield  {journal}
  {\bibinfo  {journal} {Phys. Rev. Lett.}\ }\textbf {\bibinfo {volume} {122}},\
  \bibinfo {pages} {173401} (\bibinfo {year} {2019})}\BibitemShut {NoStop}%
\bibitem [{\citenamefont {Schecter}\ and\ \citenamefont
  {Iadecola}(2019)}]{schechter2019weak}%
  \BibitemOpen
  \bibfield  {author} {\bibinfo {author} {\bibfnamefont {M.}~\bibnamefont
  {Schecter}}\ and\ \bibinfo {author} {\bibfnamefont {T.}~\bibnamefont
  {Iadecola}},\ }\bibfield  {title} {\bibinfo {title} {Weak ergodicity breaking
  and quantum many-body scars in spin-1 $xy$ magnets},\ }\href
  {https://doi.org/10.1103/PhysRevLett.123.147201} {\bibfield  {journal}
  {\bibinfo  {journal} {Phys. Rev. Lett.}\ }\textbf {\bibinfo {volume} {123}},\
  \bibinfo {pages} {147201} (\bibinfo {year} {2019})}\BibitemShut {NoStop}%
\bibitem [{\citenamefont {Moudgalya}\ \emph
  {et~al.}(2022{\natexlab{a}})\citenamefont {Moudgalya}, \citenamefont
  {Bernevig},\ and\ \citenamefont {Regnault}}]{moudgalya2022quantum}%
  \BibitemOpen
  \bibfield  {author} {\bibinfo {author} {\bibfnamefont {S.}~\bibnamefont
  {Moudgalya}}, \bibinfo {author} {\bibfnamefont {B.~A.}\ \bibnamefont
  {Bernevig}},\ and\ \bibinfo {author} {\bibfnamefont {N.}~\bibnamefont
  {Regnault}},\ }\bibfield  {title} {\bibinfo {title} {Quantum many-body scars
  and hilbert space fragmentation: a review of exact results},\ }\href
  {https://doi.org/10.1088/1361-6633/ac73a0} {\bibfield  {journal} {\bibinfo
  {journal} {Reports on Progress in Physics}\ }\textbf {\bibinfo {volume}
  {85}},\ \bibinfo {pages} {086501} (\bibinfo {year}
  {2022}{\natexlab{a}})}\BibitemShut {NoStop}%
\bibitem [{\citenamefont {Chandran}\ \emph {et~al.}(2023)\citenamefont
  {Chandran}, \citenamefont {Iadecola}, \citenamefont {Khemani},\ and\
  \citenamefont {Moessner}}]{chandran2023quantum}%
  \BibitemOpen
  \bibfield  {author} {\bibinfo {author} {\bibfnamefont {A.}~\bibnamefont
  {Chandran}}, \bibinfo {author} {\bibfnamefont {T.}~\bibnamefont {Iadecola}},
  \bibinfo {author} {\bibfnamefont {V.}~\bibnamefont {Khemani}},\ and\ \bibinfo
  {author} {\bibfnamefont {R.}~\bibnamefont {Moessner}},\ }\bibfield  {title}
  {\bibinfo {title} {Quantum many-body scars: A quasiparticle perspective},\
  }\href {https://doi.org/10.1146/annurev-conmatphys-031620-101617} {\bibfield
  {journal} {\bibinfo  {journal} {Annual Review of Condensed Matter Physics}\
  }\textbf {\bibinfo {volume} {14}},\ \bibinfo {pages} {443} (\bibinfo {year}
  {2023})},\ \Eprint
  {https://arxiv.org/abs/https://doi.org/10.1146/annurev-conmatphys-031620-101617}
  {https://doi.org/10.1146/annurev-conmatphys-031620-101617} \BibitemShut
  {NoStop}%
\bibitem [{\citenamefont {Sala}\ \emph {et~al.}(2020)\citenamefont {Sala},
  \citenamefont {Rakovszky}, \citenamefont {Verresen}, \citenamefont {Knap},\
  and\ \citenamefont {Pollmann}}]{sala2020ergodicity}%
  \BibitemOpen
  \bibfield  {author} {\bibinfo {author} {\bibfnamefont {P.}~\bibnamefont
  {Sala}}, \bibinfo {author} {\bibfnamefont {T.}~\bibnamefont {Rakovszky}},
  \bibinfo {author} {\bibfnamefont {R.}~\bibnamefont {Verresen}}, \bibinfo
  {author} {\bibfnamefont {M.}~\bibnamefont {Knap}},\ and\ \bibinfo {author}
  {\bibfnamefont {F.}~\bibnamefont {Pollmann}},\ }\bibfield  {title} {\bibinfo
  {title} {Ergodicity breaking arising from hilbert space fragmentation in
  dipole-conserving hamiltonians},\ }\href
  {https://doi.org/10.1103/PhysRevX.10.011047} {\bibfield  {journal} {\bibinfo
  {journal} {Phys. Rev. X}\ }\textbf {\bibinfo {volume} {10}},\ \bibinfo
  {pages} {011047} (\bibinfo {year} {2020})}\BibitemShut {NoStop}%
\bibitem [{\citenamefont {Khemani}\ \emph {et~al.}(2020)\citenamefont
  {Khemani}, \citenamefont {Hermele},\ and\ \citenamefont
  {Nandkishore}}]{khemani2020localization}%
  \BibitemOpen
  \bibfield  {author} {\bibinfo {author} {\bibfnamefont {V.}~\bibnamefont
  {Khemani}}, \bibinfo {author} {\bibfnamefont {M.}~\bibnamefont {Hermele}},\
  and\ \bibinfo {author} {\bibfnamefont {R.}~\bibnamefont {Nandkishore}},\
  }\bibfield  {title} {\bibinfo {title} {Localization from hilbert space
  shattering: From theory to physical realizations},\ }\href
  {https://doi.org/10.1103/PhysRevB.101.174204} {\bibfield  {journal} {\bibinfo
   {journal} {Phys. Rev. B}\ }\textbf {\bibinfo {volume} {101}},\ \bibinfo
  {pages} {174204} (\bibinfo {year} {2020})}\BibitemShut {NoStop}%
\bibitem [{\citenamefont {Moudgalya}\ and\ \citenamefont
  {Motrunich}(2022)}]{motrunich}%
  \BibitemOpen
  \bibfield  {author} {\bibinfo {author} {\bibfnamefont {S.}~\bibnamefont
  {Moudgalya}}\ and\ \bibinfo {author} {\bibfnamefont {O.~I.}\ \bibnamefont
  {Motrunich}},\ }\bibfield  {title} {\bibinfo {title} {Hilbert space
  fragmentation and commutant algebras},\ }\href
  {https://doi.org/10.1103/PhysRevX.12.011050} {\bibfield  {journal} {\bibinfo
  {journal} {Phys. Rev. X}\ }\textbf {\bibinfo {volume} {12}},\ \bibinfo
  {pages} {011050} (\bibinfo {year} {2022})}\BibitemShut {NoStop}%
\bibitem [{\citenamefont {Pancotti}\ \emph {et~al.}(2020)\citenamefont
  {Pancotti}, \citenamefont {Giudice}, \citenamefont {Cirac}, \citenamefont
  {Garrahan},\ and\ \citenamefont {Ba{\~n}uls}}]{pancotti2020quantum}%
  \BibitemOpen
  \bibfield  {author} {\bibinfo {author} {\bibfnamefont {N.}~\bibnamefont
  {Pancotti}}, \bibinfo {author} {\bibfnamefont {G.}~\bibnamefont {Giudice}},
  \bibinfo {author} {\bibfnamefont {J.~I.}\ \bibnamefont {Cirac}}, \bibinfo
  {author} {\bibfnamefont {J.~P.}\ \bibnamefont {Garrahan}},\ and\ \bibinfo
  {author} {\bibfnamefont {M.~C.}\ \bibnamefont {Ba{\~n}uls}},\ }\bibfield
  {title} {\bibinfo {title} {Quantum east model: Localization, nonthermal
  eigenstates, and slow dynamics},\ }\href@noop {} {\bibfield  {journal}
  {\bibinfo  {journal} {Physical Review X}\ }\textbf {\bibinfo {volume} {10}},\
  \bibinfo {pages} {021051} (\bibinfo {year} {2020})}\BibitemShut {NoStop}%
\bibitem [{\citenamefont {van Horssen}\ \emph {et~al.}(2015)\citenamefont {van
  Horssen}, \citenamefont {Levi},\ and\ \citenamefont
  {Garrahan}}]{van2015dynamics}%
  \BibitemOpen
  \bibfield  {author} {\bibinfo {author} {\bibfnamefont {M.}~\bibnamefont {van
  Horssen}}, \bibinfo {author} {\bibfnamefont {E.}~\bibnamefont {Levi}},\ and\
  \bibinfo {author} {\bibfnamefont {J.~P.}\ \bibnamefont {Garrahan}},\
  }\bibfield  {title} {\bibinfo {title} {Dynamics of many-body localization in
  a translation-invariant quantum glass model},\ }\href@noop {} {\bibfield
  {journal} {\bibinfo  {journal} {Physical Review B}\ }\textbf {\bibinfo
  {volume} {92}},\ \bibinfo {pages} {100305} (\bibinfo {year}
  {2015})}\BibitemShut {NoStop}%
\bibitem [{\citenamefont {Brighi}\ \emph {et~al.}(2023)\citenamefont {Brighi},
  \citenamefont {Ljubotina},\ and\ \citenamefont {Serbyn}}]{brighi2023hilbert}%
  \BibitemOpen
  \bibfield  {author} {\bibinfo {author} {\bibfnamefont {P.}~\bibnamefont
  {Brighi}}, \bibinfo {author} {\bibfnamefont {M.}~\bibnamefont {Ljubotina}},\
  and\ \bibinfo {author} {\bibfnamefont {M.}~\bibnamefont {Serbyn}},\
  }\bibfield  {title} {\bibinfo {title} {Hilbert space fragmentation and slow
  dynamics in particle-conserving quantum east models},\ }\href@noop {}
  {\bibfield  {journal} {\bibinfo  {journal} {SciPost Physics}\ }\textbf
  {\bibinfo {volume} {15}},\ \bibinfo {pages} {093} (\bibinfo {year}
  {2023})}\BibitemShut {NoStop}%
\bibitem [{\citenamefont {Lan}\ \emph {et~al.}(2018)\citenamefont {Lan},
  \citenamefont {van Horssen}, \citenamefont {Powell},\ and\ \citenamefont
  {Garrahan}}]{lan2018quantum}%
  \BibitemOpen
  \bibfield  {author} {\bibinfo {author} {\bibfnamefont {Z.}~\bibnamefont
  {Lan}}, \bibinfo {author} {\bibfnamefont {M.}~\bibnamefont {van Horssen}},
  \bibinfo {author} {\bibfnamefont {S.}~\bibnamefont {Powell}},\ and\ \bibinfo
  {author} {\bibfnamefont {J.~P.}\ \bibnamefont {Garrahan}},\ }\bibfield
  {title} {\bibinfo {title} {Quantum slow relaxation and metastability due to
  dynamical constraints},\ }\href@noop {} {\bibfield  {journal} {\bibinfo
  {journal} {Physical review letters}\ }\textbf {\bibinfo {volume} {121}},\
  \bibinfo {pages} {040603} (\bibinfo {year} {2018})}\BibitemShut {NoStop}%
\bibitem [{\citenamefont {Garrahan}(2018)}]{garrahan2018aspects}%
  \BibitemOpen
  \bibfield  {author} {\bibinfo {author} {\bibfnamefont {J.~P.}\ \bibnamefont
  {Garrahan}},\ }\bibfield  {title} {\bibinfo {title} {Aspects of
  non-equilibrium in classical and quantum systems: Slow relaxation and
  glasses, dynamical large deviations, quantum non-ergodicity, and open quantum
  dynamics},\ }\href@noop {} {\bibfield  {journal} {\bibinfo  {journal}
  {Physica A: Statistical Mechanics and its Applications}\ }\textbf {\bibinfo
  {volume} {504}},\ \bibinfo {pages} {130} (\bibinfo {year}
  {2018})}\BibitemShut {NoStop}%
\bibitem [{\citenamefont {De~Tomasi}\ \emph {et~al.}(2019)\citenamefont
  {De~Tomasi}, \citenamefont {Hetterich}, \citenamefont {Sala},\ and\
  \citenamefont {Pollmann}}]{detomasi2019dynamics}%
  \BibitemOpen
  \bibfield  {author} {\bibinfo {author} {\bibfnamefont {G.}~\bibnamefont
  {De~Tomasi}}, \bibinfo {author} {\bibfnamefont {D.}~\bibnamefont
  {Hetterich}}, \bibinfo {author} {\bibfnamefont {P.}~\bibnamefont {Sala}},\
  and\ \bibinfo {author} {\bibfnamefont {F.}~\bibnamefont {Pollmann}},\
  }\bibfield  {title} {\bibinfo {title} {Dynamics of strongly interacting
  systems: From fock-space fragmentation to many-body localization},\ }\href
  {https://doi.org/10.1103/PhysRevB.100.214313} {\bibfield  {journal} {\bibinfo
   {journal} {Phys. Rev. B}\ }\textbf {\bibinfo {volume} {100}},\ \bibinfo
  {pages} {214313} (\bibinfo {year} {2019})}\BibitemShut {NoStop}%
\bibitem [{\citenamefont {Moudgalya}\ \emph
  {et~al.}(2022{\natexlab{b}})\citenamefont {Moudgalya}, \citenamefont {Prem},
  \citenamefont {Nandkishore}, \citenamefont {Regnault},\ and\ \citenamefont
  {Bernevig}}]{moudgalya2022thermalization}%
  \BibitemOpen
  \bibfield  {author} {\bibinfo {author} {\bibfnamefont {S.}~\bibnamefont
  {Moudgalya}}, \bibinfo {author} {\bibfnamefont {A.}~\bibnamefont {Prem}},
  \bibinfo {author} {\bibfnamefont {R.}~\bibnamefont {Nandkishore}}, \bibinfo
  {author} {\bibfnamefont {N.}~\bibnamefont {Regnault}},\ and\ \bibinfo
  {author} {\bibfnamefont {B.~A.}\ \bibnamefont {Bernevig}},\ }\bibfield
  {title} {\bibinfo {title} {Thermalization and its absence within krylov
  subspaces of a constrained hamiltonian},\ }in\ \href@noop {} {\emph {\bibinfo
  {booktitle} {Memorial Volume for Shoucheng Zhang}}}\ (\bibinfo  {publisher}
  {World Scientific},\ \bibinfo {year} {2022})\ pp.\ \bibinfo {pages}
  {147--209}\BibitemShut {NoStop}%
\bibitem [{\citenamefont {Kohlert}\ \emph {et~al.}(2023)\citenamefont
  {Kohlert}, \citenamefont {Scherg}, \citenamefont {Sala}, \citenamefont
  {Pollmann}, \citenamefont {Hebbe~Madhusudhana}, \citenamefont {Bloch},\ and\
  \citenamefont {Aidelsburger}}]{kohlert2023exploring}%
  \BibitemOpen
  \bibfield  {author} {\bibinfo {author} {\bibfnamefont {T.}~\bibnamefont
  {Kohlert}}, \bibinfo {author} {\bibfnamefont {S.}~\bibnamefont {Scherg}},
  \bibinfo {author} {\bibfnamefont {P.}~\bibnamefont {Sala}}, \bibinfo {author}
  {\bibfnamefont {F.}~\bibnamefont {Pollmann}}, \bibinfo {author}
  {\bibfnamefont {B.}~\bibnamefont {Hebbe~Madhusudhana}}, \bibinfo {author}
  {\bibfnamefont {I.}~\bibnamefont {Bloch}},\ and\ \bibinfo {author}
  {\bibfnamefont {M.}~\bibnamefont {Aidelsburger}},\ }\bibfield  {title}
  {\bibinfo {title} {Exploring the regime of fragmentation in strongly tilted
  fermi-hubbard chains},\ }\href
  {https://doi.org/10.1103/PhysRevLett.130.010201} {\bibfield  {journal}
  {\bibinfo  {journal} {Phys. Rev. Lett.}\ }\textbf {\bibinfo {volume} {130}},\
  \bibinfo {pages} {010201} (\bibinfo {year} {2023})}\BibitemShut {NoStop}%
\bibitem [{\citenamefont {Kinoshita}\ \emph {et~al.}(2006)\citenamefont
  {Kinoshita}, \citenamefont {Wenger},\ and\ \citenamefont
  {Weiss}}]{kinoshita2006quantum}%
  \BibitemOpen
  \bibfield  {author} {\bibinfo {author} {\bibfnamefont {T.}~\bibnamefont
  {Kinoshita}}, \bibinfo {author} {\bibfnamefont {T.}~\bibnamefont {Wenger}},\
  and\ \bibinfo {author} {\bibfnamefont {D.~S.}\ \bibnamefont {Weiss}},\
  }\bibfield  {title} {\bibinfo {title} {A quantum newton's cradle},\
  }\href@noop {} {\bibfield  {journal} {\bibinfo  {journal} {Nature}\ }\textbf
  {\bibinfo {volume} {440}},\ \bibinfo {pages} {900} (\bibinfo {year}
  {2006})}\BibitemShut {NoStop}%
\bibitem [{\citenamefont {Scheie}\ \emph {et~al.}(2021)\citenamefont {Scheie},
  \citenamefont {Sherman}, \citenamefont {Dupont}, \citenamefont {Nagler},
  \citenamefont {Stone}, \citenamefont {Granroth}, \citenamefont {Moore},\ and\
  \citenamefont {Tennant}}]{scheie2021detection}%
  \BibitemOpen
  \bibfield  {author} {\bibinfo {author} {\bibfnamefont {A.}~\bibnamefont
  {Scheie}}, \bibinfo {author} {\bibfnamefont {N.}~\bibnamefont {Sherman}},
  \bibinfo {author} {\bibfnamefont {M.}~\bibnamefont {Dupont}}, \bibinfo
  {author} {\bibfnamefont {S.}~\bibnamefont {Nagler}}, \bibinfo {author}
  {\bibfnamefont {M.}~\bibnamefont {Stone}}, \bibinfo {author} {\bibfnamefont
  {G.}~\bibnamefont {Granroth}}, \bibinfo {author} {\bibfnamefont
  {J.}~\bibnamefont {Moore}},\ and\ \bibinfo {author} {\bibfnamefont
  {D.}~\bibnamefont {Tennant}},\ }\bibfield  {title} {\bibinfo {title}
  {Detection of kardar--parisi--zhang hydrodynamics in a quantum heisenberg
  spin-1/2 chain},\ }\href@noop {} {\bibfield  {journal} {\bibinfo  {journal}
  {Nature Physics}\ }\textbf {\bibinfo {volume} {17}},\ \bibinfo {pages} {726}
  (\bibinfo {year} {2021})}\BibitemShut {NoStop}%
\bibitem [{\citenamefont {Malvania}\ \emph {et~al.}(2021)\citenamefont
  {Malvania}, \citenamefont {Zhang}, \citenamefont {Le}, \citenamefont
  {Dubail}, \citenamefont {Rigol},\ and\ \citenamefont
  {Weiss}}]{malvania2021generalized}%
  \BibitemOpen
  \bibfield  {author} {\bibinfo {author} {\bibfnamefont {N.}~\bibnamefont
  {Malvania}}, \bibinfo {author} {\bibfnamefont {Y.}~\bibnamefont {Zhang}},
  \bibinfo {author} {\bibfnamefont {Y.}~\bibnamefont {Le}}, \bibinfo {author}
  {\bibfnamefont {J.}~\bibnamefont {Dubail}}, \bibinfo {author} {\bibfnamefont
  {M.}~\bibnamefont {Rigol}},\ and\ \bibinfo {author} {\bibfnamefont {D.~S.}\
  \bibnamefont {Weiss}},\ }\bibfield  {title} {\bibinfo {title} {Generalized
  hydrodynamics in strongly interacting 1d bose gases},\ }\href@noop {}
  {\bibfield  {journal} {\bibinfo  {journal} {Science}\ }\textbf {\bibinfo
  {volume} {373}},\ \bibinfo {pages} {1129} (\bibinfo {year}
  {2021})}\BibitemShut {NoStop}%
\bibitem [{\citenamefont {Wei}\ \emph {et~al.}(2022)\citenamefont {Wei},
  \citenamefont {Rubio-Abadal}, \citenamefont {Ye}, \citenamefont {Machado},
  \citenamefont {Kemp}, \citenamefont {Srakaew}, \citenamefont {Hollerith},
  \citenamefont {Rui}, \citenamefont {Gopalakrishnan}, \citenamefont {Yao}
  \emph {et~al.}}]{wei2022quantum}%
  \BibitemOpen
  \bibfield  {author} {\bibinfo {author} {\bibfnamefont {D.}~\bibnamefont
  {Wei}}, \bibinfo {author} {\bibfnamefont {A.}~\bibnamefont {Rubio-Abadal}},
  \bibinfo {author} {\bibfnamefont {B.}~\bibnamefont {Ye}}, \bibinfo {author}
  {\bibfnamefont {F.}~\bibnamefont {Machado}}, \bibinfo {author} {\bibfnamefont
  {J.}~\bibnamefont {Kemp}}, \bibinfo {author} {\bibfnamefont {K.}~\bibnamefont
  {Srakaew}}, \bibinfo {author} {\bibfnamefont {S.}~\bibnamefont {Hollerith}},
  \bibinfo {author} {\bibfnamefont {J.}~\bibnamefont {Rui}}, \bibinfo {author}
  {\bibfnamefont {S.}~\bibnamefont {Gopalakrishnan}}, \bibinfo {author}
  {\bibfnamefont {N.~Y.}\ \bibnamefont {Yao}}, \emph {et~al.},\ }\bibfield
  {title} {\bibinfo {title} {Quantum gas microscopy of kardar-parisi-zhang
  superdiffusion},\ }\href@noop {} {\bibfield  {journal} {\bibinfo  {journal}
  {Science}\ }\textbf {\bibinfo {volume} {376}},\ \bibinfo {pages} {716}
  (\bibinfo {year} {2022})}\BibitemShut {NoStop}%
\bibitem [{\citenamefont {Scherg}\ \emph {et~al.}(2021)\citenamefont {Scherg},
  \citenamefont {Kohlert}, \citenamefont {Sala}, \citenamefont {Pollmann},
  \citenamefont {Hebbe~Madhusudhana}, \citenamefont {Bloch},\ and\
  \citenamefont {Aidelsburger}}]{scherg2021observing}%
  \BibitemOpen
  \bibfield  {author} {\bibinfo {author} {\bibfnamefont {S.}~\bibnamefont
  {Scherg}}, \bibinfo {author} {\bibfnamefont {T.}~\bibnamefont {Kohlert}},
  \bibinfo {author} {\bibfnamefont {P.}~\bibnamefont {Sala}}, \bibinfo {author}
  {\bibfnamefont {F.}~\bibnamefont {Pollmann}}, \bibinfo {author}
  {\bibfnamefont {B.}~\bibnamefont {Hebbe~Madhusudhana}}, \bibinfo {author}
  {\bibfnamefont {I.}~\bibnamefont {Bloch}},\ and\ \bibinfo {author}
  {\bibfnamefont {M.}~\bibnamefont {Aidelsburger}},\ }\bibfield  {title}
  {\bibinfo {title} {Observing non-ergodicity due to kinetic constraints in
  tilted fermi-hubbard chains},\ }\href@noop {} {\bibfield  {journal} {\bibinfo
   {journal} {Nature Communications}\ }\textbf {\bibinfo {volume} {12}},\
  \bibinfo {pages} {4490} (\bibinfo {year} {2021})}\BibitemShut {NoStop}%
\bibitem [{\citenamefont {Hastings}\ and\ \citenamefont
  {Freedman}(2013)}]{hastings}%
  \BibitemOpen
  \bibfield  {author} {\bibinfo {author} {\bibfnamefont {M.~B.}\ \bibnamefont
  {Hastings}}\ and\ \bibinfo {author} {\bibfnamefont {M.~H.}\ \bibnamefont
  {Freedman}},\ }\bibfield  {title} {\bibinfo {title} {Obstructions to
  classically simulating the quantum adiabatic algorithm},\ }\href@noop {}
  {\bibfield  {journal} {\bibinfo  {journal} {Quantum Info. Comput.}\ }\textbf
  {\bibinfo {volume} {13}},\ \bibinfo {pages} {1038–1076} (\bibinfo {year}
  {2013})}\BibitemShut {NoStop}%
\bibitem [{\citenamefont {Aldous}\ and\ \citenamefont
  {Diaconis}(1986)}]{aldous1986shuffling}%
  \BibitemOpen
  \bibfield  {author} {\bibinfo {author} {\bibfnamefont {D.}~\bibnamefont
  {Aldous}}\ and\ \bibinfo {author} {\bibfnamefont {P.}~\bibnamefont
  {Diaconis}},\ }\bibfield  {title} {\bibinfo {title} {Shuffling cards and
  stopping times},\ }\href@noop {} {\bibfield  {journal} {\bibinfo  {journal}
  {The American Mathematical Monthly}\ }\textbf {\bibinfo {volume} {93}},\
  \bibinfo {pages} {333} (\bibinfo {year} {1986})}\BibitemShut {NoStop}%
\bibitem [{\citenamefont {Diaconis}(1996)}]{diaconis1996cutoff}%
  \BibitemOpen
  \bibfield  {author} {\bibinfo {author} {\bibfnamefont {P.}~\bibnamefont
  {Diaconis}},\ }\bibfield  {title} {\bibinfo {title} {The cutoff phenomenon in
  finite markov chains.},\ }\href@noop {} {\bibfield  {journal} {\bibinfo
  {journal} {Proceedings of the National Academy of Sciences}\ }\textbf
  {\bibinfo {volume} {93}},\ \bibinfo {pages} {1659} (\bibinfo {year}
  {1996})}\BibitemShut {NoStop}%
\bibitem [{\citenamefont {Po}\ \emph {et~al.}(2018)\citenamefont {Po},
  \citenamefont {Watanabe},\ and\ \citenamefont {Vishwanath}}]{po2018fragile}%
  \BibitemOpen
  \bibfield  {author} {\bibinfo {author} {\bibfnamefont {H.~C.}\ \bibnamefont
  {Po}}, \bibinfo {author} {\bibfnamefont {H.}~\bibnamefont {Watanabe}},\ and\
  \bibinfo {author} {\bibfnamefont {A.}~\bibnamefont {Vishwanath}},\ }\bibfield
   {title} {\bibinfo {title} {Fragile topology and wannier obstructions},\
  }\href@noop {} {\bibfield  {journal} {\bibinfo  {journal} {Physical review
  letters}\ }\textbf {\bibinfo {volume} {121}},\ \bibinfo {pages} {126402}
  (\bibinfo {year} {2018})}\BibitemShut {NoStop}%
\bibitem [{\citenamefont {Clay}\ and\ \citenamefont
  {Margalit}(2017)}]{clay2017office}%
  \BibitemOpen
  \bibfield  {author} {\bibinfo {author} {\bibfnamefont {M.}~\bibnamefont
  {Clay}}\ and\ \bibinfo {author} {\bibfnamefont {D.}~\bibnamefont
  {Margalit}},\ }\href@noop {} {\emph {\bibinfo {title} {Office hours with a
  geometric group theorist}}}\ (\bibinfo  {publisher} {Princeton University
  Press},\ \bibinfo {year} {2017})\BibitemShut {NoStop}%
\bibitem [{\citenamefont {Gersten}\ \emph {et~al.}(2003)\citenamefont
  {Gersten}, \citenamefont {Holt},\ and\ \citenamefont
  {Riley}}]{gersten2003isoperimetric}%
  \BibitemOpen
  \bibfield  {author} {\bibinfo {author} {\bibfnamefont {S.~M.}\ \bibnamefont
  {Gersten}}, \bibinfo {author} {\bibfnamefont {D.~F.}\ \bibnamefont {Holt}},\
  and\ \bibinfo {author} {\bibfnamefont {T.~R.}\ \bibnamefont {Riley}},\
  }\bibfield  {title} {\bibinfo {title} {Isoperimetric inequalities for
  nilpotent groups},\ }\href@noop {} {\bibfield  {journal} {\bibinfo  {journal}
  {Geometric \& Functional Analysis GAFA}\ }\textbf {\bibinfo {volume} {13}},\
  \bibinfo {pages} {795} (\bibinfo {year} {2003})}\BibitemShut {NoStop}%
\bibitem [{\citenamefont {Brady}\ and\ \citenamefont
  {Bridson}(2000)}]{brady2000there}%
  \BibitemOpen
  \bibfield  {author} {\bibinfo {author} {\bibfnamefont {N.}~\bibnamefont
  {Brady}}\ and\ \bibinfo {author} {\bibfnamefont {M.~R.}\ \bibnamefont
  {Bridson}},\ }\bibfield  {title} {\bibinfo {title} {There is only one gap in
  the isoperimetric spectrum},\ }\href@noop {} {\bibfield  {journal} {\bibinfo
  {journal} {Geometric and Functional Analysis}\ }\textbf {\bibinfo {volume}
  {10}},\ \bibinfo {pages} {1053} (\bibinfo {year} {2000})}\BibitemShut
  {NoStop}%
\bibitem [{\citenamefont {Sapir}\ \emph {et~al.}(2002)\citenamefont {Sapir},
  \citenamefont {Birget},\ and\ \citenamefont {Rips}}]{sapir2002isoperimetric}%
  \BibitemOpen
  \bibfield  {author} {\bibinfo {author} {\bibfnamefont {M.~V.}\ \bibnamefont
  {Sapir}}, \bibinfo {author} {\bibfnamefont {J.-C.}\ \bibnamefont {Birget}},\
  and\ \bibinfo {author} {\bibfnamefont {E.}~\bibnamefont {Rips}},\ }\bibfield
  {title} {\bibinfo {title} {Isoperimetric and isodiametric functions of
  groups},\ }\href@noop {} {\bibfield  {journal} {\bibinfo  {journal} {Annals
  of Mathematics}\ ,\ \bibinfo {pages} {345}} (\bibinfo {year}
  {2002})}\BibitemShut {NoStop}%
\bibitem [{\citenamefont {Young}(2008)}]{young2008averaged}%
  \BibitemOpen
  \bibfield  {author} {\bibinfo {author} {\bibfnamefont {R.}~\bibnamefont
  {Young}},\ }\bibfield  {title} {\bibinfo {title} {Averaged dehn functions for
  nilpotent groups},\ }\href@noop {} {\bibfield  {journal} {\bibinfo  {journal}
  {Topology}\ }\textbf {\bibinfo {volume} {47}},\ \bibinfo {pages} {351}
  (\bibinfo {year} {2008})}\BibitemShut {NoStop}%
\bibitem [{\citenamefont {Moudgalya}\ and\ \citenamefont
  {Motrunich}(2023)}]{MOUDGALYA2023169384}%
  \BibitemOpen
  \bibfield  {author} {\bibinfo {author} {\bibfnamefont {S.}~\bibnamefont
  {Moudgalya}}\ and\ \bibinfo {author} {\bibfnamefont {O.~I.}\ \bibnamefont
  {Motrunich}},\ }\bibfield  {title} {\bibinfo {title} {From symmetries to
  commutant algebras in standard hamiltonians},\ }\href
  {https://doi.org/https://doi.org/10.1016/j.aop.2023.169384} {\bibfield
  {journal} {\bibinfo  {journal} {Annals of Physics}\ }\textbf {\bibinfo
  {volume} {455}},\ \bibinfo {pages} {169384} (\bibinfo {year}
  {2023})}\BibitemShut {NoStop}%
\bibitem [{\citenamefont {Caha}\ and\ \citenamefont
  {Nagaj}(2018)}]{caha2018pair}%
  \BibitemOpen
  \bibfield  {author} {\bibinfo {author} {\bibfnamefont {L.}~\bibnamefont
  {Caha}}\ and\ \bibinfo {author} {\bibfnamefont {D.}~\bibnamefont {Nagaj}},\
  }\bibfield  {title} {\bibinfo {title} {The pair-flip model: a very entangled
  translationally invariant spin chain},\ }\href@noop {} {\bibfield  {journal}
  {\bibinfo  {journal} {arXiv preprint arXiv:1805.07168}\ } (\bibinfo {year}
  {2018})}\BibitemShut {NoStop}%
\bibitem [{\citenamefont {Morningstar}\ \emph {et~al.}(2020)\citenamefont
  {Morningstar}, \citenamefont {Khemani},\ and\ \citenamefont
  {Huse}}]{morningstar2020kinetically}%
  \BibitemOpen
  \bibfield  {author} {\bibinfo {author} {\bibfnamefont {A.}~\bibnamefont
  {Morningstar}}, \bibinfo {author} {\bibfnamefont {V.}~\bibnamefont
  {Khemani}},\ and\ \bibinfo {author} {\bibfnamefont {D.~A.}\ \bibnamefont
  {Huse}},\ }\bibfield  {title} {\bibinfo {title} {Kinetically constrained
  freezing transition in a dipole-conserving system},\ }\href@noop {}
  {\bibfield  {journal} {\bibinfo  {journal} {Physical Review B}\ }\textbf
  {\bibinfo {volume} {101}},\ \bibinfo {pages} {214205} (\bibinfo {year}
  {2020})}\BibitemShut {NoStop}%
\bibitem [{\citenamefont {Pozderac}\ \emph {et~al.}(2023)\citenamefont
  {Pozderac}, \citenamefont {Speck}, \citenamefont {Feng}, \citenamefont
  {Huse},\ and\ \citenamefont {Skinner}}]{pozderac2023exact}%
  \BibitemOpen
  \bibfield  {author} {\bibinfo {author} {\bibfnamefont {C.}~\bibnamefont
  {Pozderac}}, \bibinfo {author} {\bibfnamefont {S.}~\bibnamefont {Speck}},
  \bibinfo {author} {\bibfnamefont {X.}~\bibnamefont {Feng}}, \bibinfo {author}
  {\bibfnamefont {D.~A.}\ \bibnamefont {Huse}},\ and\ \bibinfo {author}
  {\bibfnamefont {B.}~\bibnamefont {Skinner}},\ }\bibfield  {title} {\bibinfo
  {title} {Exact solution for the filling-induced thermalization transition in
  a one-dimensional fracton system},\ }\href@noop {} {\bibfield  {journal}
  {\bibinfo  {journal} {Physical Review B}\ }\textbf {\bibinfo {volume}
  {107}},\ \bibinfo {pages} {045137} (\bibinfo {year} {2023})}\BibitemShut
  {NoStop}%
\bibitem [{\citenamefont {Wang}\ and\ \citenamefont
  {Yang}(2023)}]{wang2023freezing}%
  \BibitemOpen
  \bibfield  {author} {\bibinfo {author} {\bibfnamefont {C.}~\bibnamefont
  {Wang}}\ and\ \bibinfo {author} {\bibfnamefont {Z.-C.}\ \bibnamefont
  {Yang}},\ }\bibfield  {title} {\bibinfo {title} {Freezing transition in
  particle-conserving east model},\ }\href@noop {} {\bibfield  {journal}
  {\bibinfo  {journal} {arXiv preprint arXiv:2307.01993}\ } (\bibinfo {year}
  {2023})}\BibitemShut {NoStop}%
\bibitem [{\citenamefont {Bravyi}\ \emph {et~al.}(2012)\citenamefont {Bravyi},
  \citenamefont {Caha}, \citenamefont {Movassagh}, \citenamefont {Nagaj},\ and\
  \citenamefont {Shor}}]{Bravyi}%
  \BibitemOpen
  \bibfield  {author} {\bibinfo {author} {\bibfnamefont {S.}~\bibnamefont
  {Bravyi}}, \bibinfo {author} {\bibfnamefont {L.}~\bibnamefont {Caha}},
  \bibinfo {author} {\bibfnamefont {R.}~\bibnamefont {Movassagh}}, \bibinfo
  {author} {\bibfnamefont {D.}~\bibnamefont {Nagaj}},\ and\ \bibinfo {author}
  {\bibfnamefont {P.~W.}\ \bibnamefont {Shor}},\ }\bibfield  {title} {\bibinfo
  {title} {Criticality without frustration for quantum spin-1 chains},\ }\href
  {https://doi.org/10.1103/PhysRevLett.109.207202} {\bibfield  {journal}
  {\bibinfo  {journal} {Phys. Rev. Lett.}\ }\textbf {\bibinfo {volume} {109}},\
  \bibinfo {pages} {207202} (\bibinfo {year} {2012})}\BibitemShut {NoStop}%
\bibitem [{\citenamefont {Singh}\ \emph {et~al.}(2021)\citenamefont {Singh},
  \citenamefont {Ware}, \citenamefont {Vasseur},\ and\ \citenamefont
  {Friedman}}]{Singh}%
  \BibitemOpen
  \bibfield  {author} {\bibinfo {author} {\bibfnamefont {H.}~\bibnamefont
  {Singh}}, \bibinfo {author} {\bibfnamefont {B.~A.}\ \bibnamefont {Ware}},
  \bibinfo {author} {\bibfnamefont {R.}~\bibnamefont {Vasseur}},\ and\ \bibinfo
  {author} {\bibfnamefont {A.~J.}\ \bibnamefont {Friedman}},\ }\bibfield
  {title} {\bibinfo {title} {Subdiffusion and many-body quantum chaos with
  kinetic constraints},\ }\href
  {https://doi.org/10.1103/PhysRevLett.127.230602} {\bibfield  {journal}
  {\bibinfo  {journal} {Phys. Rev. Lett.}\ }\textbf {\bibinfo {volume} {127}},\
  \bibinfo {pages} {230602} (\bibinfo {year} {2021})}\BibitemShut {NoStop}%
\bibitem [{\citenamefont {Levin}\ and\ \citenamefont
  {Peres}(2017)}]{levin2017markov}%
  \BibitemOpen
  \bibfield  {author} {\bibinfo {author} {\bibfnamefont {D.~A.}\ \bibnamefont
  {Levin}}\ and\ \bibinfo {author} {\bibfnamefont {Y.}~\bibnamefont {Peres}},\
  }\href@noop {} {\emph {\bibinfo {title} {Markov chains and mixing times}}},\
  Vol.\ \bibinfo {volume} {107}\ (\bibinfo  {publisher} {American Mathematical
  Soc.},\ \bibinfo {year} {2017})\BibitemShut {NoStop}%
\bibitem [{\citenamefont {Lyons}\ and\ \citenamefont
  {Peres}(2017)}]{lyons2017probability}%
  \BibitemOpen
  \bibfield  {author} {\bibinfo {author} {\bibfnamefont {R.}~\bibnamefont
  {Lyons}}\ and\ \bibinfo {author} {\bibfnamefont {Y.}~\bibnamefont {Peres}},\
  }\href@noop {} {\emph {\bibinfo {title} {Probability on trees and
  networks}}},\ Vol.~\bibinfo {volume} {42}\ (\bibinfo  {publisher} {Cambridge
  University Press},\ \bibinfo {year} {2017})\BibitemShut {NoStop}%
\bibitem [{\citenamefont {Baumslag}(1969)}]{baumslag1969non}%
  \BibitemOpen
  \bibfield  {author} {\bibinfo {author} {\bibfnamefont {G.}~\bibnamefont
  {Baumslag}},\ }\bibfield  {title} {\bibinfo {title} {A non-cyclic one-relator
  group all of whose finite quotients are cyclic: To bernhard hermann neumann
  on his 60th birthday},\ }\href@noop {} {\bibfield  {journal} {\bibinfo
  {journal} {Journal of the Australian Mathematical Society}\ }\textbf
  {\bibinfo {volume} {10}},\ \bibinfo {pages} {497} (\bibinfo {year}
  {1969})}\BibitemShut {NoStop}%
\bibitem [{\citenamefont {Gersten}(1992)}]{gersten1992dehn}%
  \BibitemOpen
  \bibfield  {author} {\bibinfo {author} {\bibfnamefont {S.~M.}\ \bibnamefont
  {Gersten}},\ }\bibfield  {title} {\bibinfo {title} {Dehn functions and l
  1-norms of finite presentations},\ }in\ \href@noop {} {\emph {\bibinfo
  {booktitle} {Algorithms and classification in combinatorial group theory}}}\
  (\bibinfo  {publisher} {Springer},\ \bibinfo {year} {1992})\ pp.\ \bibinfo
  {pages} {195--224}\BibitemShut {NoStop}%
\bibitem [{\citenamefont {Gersten}\ and\ \citenamefont
  {Riley}(2002)}]{gersten2002filling}%
  \BibitemOpen
  \bibfield  {author} {\bibinfo {author} {\bibfnamefont {S.~M.}\ \bibnamefont
  {Gersten}}\ and\ \bibinfo {author} {\bibfnamefont {T.~R.}\ \bibnamefont
  {Riley}},\ }\bibfield  {title} {\bibinfo {title} {Filling length in finitely
  presentable groups},\ }\href@noop {} {\bibfield  {journal} {\bibinfo
  {journal} {Geometriae Dedicata}\ }\textbf {\bibinfo {volume} {92}},\ \bibinfo
  {pages} {41} (\bibinfo {year} {2002})}\BibitemShut {NoStop}%
\bibitem [{\citenamefont {Movassagh}\ and\ \citenamefont
  {Shor}(2016)}]{movassagh}%
  \BibitemOpen
  \bibfield  {author} {\bibinfo {author} {\bibfnamefont {R.}~\bibnamefont
  {Movassagh}}\ and\ \bibinfo {author} {\bibfnamefont {P.~W.}\ \bibnamefont
  {Shor}},\ }\bibfield  {title} {\bibinfo {title} {Supercritical entanglement
  in local systems: Counterexample to the area law for quantum matter},\
  }\href@noop {} {\bibfield  {journal} {\bibinfo  {journal} {Proceedings of the
  National Academy of Sciences}\ }\textbf {\bibinfo {volume} {113}},\ \bibinfo
  {pages} {13278} (\bibinfo {year} {2016})}\BibitemShut {NoStop}%
\bibitem [{\citenamefont {Stephen}\ \emph {et~al.}(2022)\citenamefont
  {Stephen}, \citenamefont {Hart},\ and\ \citenamefont
  {Nandkishore}}]{stephen2022ergodicity}%
  \BibitemOpen
  \bibfield  {author} {\bibinfo {author} {\bibfnamefont {D.~T.}\ \bibnamefont
  {Stephen}}, \bibinfo {author} {\bibfnamefont {O.}~\bibnamefont {Hart}},\ and\
  \bibinfo {author} {\bibfnamefont {R.~M.}\ \bibnamefont {Nandkishore}},\
  }\bibfield  {title} {\bibinfo {title} {Ergodicity breaking provably robust to
  arbitrary perturbations},\ }\href@noop {} {\bibfield  {journal} {\bibinfo
  {journal} {arXiv preprint arXiv:2209.03966}\ } (\bibinfo {year}
  {2022})}\BibitemShut {NoStop}%
\bibitem [{\citenamefont {Stahl}\ \emph {et~al.}(2023)\citenamefont {Stahl},
  \citenamefont {Nandkishore},\ and\ \citenamefont
  {Hart}}]{stahl2023topologically}%
  \BibitemOpen
  \bibfield  {author} {\bibinfo {author} {\bibfnamefont {C.}~\bibnamefont
  {Stahl}}, \bibinfo {author} {\bibfnamefont {R.}~\bibnamefont {Nandkishore}},\
  and\ \bibinfo {author} {\bibfnamefont {O.}~\bibnamefont {Hart}},\ }\bibfield
  {title} {\bibinfo {title} {Topologically stable ergodicity breaking from
  emergent higher-form symmetries in generalized quantum loop models},\
  }\href@noop {} {\bibfield  {journal} {\bibinfo  {journal} {arXiv preprint
  arXiv:2304.04792}\ } (\bibinfo {year} {2023})}\BibitemShut {NoStop}%
\bibitem [{\citenamefont {Balasubramanian}\ \emph {et~al.}(2023)\citenamefont
  {Balasubramanian}, \citenamefont {Lake},\ and\ \citenamefont
  {Choi}}]{Balasubramanian2023}%
  \BibitemOpen
  \bibfield  {author} {\bibinfo {author} {\bibfnamefont {S.}~\bibnamefont
  {Balasubramanian}}, \bibinfo {author} {\bibfnamefont {E.}~\bibnamefont
  {Lake}},\ and\ \bibinfo {author} {\bibfnamefont {S.}~\bibnamefont {Choi}},\
  }\bibfield  {title} {\bibinfo {title} {2d hamiltonians with exotic bipartite
  and topological entanglement},\ }\href@noop {} {\bibfield  {journal}
  {\bibinfo  {journal} {arXiv preprint arXiv:2305.07028}\ } (\bibinfo {year}
  {2023})}\BibitemShut {NoStop}%
\bibitem [{\citenamefont {Balasubramanian}\ \emph {et~al.}(2024)\citenamefont
  {Balasubramanian}, \citenamefont {Lake},\ and\ \citenamefont
  {Mahadevan}}]{BLM2024}%
  \BibitemOpen
  \bibfield  {author} {\bibinfo {author} {\bibfnamefont {S.}~\bibnamefont
  {Balasubramanian}}, \bibinfo {author} {\bibfnamefont {E.}~\bibnamefont
  {Lake}},\ and\ \bibinfo {author} {\bibfnamefont {A.}~\bibnamefont
  {Mahadevan}},\ }\href@noop {} {\bibfield  {journal} {\bibinfo  {journal} {to
  appear}\ } (\bibinfo {year} {2024})}\BibitemShut {NoStop}%
\bibitem [{\citenamefont {Rokhsar}\ and\ \citenamefont {Kivelson}(1988)}]{RK}%
  \BibitemOpen
  \bibfield  {author} {\bibinfo {author} {\bibfnamefont {D.~S.}\ \bibnamefont
  {Rokhsar}}\ and\ \bibinfo {author} {\bibfnamefont {S.~A.}\ \bibnamefont
  {Kivelson}},\ }\bibfield  {title} {\bibinfo {title} {Superconductivity and
  the quantum hard-core dimer gas},\ }\href
  {https://doi.org/10.1103/PhysRevLett.61.2376} {\bibfield  {journal} {\bibinfo
   {journal} {Phys. Rev. Lett.}\ }\textbf {\bibinfo {volume} {61}},\ \bibinfo
  {pages} {2376} (\bibinfo {year} {1988})}\BibitemShut {NoStop}%
\bibitem [{\citenamefont {Batchelor}\ and\ \citenamefont
  {Kuniba}(1991)}]{batchelor1991temperley}%
  \BibitemOpen
  \bibfield  {author} {\bibinfo {author} {\bibfnamefont {M.~T.}\ \bibnamefont
  {Batchelor}}\ and\ \bibinfo {author} {\bibfnamefont {A.}~\bibnamefont
  {Kuniba}},\ }\bibfield  {title} {\bibinfo {title} {Temperley-lieb lattice
  models arising from quantum groups},\ }\href@noop {} {\bibfield  {journal}
  {\bibinfo  {journal} {Journal of Physics A: Mathematical and General}\
  }\textbf {\bibinfo {volume} {24}},\ \bibinfo {pages} {2599} (\bibinfo {year}
  {1991})}\BibitemShut {NoStop}%
\bibitem [{\citenamefont {Aufgebauer}\ and\ \citenamefont
  {Kl{\"u}mper}(2010)}]{aufgebauer2010quantum}%
  \BibitemOpen
  \bibfield  {author} {\bibinfo {author} {\bibfnamefont {B.}~\bibnamefont
  {Aufgebauer}}\ and\ \bibinfo {author} {\bibfnamefont {A.}~\bibnamefont
  {Kl{\"u}mper}},\ }\bibfield  {title} {\bibinfo {title} {Quantum spin chains
  of temperley--lieb type: periodic boundary conditions, spectral
  multiplicities and finite temperature},\ }\href@noop {} {\bibfield  {journal}
  {\bibinfo  {journal} {Journal of Statistical Mechanics: Theory and
  Experiment}\ }\textbf {\bibinfo {volume} {2010}},\ \bibinfo {pages} {P05018}
  (\bibinfo {year} {2010})}\BibitemShut {NoStop}%
\bibitem [{\citenamefont {Han}\ \emph {et~al.}(2024)\citenamefont {Han},
  \citenamefont {Chen},\ and\ \citenamefont {Lake}}]{hsf_impurities}%
  \BibitemOpen
  \bibfield  {author} {\bibinfo {author} {\bibfnamefont {Y.}~\bibnamefont
  {Han}}, \bibinfo {author} {\bibfnamefont {X.}~\bibnamefont {Chen}},\ and\
  \bibinfo {author} {\bibfnamefont {E.}~\bibnamefont {Lake}},\ }\bibfield
  {title} {\bibinfo {title} {Exponentially slow thermalization and the
  robustness of hilbert space fragmentation},\ }\href@noop {} {\bibfield
  {journal} {\bibinfo  {journal} {arXiv preprint arXiv:2401.11294}\ } (\bibinfo
  {year} {2024})}\BibitemShut {NoStop}%
\bibitem [{\citenamefont {Sapir}(2011)}]{sapir2011asymptotic}%
  \BibitemOpen
  \bibfield  {author} {\bibinfo {author} {\bibfnamefont {M.}~\bibnamefont
  {Sapir}},\ }\bibfield  {title} {\bibinfo {title} {Asymptotic invariants,
  complexity of groups and related problems},\ }\href@noop {} {\bibfield
  {journal} {\bibinfo  {journal} {Bulletin of Mathematical Sciences}\ }\textbf
  {\bibinfo {volume} {1}},\ \bibinfo {pages} {277} (\bibinfo {year}
  {2011})}\BibitemShut {NoStop}%
\bibitem [{\citenamefont {Baumslag}\ and\ \citenamefont
  {Solitar}(1962)}]{baumslag1962some}%
  \BibitemOpen
  \bibfield  {author} {\bibinfo {author} {\bibfnamefont {G.}~\bibnamefont
  {Baumslag}}\ and\ \bibinfo {author} {\bibfnamefont {D.}~\bibnamefont
  {Solitar}},\ }\bibfield  {title} {\bibinfo {title} {Some two-generator
  one-relator non-hopfian groups},\ }\href@noop {} {\  (\bibinfo {year}
  {1962})}\BibitemShut {NoStop}%
\bibitem [{\citenamefont {Gromov}(1992)}]{gromov1992asymptotic}%
  \BibitemOpen
  \bibfield  {author} {\bibinfo {author} {\bibfnamefont {M.}~\bibnamefont
  {Gromov}},\ }\href@noop {} {\emph {\bibinfo {title} {Asymptotic invariants of
  infinite groups}}},\ \bibinfo {type} {Tech. Rep.}\ (\bibinfo  {institution}
  {P00001028},\ \bibinfo {year} {1992})\BibitemShut {NoStop}%
\bibitem [{\citenamefont {Riley}(2006)}]{riley2006filling}%
  \BibitemOpen
  \bibfield  {author} {\bibinfo {author} {\bibfnamefont {T.}~\bibnamefont
  {Riley}},\ }\bibfield  {title} {\bibinfo {title} {Filling functions},\
  }\href@noop {} {\bibfield  {journal} {\bibinfo  {journal} {arXiv preprint
  math/0603059}\ } (\bibinfo {year} {2006})}\BibitemShut {NoStop}%
\bibitem [{\citenamefont {Gersten}(1993)}]{gersten1993isoperimetric}%
  \BibitemOpen
  \bibfield  {author} {\bibinfo {author} {\bibfnamefont {S.~M.}\ \bibnamefont
  {Gersten}},\ }\bibfield  {title} {\bibinfo {title} {Isoperimetric and
  isodiametric functions of finite presentations},\ }\href@noop {} {\bibfield
  {journal} {\bibinfo  {journal} {Geometric group theory}\ }\textbf {\bibinfo
  {volume} {1}},\ \bibinfo {pages} {79} (\bibinfo {year} {1993})}\BibitemShut
  {NoStop}%
\bibitem [{\citenamefont {Grigorchuk}\ and\ \citenamefont
  {Pak}(2008)}]{grigorchuk2008groups}%
  \BibitemOpen
  \bibfield  {author} {\bibinfo {author} {\bibfnamefont {R.}~\bibnamefont
  {Grigorchuk}}\ and\ \bibinfo {author} {\bibfnamefont {I.}~\bibnamefont
  {Pak}},\ }\bibfield  {title} {\bibinfo {title} {Groups of intermediate
  growth: an introduction},\ }\href@noop {} {\bibfield  {journal} {\bibinfo
  {journal} {Enseign. Math.(2)}\ }\textbf {\bibinfo {volume} {54}},\ \bibinfo
  {pages} {251} (\bibinfo {year} {2008})}\BibitemShut {NoStop}%
\bibitem [{\citenamefont {Gromov}(1981)}]{gromov1981groups}%
  \BibitemOpen
  \bibfield  {author} {\bibinfo {author} {\bibfnamefont {M.}~\bibnamefont
  {Gromov}},\ }\bibfield  {title} {\bibinfo {title} {Groups of polynomial
  growth and expanding maps (with an appendix by jacques tits)},\ }\href@noop
  {} {\bibfield  {journal} {\bibinfo  {journal} {Publications Math{\'e}matiques
  de l'IH{\'E}S}\ }\textbf {\bibinfo {volume} {53}},\ \bibinfo {pages} {53}
  (\bibinfo {year} {1981})}\BibitemShut {NoStop}%
\bibitem [{\citenamefont {Davies}(1992)}]{davies1992heat}%
  \BibitemOpen
  \bibfield  {author} {\bibinfo {author} {\bibfnamefont {E.~B.}\ \bibnamefont
  {Davies}},\ }\bibfield  {title} {\bibinfo {title} {Heat kernel bounds,
  conservation of probability and the feller property},\ }\href@noop {}
  {\bibfield  {journal} {\bibinfo  {journal} {Journal d’Analyse
  Math{\'e}matique}\ }\textbf {\bibinfo {volume} {58}},\ \bibinfo {pages} {99}
  (\bibinfo {year} {1992})}\BibitemShut {NoStop}%
\bibitem [{\citenamefont {Woess}(2000)}]{woess2000random}%
  \BibitemOpen
  \bibfield  {author} {\bibinfo {author} {\bibfnamefont {W.}~\bibnamefont
  {Woess}},\ }\href@noop {} {\emph {\bibinfo {title} {Random walks on infinite
  graphs and groups}}},\ \bibinfo {number} {138}\ (\bibinfo  {publisher}
  {Cambridge university press},\ \bibinfo {year} {2000})\BibitemShut {NoStop}%
\bibitem [{\citenamefont {Markov}(1948)}]{markov1948impossibility}%
  \BibitemOpen
  \bibfield  {author} {\bibinfo {author} {\bibfnamefont {A.}~\bibnamefont
  {Markov}},\ }\bibfield  {title} {\bibinfo {title} {Impossibility of certain
  algorithms in the theory of associative systems},\ }\href@noop {} {\
  (\bibinfo {year} {1948})}\BibitemShut {NoStop}%
\bibitem [{\citenamefont {Tseitin}(1958)}]{tseitin1958associative}%
  \BibitemOpen
  \bibfield  {author} {\bibinfo {author} {\bibfnamefont {G.~S.}\ \bibnamefont
  {Tseitin}},\ }\bibfield  {title} {\bibinfo {title} {An associative calculus
  with an insoluble problem of equivalence},\ }\href@noop {} {\bibfield
  {journal} {\bibinfo  {journal} {Trudy Matematicheskogo Instituta imeni VA
  Steklova}\ }\textbf {\bibinfo {volume} {52}},\ \bibinfo {pages} {172}
  (\bibinfo {year} {1958})}\BibitemShut {NoStop}%
\bibitem [{\citenamefont {Novikov}(1955)}]{novikov1955algorithmic}%
  \BibitemOpen
  \bibfield  {author} {\bibinfo {author} {\bibfnamefont {P.~S.}\ \bibnamefont
  {Novikov}},\ }\bibfield  {title} {\bibinfo {title} {On the algorithmic
  unsolvability of the word problem in group theory},\ }\href@noop {}
  {\bibfield  {journal} {\bibinfo  {journal} {Trudy Matematicheskogo Instituta
  imeni VA Steklova}\ }\textbf {\bibinfo {volume} {44}},\ \bibinfo {pages} {3}
  (\bibinfo {year} {1955})}\BibitemShut {NoStop}%
\bibitem [{\citenamefont {Boone}(1959)}]{boone1959word}%
  \BibitemOpen
  \bibfield  {author} {\bibinfo {author} {\bibfnamefont {W.~W.}\ \bibnamefont
  {Boone}},\ }\bibfield  {title} {\bibinfo {title} {The word problem},\
  }\href@noop {} {\bibfield  {journal} {\bibinfo  {journal} {Annals of
  mathematics}\ }\textbf {\bibinfo {volume} {70}},\ \bibinfo {pages} {207}
  (\bibinfo {year} {1959})}\BibitemShut {NoStop}%
\bibitem [{\citenamefont {Rabin}(1958)}]{rabin1958recursive}%
  \BibitemOpen
  \bibfield  {author} {\bibinfo {author} {\bibfnamefont {M.~O.}\ \bibnamefont
  {Rabin}},\ }\bibfield  {title} {\bibinfo {title} {Recursive unsolvability of
  group theoretic problems},\ }\href@noop {} {\bibfield  {journal} {\bibinfo
  {journal} {Annals of Mathematics}\ }\textbf {\bibinfo {volume} {67}},\
  \bibinfo {pages} {172} (\bibinfo {year} {1958})}\BibitemShut {NoStop}%
\bibitem [{\citenamefont {Adyan}(1955)}]{adyan1955algorithmic}%
  \BibitemOpen
  \bibfield  {author} {\bibinfo {author} {\bibfnamefont {S.~I.}\ \bibnamefont
  {Adyan}},\ }\bibfield  {title} {\bibinfo {title} {Algorithmic unsolvability
  of problems of recognition of certain properties of groups},\ }in\ \href@noop
  {} {\emph {\bibinfo {booktitle} {Dokl. Akad. Nauk SSSR (NS)}}},\ Vol.\
  \bibinfo {volume} {103}\ (\bibinfo {year} {1955})\ p.~\bibinfo {pages}
  {48}\BibitemShut {NoStop}%
\bibitem [{\citenamefont {Cubitt}\ \emph {et~al.}(2015)\citenamefont {Cubitt},
  \citenamefont {Perez-Garcia},\ and\ \citenamefont
  {Wolf}}]{cubitt2015undecidability}%
  \BibitemOpen
  \bibfield  {author} {\bibinfo {author} {\bibfnamefont {T.~S.}\ \bibnamefont
  {Cubitt}}, \bibinfo {author} {\bibfnamefont {D.}~\bibnamefont
  {Perez-Garcia}},\ and\ \bibinfo {author} {\bibfnamefont {M.~M.}\ \bibnamefont
  {Wolf}},\ }\bibfield  {title} {\bibinfo {title} {Undecidability of the
  spectral gap},\ }\href@noop {} {\bibfield  {journal} {\bibinfo  {journal}
  {Nature}\ }\textbf {\bibinfo {volume} {528}},\ \bibinfo {pages} {207}
  (\bibinfo {year} {2015})}\BibitemShut {NoStop}%
\bibitem [{\citenamefont {Bausch}\ \emph {et~al.}(2020)\citenamefont {Bausch},
  \citenamefont {Cubitt}, \citenamefont {Lucia},\ and\ \citenamefont
  {Perez-Garcia}}]{bausch2020undecidability}%
  \BibitemOpen
  \bibfield  {author} {\bibinfo {author} {\bibfnamefont {J.}~\bibnamefont
  {Bausch}}, \bibinfo {author} {\bibfnamefont {T.~S.}\ \bibnamefont {Cubitt}},
  \bibinfo {author} {\bibfnamefont {A.}~\bibnamefont {Lucia}},\ and\ \bibinfo
  {author} {\bibfnamefont {D.}~\bibnamefont {Perez-Garcia}},\ }\bibfield
  {title} {\bibinfo {title} {Undecidability of the spectral gap in one
  dimension},\ }\href@noop {} {\bibfield  {journal} {\bibinfo  {journal}
  {Physical Review X}\ }\textbf {\bibinfo {volume} {10}},\ \bibinfo {pages}
  {031038} (\bibinfo {year} {2020})}\BibitemShut {NoStop}%
\bibitem [{\citenamefont {Shiraishi}\ and\ \citenamefont
  {Matsumoto}(2021)}]{shiraishi2021undecidability}%
  \BibitemOpen
  \bibfield  {author} {\bibinfo {author} {\bibfnamefont {N.}~\bibnamefont
  {Shiraishi}}\ and\ \bibinfo {author} {\bibfnamefont {K.}~\bibnamefont
  {Matsumoto}},\ }\bibfield  {title} {\bibinfo {title} {Undecidability in
  quantum thermalization},\ }\href@noop {} {\bibfield  {journal} {\bibinfo
  {journal} {Nature communications}\ }\textbf {\bibinfo {volume} {12}},\
  \bibinfo {pages} {5084} (\bibinfo {year} {2021})}\BibitemShut {NoStop}%
\bibitem [{\citenamefont {Bausch}\ \emph {et~al.}(2021)\citenamefont {Bausch},
  \citenamefont {Cubitt},\ and\ \citenamefont
  {Watson}}]{bausch2021uncomputability}%
  \BibitemOpen
  \bibfield  {author} {\bibinfo {author} {\bibfnamefont {J.}~\bibnamefont
  {Bausch}}, \bibinfo {author} {\bibfnamefont {T.~S.}\ \bibnamefont {Cubitt}},\
  and\ \bibinfo {author} {\bibfnamefont {J.~D.}\ \bibnamefont {Watson}},\
  }\bibfield  {title} {\bibinfo {title} {Uncomputability of phase diagrams},\
  }\href@noop {} {\bibfield  {journal} {\bibinfo  {journal} {Nature
  Communications}\ }\textbf {\bibinfo {volume} {12}},\ \bibinfo {pages} {452}
  (\bibinfo {year} {2021})}\BibitemShut {NoStop}%
\bibitem [{\citenamefont {Epstein}(1992)}]{epstein1992word}%
  \BibitemOpen
  \bibfield  {author} {\bibinfo {author} {\bibfnamefont {D.~B.}\ \bibnamefont
  {Epstein}},\ }\href@noop {} {\emph {\bibinfo {title} {Word processing in
  groups}}}\ (\bibinfo  {publisher} {CRC Press},\ \bibinfo {year}
  {1992})\BibitemShut {NoStop}%
\bibitem [{\citenamefont {Baumslag}\ \emph {et~al.}(1991)\citenamefont
  {Baumslag}, \citenamefont {Gersten}, \citenamefont {Shapiro},\ and\
  \citenamefont {Short}}]{baumslag1991automatic}%
  \BibitemOpen
  \bibfield  {author} {\bibinfo {author} {\bibfnamefont {G.}~\bibnamefont
  {Baumslag}}, \bibinfo {author} {\bibfnamefont {S.~M.}\ \bibnamefont
  {Gersten}}, \bibinfo {author} {\bibfnamefont {M.}~\bibnamefont {Shapiro}},\
  and\ \bibinfo {author} {\bibfnamefont {H.}~\bibnamefont {Short}},\ }\bibfield
   {title} {\bibinfo {title} {Automatic groups and amalgams},\ }\href@noop {}
  {\bibfield  {journal} {\bibinfo  {journal} {Journal of Pure and Applied
  Algebra}\ }\textbf {\bibinfo {volume} {76}},\ \bibinfo {pages} {229}
  (\bibinfo {year} {1991})}\BibitemShut {NoStop}%
\bibitem [{\citenamefont {Burillo}\ and\ \citenamefont
  {Elder}(2015)}]{burillo2015metric}%
  \BibitemOpen
  \bibfield  {author} {\bibinfo {author} {\bibfnamefont {J.}~\bibnamefont
  {Burillo}}\ and\ \bibinfo {author} {\bibfnamefont {M.}~\bibnamefont
  {Elder}},\ }\bibfield  {title} {\bibinfo {title} {Metric properties of
  baumslag--solitar groups},\ }\href@noop {} {\bibfield  {journal} {\bibinfo
  {journal} {International Journal of Algebra and Computation}\ }\textbf
  {\bibinfo {volume} {25}},\ \bibinfo {pages} {799} (\bibinfo {year}
  {2015})}\BibitemShut {NoStop}%
\bibitem [{\citenamefont {Lyndon}\ \emph {et~al.}(1977)\citenamefont {Lyndon},
  \citenamefont {Schupp}, \citenamefont {Lyndon},\ and\ \citenamefont
  {Schupp}}]{lyndon1977combinatorial}%
  \BibitemOpen
  \bibfield  {author} {\bibinfo {author} {\bibfnamefont {R.~C.}\ \bibnamefont
  {Lyndon}}, \bibinfo {author} {\bibfnamefont {P.~E.}\ \bibnamefont {Schupp}},
  \bibinfo {author} {\bibfnamefont {R.}~\bibnamefont {Lyndon}},\ and\ \bibinfo
  {author} {\bibfnamefont {P.}~\bibnamefont {Schupp}},\ }\href@noop {} {\emph
  {\bibinfo {title} {Combinatorial group theory}}},\ Vol.\ \bibinfo {volume}
  {188}\ (\bibinfo  {publisher} {Springer},\ \bibinfo {year}
  {1977})\BibitemShut {NoStop}%
\end{thebibliography}%
	
	\appendix
	
	\section{Semigroup presentations} \label{app:group_props}
	
	In this section we review some basic notions about discrete semigroup presentations and prove a few small results mentioned in the main text. 
	
	Formally, a discrete semigroup $G$ is a set equipped with a binary associative operation; when $G$ is a group it additionally has a distinguished element that acts as the identity, and each member of the set has a corresponding inverse. It is common to discuss a semigroup $G$ in terms of a specific set of generators $S$ and relations $R$ between them, writing 
	\be G = {\sf semi}\lan S \,| \,R \ran \ee 
	to signify this relationship (in the main text, when $G$ is a group we will write simply $G = \lan S \, | \, R\ran$ and omit inverse generators and the identity from $S$, as well as trivial relations involving the identity from $R$). For example, one might think of the group $\zz$ as being defined by a single generator $S = \{\ttx\}$ which obeys no nontrivial relations. This is however too narrow of a viewpoint, since it is possible for different choices of $S$ and $R$ to produce the same semigroup. As an example, consider the groups 
	\bea G_1 & = \lan \ttx  \, | \, \ran \\ 
	G_2 & = \lan  \ttx , \tty \, | \, \ttx^m = \tty^n, \, \ttx\tty = \tty \ttx \ran \eea 
	where $m,n$ are relatively prime. 
	These two groups are isomorphic, $G_1\cong G_2 \cong \zz$, with the isomorphism associating an element $\ttx^a\tty^b$ with the integer $am+bn$; note that this is true despite the fact that $G_1$ and $G_2$ have a different number of generators and relations. 
	
	A given semigroup in general admits an infinite number of different presentations,
	but below we will prove that the group-theoretic functions defined in the main text---the Dehn function, expansion length, and so on---have asymptotic scaling behaviors that are presentation independent.
	
	We can exploit this fact to choose presentations satisfying some particular desired property. For example, we may be concerned with choosing a model of dynamics where the Hamiltonian or unitary gates are as local as possible. Since the locality of the dynamics is limited by the maximum size of the relations in $R$, we would thus like to minimize the size of the relations. To this end we have the following proposition: 
	
	\begin{proposition}
		Every finitely generated group has a presentation $\lan S | R\ran$ where all relations $r_i \in R$ have length $|r_i| \leq 3$.
	\end{proposition}
	\begin{proof}
		Consider the Cayley 2-complex $\cg_G$ of a finitely presentable group $G$ (see Sec.~\ref{sec:bs} for a brief discussion of its definition) obtained from a finite presentation $P$. While the exact structure of $\cg_G$ depends on $P$, $\cg_G$ can always be subdivided to obtain a simplicial complex where each 2-cell has 3 edges. Since each 2-cell in the complex corresponds to a relation, $|r_i| \leq 3$ for all $r_i \in R$ in the subdivided complex, thereby defining a presentation $P'$ whose relations all have length $\leq 3$. Since $P$ is finite this subdivision is completed after only a finite number of steps, and $P'$ is consequently also finite. 
	\end{proof}
	
	Note that while the Cayley 2-complex of a {\it semi}group can also be subdivided, the lack of translation invariance in a semigroup's Cayley complex means that the resulting subdivision may yield a presentation with infinitely many generators (an illustrative example is to compare the semigroup $\nn \times \nn$ with the group $\zz\times \zz$). 
	
	As mentioned above, the group theoretic properties we are interested in from the point of view of group dynamics are largely insensitive to the choice of presentation. For example, recall the growth function $N_K(L) \triangleq |\{ g,\, |g| \leq L\}|$ defined in \eqref{groupgrowth} of the main text, which measures the number of dynamical sectors as a function of system size. The scaling of $N_K(L)$ with $L$ is independent of the choice of presentation, allowing us to meaningfully talk about the growth function of a semigroup, rather than of a presentation:
	\bpro \label{prop:lambdaindep}
	Let $N_{K,P}$ denote the growth function for a particular presentation $P = {\sf semi}\lan S | R\ran$ of $G$. Then 
	\be N_{K,P} \sim N_{K,P'} \ee 
	for all finite presentations $P,P'$ of $G$. 
	\epro 
	\bproof
	Let $P = {\sf semi}\lan a_1, \dots, a_{|S|}  \,|\, R\ran$ and $P' ={\sf semi}\lan a_1',\dots,a'_{|S'|} \, | \, R'\ran$. Then since both $P,P'$ present $G$, each $a_i'$ can be expressed as a product of a finite number of $a_i$. Let $n_{PP'}$ denote the maximal number of generators of $P$ that appear when writing the $a_i'$ in terms of these generators. Let also $|g|_P$ denote the geodesic distance of $g\in G$ with respect to the presentation $P$. Then $|g|_{P'} \leq n_{PP'} |g|_P$. Thus 
	\be N_{K,S'}(L) \leq N_{K,S}(n_{PP'}L).\ee 
	We may also perform a similar rewriting of generators of $P$ in terms of those of $P'$.  After running the same argument, we find that there exist $O(1)$ constants $n_{PP'}, n_{P'P}$ such that 
	\be N_{K,{P'}}(L/n_{P'P}) \leq N_{K,P}(L) \leq N_{K,P'}(n_{PP'}L),\ee 
	and hence $N_{K,P}\sim N_{K,{P'}}$. 
	\eproof 
	
	Similar reasoning can be applied to show that the Dehn function and expansion length (see Sec.~\ref{sec:general_semigroup} or the following appendix for definitions) of a semigroup are presentation-independent: 
	\bpro\label{prop:dpresenindep}
	Let $\dehn_P(L)$, $\el_P(L)$ be the Dehn function and expansion length of a semigroup with a finite presentation $P = \lan S | R\ran$. Then if $P,P'$ are any two such finite presentations, 
	\bea \dehn_P(L) & \sim \dehn_{P'}(L) \\ 
	\el_P(L) & \sim \el_{P'}(L). \eea 
	\epro 
	
	\ms 

	\section{A primer in geometric group theory} \label{app:primer}
	
	In this appendix we state and prove some useful facts about the geometry and complexity of finitely presentable discrete {\it groups}. 
	Many of the statements derived below are well-known results in the math literature, and we have tried to provide citations to the original works when appropriate. A particularly accessible review of background material relevant to the discussion to follow can be found in Ref.~\cite{clay2017office}; a more advanced reference is Ref.~\cite{sapir2011asymptotic}. As a small notational convenience, in the following the notation $w \sim g$ will be used as shorthand to denote that the word $w$ evalutes to $g$; in the main text this was written as $\vp( \k w) = g$: 
	\be w\sim g \quad \lra \quad  \vp(\k w) = g.\ee   
	
	We will mostly be interested in infinite groups, since finite ones have trivial large-scale geometry (in a sense soon to be made precise). Finitely presentable infinite semigroups are of course very easily constructed; indeed it is easily verified any group presentation where the number of generators exceeds the number of nontrivial relations\footnote{By a ``trivial relation'' we mean a free reduction/expansion, namely a relation of the form $\tta \tta\inv = \tte$.} will generate an infinite group. Free groups (on $n>1$ generators) and Abelian groups in some sense define opposite extremes, since the free group has a Cayley graph that is embeddable in the hyperbolic space $\mathbb{H}^n$, while Abelian groups have Cayley graphs that are embeddable in $\mathbb{R}^n$. Most of the interesting cases for us correspond to when an intermediate amount of `Abelian-ness' is introduced to the non-Abelian free group. 
	
	\sss*{Time complexity: the Dehn function}

	The definition of the Dehn function \eqref{dehng} in the main text relates only to the (worst-case) complexity of deforming a given word $w\in K_e$ into the identity word (recall that $K_{g, L}$ is the set of length-$L$ words that represent the element $g$, namely $K_g = \{ w \, | \, |w| = L, g(w) = g\}$). We claimed in the main text that focusing on the complexity of words in $K_e$---as opposed to $K_g$ for $g\neq e$---was done without loss of generality. We now prove that studying the complexity of the word problem in $K_{g\neq e}$ does indeed not produce anything which is not already captured by $\dehn(L)$: 	\bpro \label{prop:sectorindep}
	For a given element $g$ of geodesic distance $|g| \leq L$, define the $g$-sector Dehn function as 
	\be \dehn_g(L) \triangleq \max_{w,w'\in K_{g, L}} \dehn(w,w')\ee
	where $\dehn(w,w')$ is the minimum number of applications of relations needed to transform $w$ into $w'$.  
	Then 
	\be \dehn_g(L) \sim \dehn_e(L) \triangleq \dehn(L)\ee 
	for all $g$. 
	\epro 
	\bproof For two words $w_{1,2}$ in the same $K_g$ sector, any deformation (aka {\it based homotopy}) of $w_1$ to $w_2$ gives a deformation between the length-$2L$ word $w_1w_2\inv \in K_e(2L)$ and $e$. Thus the minimal number of steps needed to relate $w_1$ to $w_2$ cannot be asymptotically smaller than the minimal number of steps needed to deform $w(w')\inv$ to $e$. This implies 
	\be \dehn(w(w')\inv,\tte)\lesssim \dehn(w,w')\ee 
	where $\lesssim$ denotes equivalence up additional contributions linear in $L$.\footnote{These additional contributions occur because we need $O(L)$ steps just to get rid of trivial words like $ww\inv$. However, an extra $L$ steps can be important when we do not allow words of length greater than $L$ to appear at any stage of the deformation (which we very well may want to do on physical grounds).} This means that $\dehn(L) \lesssim \dehn_g(L)$. Conversely, since we can deform $e$ to $w_1\inv w_2\sim e$ in time $\sim \dehn(w_1w_2\inv)$, $w_1$ can be deformed into $w_1(w_1\inv w_2) = w_2$ in time $\lesssim \dehn(L)$. Thus we also have $\dehn(w,w') \lesssim \dehn(w_1w_2\inv)$, so up to factors of order $L$, we have 
	\be \dehn(w_1w_2\inv) \sim \dehn(w,w') \implies \dehn_g(L) \sim \dehn(L).\ee
	\eproof 
	The above argument means that---if the amount of space available to us is not restricted---all sectors have asymptotically the same worst-case complexity. If we however require that all words involved have length $\leq L$, Prop.~\ref{prop:sectorindep} is modified to read 
	\be \dehn_g(L) \lesssim \dehn(L).\ee 
	In this case, all sectors $K_{g, L}$ for which the geodesic length\footnote{Recall that the geodesic length of a group element is defined by the length of the shortest word representing $g$, namely $|g| = \min\{w\, : \, \vp(\k w) = g\}$.} of $g$ satisfies $|g|/L < 1$ in the $L\ra\infty$ limit will still have $\dehn_g(L) \sim \dehn(L)$. On the other hand, sectors where $|g|/L =1$ as $L\ra\infty$ will not have enough ``free space'' for the above argument to work, and will thus have $\dehn_g(L) < \dehn(L)$---each word in this sector must have at least $|g|$ locations allotted to represent $g$, giving it less room to form large loops in the Cayley graph with the remaining locations (in the case where fragile fragmentation occurs [see Sec.~\ref{sec:ff}], this remains true provided we define $\dehn(w,w') = 0$ if $w_1,w_2$ are in different subsectors of $K_{g, L}$). Sectors where $|g|/L =1$ as $L\ra\infty$ are however necessarily exponentially smaller in size than those with $|g|/L < 1$ (see Prop.~\ref{prop:ksizes}), and hence a random state---which with high probability is associated to a geodesic word where $L-|g|$ diverges as $L\ra\infty$---will be in a sector with $\dehn_g(L) \sim \dehn(L)$ with high probability. 
	
	We now provide examples of simple groups and their Dehn functions as well as discuss which sorts of groups we can expect to have large Dehn functions. A basic result is that only infinite non-Abelian groups can have interesting Dehn functions, due to the following easily verified statements: 
	\begin{fact}\label{ex:finitedehn}
		All finite groups have $\dehn(L) \lesssim CL$ for some (presentation-dependent) constant $C$ related to the diameter of the group's Cayley graph.\footnote{Since finite groups have trivial geometry (in mathematical language, they are quasi-isometric to the trivial group), this provides another sense in which $O(L)$ contributions to Dehn functions are trivial. }
	\end{fact} 
	\begin{fact}
		$\dehn(L) \lesssim L^2$ for any Abelian group.\footnote{The $L^2$ comes from applying relations to move around generators past other ones and next to their inverses. }
	\end{fact}
	Groups which are {\it too} non-Abelian also have quadratic Dehn functions:
	\begin{proposition}
		Free groups have $\dehn(L) \sim L^2$. 
	\end{proposition} 
	\begin{proof}
		This follows from the fact that the Cayley graphs of free groups are trees, and thus any word $w\in K_e$ is a (potentially backtracking) path on the tree which contains no nontrivial loops. Such paths can be contracted to the trivial path by applying $O(L)$ relations $\ttx\ttx\inv = \tte\tte$, and $O(L^2)$ relations $\ttx\tte=\tte\ttx$.\footnote{In the mathematical literature, a convention is often used where $\dehn(w)$ is the number of relations needed to reduce $w$ to the identity $\tte$, allowing the length of the word to change in the process (so that e.g. $\tte\tte$ may be replaced by $\tte$). For physical models of group dynamics, where the total length of the word is always fixed, $\dehn(w)$ will generically be longer than the value obtained with this convention by a factor of $L$, corresponding to moves of the form $\tte \tta = \tta \tte$ that need to be performed during the reduction of $w$ to $\tte^L$. }
	\end{proof}
	
	To look for interesting Dehn functions we thus need infinite groups with a moderate amount of Abelianness, which (roughly speaking) possess nontrivial loops at all length scales. One can start by finding examples of groups where $\dehn(L)$ scales as a higher order polynomial. It was shown in Ref. \cite{gersten2003isoperimetric} that these turn out to be virtually nilpotent groups. The simplest example is: 
	\begin{example}
		The discrete Heisenberg group 
		\be \sfH_3 = \lan x,y,z\, |\, [x,y]z\inv, \, [x,z], \, [y,z] \ran \ee 
		has $\dehn(L) \sim L^3$ \cite{gersten2003isoperimetric}. 
	\end{example} 
	
	The group $\bs(1,2)$ \cite{baumslag1962some} studied in detail in the main text is the simplest group known to the authors with exponential Dehn function, whose properties we discuss in detail in App.~\ref{app:bs_details}.  
	
	\sss*{Space complexity: the expansion length}
	
	As discussed in the main text, the {\it space} complexity of the word problem is given by the maximal size of words that one must encounter when reducing a $w \in K_e$ to the identity, as measured by the expansion length $\el(w)$, a quantity originally introduced by Gromov in Ref.~\cite{gromov1992asymptotic}. 
	If $\el(w) > |w|$, then $w$ must expand by a nontrivial amount while being reduced to the identity. 
	
	As we did with the area, we may use the expansion length to define a distance metric between any two words $w_{1,2}$ that represent the same group element. Instead of using $\el(w_1w_2\inv)$ (which we do not do on account of the homotopy relating $w_1$ to $w_2$ needing to be properly based), we define the {\it relative expansion length} $\el(w,w')$ between two words $w_1\sim w_2$ by replacing $e$ by $w_2$ in the above definition: 
	\be \el(w,w') \triangleq \min_{\{\De_{w_1\ra w_2}\}} \max_t |\De_{w_1\ra w_2}(t)|,\ee 
	where the $\De_{w_1\ra w_2}(t)$ are all paths in the Cayley graph with endpoints fixed at $e$ and $g(w_1) = g(w_2)$. When the homotopy is trivial, i.e. when $w_1 = w_2$, we define the relative expansion length to vanish.

	The expansion length of a {\it group} can be defined as the worst-case spatial complexity of words in $K_e$, as was done in our definition of the Dehn function: 
	\be \el(L) \triangleq \max_{w \in K_{e, L}} \el(w).\ee 
	A comprehensive survey of geometric properties of $\el$ can be found in Ref.~\cite{riley2006filling}.
	
	Similarly to the Dehn function, the asymptotic scaling of $\el(L)$ is independent of the choice of (finite) presentation. 
	Also as with the Dehn function, considering expansion in other $K_g$ sectors does not yield anything new. Just as with proposition~\ref{prop:sectorindep}, one can similarly show that 
	\bpro \label{prop:explength_sectorindep}
	For a given $g$ of geodesic distance $|g|\leq L$, define the expansion lengths 
	\be \el_g(L) \triangleq \max_{w_1,w_2\in K_{g, L}} \el(w,w').\ee 
	Then 
	\be \el_g(L) \sim \el_e(L) \triangleq \el(L)\ee 
	for all $g$. 
	\epro 
	
	\ms 
	
	Most groups one is familiar with have $\el(L) \lesssim L$. For example, it is easy to see that all Abelian groups have $\el(L) \lesssim  CL$ for some constant $C$. One can also show 
	\bpro 
	All finite groups have $\el(L) < L + C$ for some constant $C$. 
	\epro 
	\bproof
	The proof proceeds according to the same one that would be used in demonstrating the correctness of example~\ref{ex:finitedehn}. Let $w \in K_e$ be a length-$L$ closed loop in the Cayley graph of a finite group, and let $N$ be the number of vertices in the Cayley graph. Then since $N$ is finite, we may write 
	\be w = \prod_j w_j, \qq w_j \sim e, \, |w_j| \leq N,\, \, \forall \, j, \ee 
	since a path in the Cayley graph can only reach at most $N$ different vertices before returning to its starting point. Let $M = \el(N)$. If we append $M$ identity characters to the end of $w$, we can use them to turn any one of the $w_j$ into $e$ without increasing the length of the appended word. Since we can do this for all of the $w_j$, we thus have $\el(L) = L + M$ as claimed. 
	\eproof   
	Furthermore, we will see later that despite having exponential time complexity ($\dehn(L) \sim 2^L$), $\bs(1,2)$ only has linear spatial complexity ($\el(L) \sim L$). This is part of a more general result that {\it asychronously combable groups} (which $\bs(1,2)$ is an example of) have $\el(L) \lesssim  L$ \cite{gersten1993isoperimetric}.
	
	It is however relatively straightforward to construct examples where $\el(L)$ grows faster than linearly. This is done simply by finding groups with Dehn functions that scale as $\dehn(L) = \omega(2^L)$. This is due to the ``spacetime'' bound mentioned in the main text:
	\begin{proposition}[\cite{gersten2002filling}]
		For a finitely presentable group generated by $n_g$ generators, 
		\be\label{spacetimebound} \dehn(L) \lesssim (2n_g + 1)^{\el(L)}.\ee 
	\end{proposition}
	This can be rewritten as $\el(L) \gtrsim \log_{2n_g+1} \dehn(L)$, which grows superlinearly if $\dehn(L)$ grows super-exponentially. 
	\bproof 
	For $w\in K_e$ with expansion length $\el(w)$, the number of words that $w$ can possibly visit as it is reduced to the identity is $|\{ a_i, a_i\inv , e\}|^{\el(w)} = (2n_g + 1)^{\el(w)}$. Since the shortest reduction of $w$ to $e$ cannot visit a given word $w'$ more than once, the number of steps in the reduction is (often very loosely) upper bounded by $(2n_g + 1)^{\el(w)}$. 
	\eproof 
	
	\sss*{Growth rates of groups} 
	
	A simple measure of a group's geometry is how fast the group grows, namely how the number of elements within a distance $L$ of the origin of the Cayley graph grows as $L$ is increased. To this end, let $\mcb(L)$ denote the radius-$L$ ball centered at the origin of the Cayley graph, namely 
	\be \mcb(L) \triangleq \{g \, |\, |g| \leq L\},\ee 
	and define the {\it growth function} (or {\it group volume}) as $N_K(L) \triangleq |\mcb(L)|$. Physically, $N_K(L)$ places a lower bound on the number of Krylov sectors that $\dyn_G$ possesses (in the absence of fragile fragmentation, the number of Krylov sectors equals $N_K(L)$). 	
	A basic property of a group is the asymptotic scaling of $N_K$ with $L$. It is straightforward to $N_K\sim O(1)$ for finite groups and $N_K \sim {\rm poly}(L)$ for Abelian groups.
	Free groups provide the simplest examples where $N_K \sim \exp (L)$, which is the maximum possible growth rate.\footnote{Groups with growth rate intermediate between polynomial and exponential exist --- the ``simplest'' examples are groups arising from automorphisms of trees \cite{grigorchuk2008groups} --- but it is conjectured that all such groups have only infinite presentations. }
	
	We now ask what implications group expansion has for the time and space complexity measures introduced above. It turns out that exponential growth is needed in order to have $\dehn(L)$ growing faster than ${\rm poly}(L)$: 
	\begin{proposition} \label{thm:gromov} 
		If a finitely presented
		group $G$ possesses a superpolynomial Dehn function, then $G$ has exponential growth.
	\end{proposition} 
	\bproof 
	The contrapositive of this proposition follows from Gromov's theorem that all finitely generated groups with polynomial growth are virtually nilpotent (namely have a nilpotent subgroup of finite index) \cite{gromov1981groups}. Since virtually nilpotent groups have the same Dehn functions as nilpotent ones, we may combine Gromov's theorem with the fact that all nilpotent groups have $\dehn(L) \sim L^d$ for some $d$ \cite{gersten2003isoperimetric} to arrive at the result. 
	\eproof 
	
	Note that the converse to proposition~\ref{thm:gromov} is obviously false, as free groups on more than one generator provide examples of groups with exponential growth but with polynomial (in fact $\dehn(L) \sim L$) Dehn functions. 
	
	\ms 
	
	We would also like to know the sizes of the different sectors $K_g$. 
	One result along these lines is that the sector sizes must get small as $|g|$ gets large: 
	\bpro \label{prop:ksizes}
	Define 
	\be \mcd \triangleq \sum_{g \, | \, |g|\leq L} |K_{g, L} | = (2n_g+1)^L\ee 
	as the total number of words of length $L$.
	Then the size of $K_{g, L}$ is upper bounded as 
	\be \frac{|K_{g, L}|}\mcd \leq C \exp\( - c \frac{|g|^2}L \)\ee
	for some  $g,L$-independent constants $C,c$. 
	\epro 
	
	In the context of groups (like $\bs$) with exponential growth, this tells us that almost all Krylov sectors are exponentially smaller than the largest ones. 
	
	\bproof 
	The proof follows from connecting the counting of walks in $K_{g, L}$ with the heat kernel on the Cayley graph. Consider the symmetric simple lazy walk on the Cayley graph, and let $p_L(g,h)$ be the probability that a length-$L$ walk starting at $h$ ends at $g$. Then 
	\bea p_L(g,h)  & = \frac1{2n_g+1} \( p_{L-1}(g,h) + \sum_{k \in \p g} p(k,h)\) \\
	& = p_{L-1}(g,h) + \frac1{2n_g+1} \sum_{k\in \p g} (p(k,h) - p(g,h)),\eea 
	where the sum over $k$ runs over the neighbors of $g$ in the Cayley graph. By translation invariance of the Cayley graph, we can fix $h=e$ without loss of generality, and will simply write $p_L(g)$ for $p_L(g,e)$. 
	
	The previous equation can be written more succinctly as 
	\be \d_L p = \De p,\ee 
	where $\De$ is the normalized graph Laplacian and $\d_L$ denotes the discrete derivative along the ``time'' direction determined by $L$. Thus the different sizes of the $K_{g, L}$ are determined by using the heat equation to evolve a delta function concentrated on $e$ for a total time of $L$. The claim we are trying to prove then follows from estimates of the discrete heat kernel Greens function developed in the graph theory literature, see \cite{davies1992heat} for a review. 
	\eproof 
	Various other facts follow from the observation that $p_L(x)$ obeys the heat equation. For example, it implies that $p_L(x)$ obeys strong maximum and minimum principles, which guarantees that all local maxima and minima of $p_L(x)$ occur on the boundaries of its domain of definition. It also means that since in groups of exponential growth almost all group elements in $\mcb(L)$ have geodesic dist stance close to $L$, almost all sectors $K_{g, L}$ will contain a number of elements exponentially smaller than $\mcd$. 
	
	In addition to the above asymptotic bound, we can also prove that $K_e$ is always the largest sector: 
	\bpro \label{prop:largest_sector}
	For all $L$ and all $g, \, |g|\leq L$, we have\footnote{Note that in order for this proposition to be true, we need the words in $K_g$ to be drawn from the alphabet $S \cup S\inv \cup \{e\}$, where $S$ is the generating set. If we were to just use $S\cup S\inv$ then it obviously cannot be strictly true, as in this case one could e.g. have a bipartite Cayley graph such that $|K_{e, L}| = 0$ for odd $L$.}
	\be |K_{e, L}| \geq |K_{g, L}|.\ee 
	\epro 
	
	\bproof 
	We aim to show that $p_L(e,e) \geq p_L(e,g)$ for all $g,L$. For simplicity of notation, let $L\in 2\nn$. Then using $p_L(g,h) = p_L(gk,hk)$ for all $g,h,k$ on account of transitivity of the Cayley graph, we have 
	\bea p_L(e,e) & = \sum_g p_{L/2}(e,g)^2 \\
	& = \sqrt{\sum_h p_{L/2}(e,h)^2} \sqrt{\sum_{h'} p_{L/2}(g,h')^2} \\
	& \geq \sum_h p_{L/2}(e,h) p_{L/2}(h,g) \\ 
	& = p_L(e,g),\eea 
	where in the third line we used Cauchy-Schwarz.
	\eproof   
	
	It is also possible to make statements about the absolute size of $K_e$. In particular, $|K_e|$ admits different bounds depending on the growth rate of the group:  
	\begin{proposition}[\cite{woess2000random}] \label{fact:bounds_on_mcke}
		For a group with growth function $N_K(L)$ scaling as a polynomial of $L$,  
		\be L^{d_1} \lesssim V(L) \lesssim L^{d_2} \implies (L\log L)^{-d_2/2} \lesssim   \frac{|K_e|}{\mcd } \lesssim L^{-d_1/2}. \ee 
		For a group with exponential growth, 
		\be \frac{|K_e|}\mcd  \lesssim e^{-L^{1/3}}.\ee 
	\end{proposition} 
	In the main text, we saw numerically that this bound is saturated for $\bs$ (see Fig.~\ref{fig:id_sector_scaling}). This means that while $\dyn_\bs$ has exponentially slow relaxation, it is as `close' to being weakly fragmented as possible.  
	
	We can also connect the above bound with a previous result to derive: 
	\begin{corollary}
		Groups with superpolynomial Dehn functions have exponentially many group sectors for words of a fixed length, and have an identity sector $K_e$ that contains a fraction of all length-$L$ words that scales at most as $e^{-L^{1/3}}$. 
	\end{corollary}
	
	\bproof 
	This follows by combining proposition~\ref{fact:bounds_on_mcke} with proposition~\ref{thm:gromov}. 
	\eproof

	\section{Semigroup dynamics and undecidability} \label{app:undecide}

	While not relevant to the examples studied in the main text, it is amusing to note that the word problem for semigroups is not just computationally hard, but can even be undecidable. For example, a foundational result in computability theory is the undecidability of the semigroup word problem, as first proven by Markov~\cite{markov1948impossibility}. This result implies the existence of finitely-presented semigroups for which $\dehn(w,w')$ grows faster than any recursive function. For such groups, one can never be sure whether or not two words $\k w, \k{w'}$ are in the same $K_g$, showing that establishing the number of dynamical sectors is impossible in general. 
	
	Remarkably, rather simple examples of semigroups with undecidable word problems are known, with the required $|\mch_{\rm loc}|$ being as low as five. An explicit example from Ref.~\cite{tseitin1958associative} provides a rather simple example of a semigroup with an undecidable word problem on the five-element generating set $S = \{ \ttv,\ttw,\ttx,\tty,\ttz\}$, and the seven relations 
	\begin{align}
		R = \{ \ttv\ttx=&\ttx\ttv, \ttv\tty=\tty\ttv, \ttw\ttx=\ttx\ttw, \ttw\tty=\tty\ttw,\nonumber\\&\ttx\ttz=\ttz\ttx\ttv,\tty\ttz=\ttz\tty\ttw,\ttx^2\ttv=\ttx^2\ttv\ttz\}.
	\end{align}
	The word problem is also undecidable when one specifies further to the setting of a finitely presented {\it group}~\cite{novikov1955algorithmic, boone1959word}. Going further still, the {\it Adyan-Rabin theorem}~\cite{rabin1958recursive,adyan1955algorithmic} states that nearly all ``reasonable'' properties of finitely-presented groups---their Dehn times, whether or not they are finite or Abelian; even whether or not they are a presentation of the trivial group---is undecidable. 
	
	As a corollary, the existence of Hilbert space fragmentation itself is therefore undecidable: even if ergodicity looks to be broken for all system sizes below $L$, for some group presentations one can in general never be sure that it will not be restored at system size $L+1$. This result compliments recent work on the undecidability of physical problems relating to local Hamiltonians, e.g. computing spectral gaps, ground state phase diagrams, and so on \cite{cubitt2015undecidability,bausch2020undecidability,shiraishi2021undecidability,bausch2021uncomputability}.

	\section{No complete symmetry labels for non-Abelian $G$} \label{app:nosyms} 
	
	In this appendix, we will prove that the Krylov sectors of $\dyn_G$ can not be associated with the quantum numbers of any global symmetry if $G$ is non-Abelian (as is the case for all of the examples of interest). In fact we will actually prove a more general result. To state the result, we will define a {\it locality-preserving unitary} as any unitary operator $\mcu$ which is such that for all local operators $\mco$, the conjugated operator 
	\be \mco^\mcu \triangleq \mcu^\da \mco \mcu\ee
	is also local. The generators of any global symmetry are locality-preserving unitaries, but for us we will not need to assume that $\mcu$ commutes with $\dyn_G$. 
	
	We will prove the following result: 
	\begin{prop}
		When $G$ is non-Abelian, the Krylov sectors of $\dyn_G$ cannot be fully distinguished by the eigenvalues of any set of locality-preserving unitaries. 
	\end{prop} 
	
	\begin{proof} 
		We will argue by contradiction. Assume that there existed a set of unitaries $\mcu_a$ whose expectation values in a given computational basis product state $\k w$ allowed one to determine the group element $\vp(\k w)$ associated with $w$ (and therefore the $K_{g,L}$ that $\k w$ belongs to). Consider the states  
		\be \k{w_{\ttg\tth} } = \k{\tte^l \ttg \tte^m \tth \tte^n},\ee
		with $l,m,n$ all proportional to the system size $L$.
		We can write 
		\be \k{w_{\ttg\tth}} = \mathcal{O}_{\ttg,l} \mathcal{O}_{\tth, l+m+1} \ket{\tte^{l+m+n+2}},\ee 
		where $\mco_{\ttg,i} = \kb\ttg\tte_i + h.c$. Our goal will be to show that $\k{w_{\ttg\tth}}$ cannot be an eigenstate of any of the $\{\mcu_a\}$ if $\ttg\tth\neq \tth\ttg$. 
		
		Denoting $\ket{\tte^{l+m+n+2}}$ as $\ket{\tte}$ for simplicity, and defining $\k{w_{\ttg\tth;a}} =  \mathcal{O}_{\ttg,l} \mathcal{O}_{\tth, l+m+1} \mcu_a  \ket{\tte^{l+m+n+2}}$, we have 
		\bea 
		\lan w_{\ttg\tth;a} | \mcu_a | w_{\ttg\tth} \ran & = \lan\tte  | \mcu_a^\da \mco_{\ttg,l+1} \mco_{\tth,l+m+2}\mcu_a \mco_{\ttg,l+1} \mco_{\tth,l+m+2} |\tte\ran \\ 
		& = \lan\tte  | \mco^{\mcu_a}_{\ttg,l+1} \mco^{\mcu_a}_{\tth,l+m+2} \mco_{\ttg,l+1} \mco_{\tth,l+m+2}  |\tte\ran.
		\eea 
		Since we have assumed $m$ to be extensive and $\mcu_a$ to be locality-preserving, we will always have $m+1 >2 \, {\rm supp}(\mco^{\mcu_a}_{\ttg,i})$  for sufficiently large system sizes, and so the above expectation value splits as 
		\bea \lan w_{\ttg\tth;a} | \mcu_a | w_{\ttg\tth} \ran & = \lan \tte | \mco^{\mcu_a}_{\ttg,l+1}\mco_{\ttg,l+1} |\tte\ran \\ 
		& \qq \times  \lan \tte | \mco^{\mcu_a}_{\tth,l+m+2} \mco_{\tth,l+m+2} | \tte \ran. \eea 
		In particular, in the translation invariant case where $\lan \tte | \mco^{\mcu_a}_{\ttg,l+1}\mco_{\ttg,l+1} |\tte\ran$ is independent of $l$, we have 
		\bea  \lan w_{\ttg\tth;a} | \mcu_a | w_{\ttg\tth} \ran = \lan w_{\tth\ttg;a} | \mcu_a | w_{\tth\ttg} \ran.\eea 
		Taking $\ttg,\tth$ to be any two generators such that $\ttg\tth \neq \tth\ttg$ as group elements, the above then shows that states in different $K_{g,L}$ cannot be distinguished by eigenvalues of the $\{\mcu_a\}$, provided that i) the above assumption about translation invariance can be removed, and ii) $\k{w_{\ttg\tth}}$ is not orthogonal to $\k{w_{\ttg\tth;a}}$ for all choices of $l,m,$ and all choices of $\ttg,\tth$ such that $\ttg\tth\neq\tth\ttg$. i) is dealt with by noting that when the $\mcu_a$ are not translation invariant (as is the case for modulated symmetries), we may exploit the fact that the sectors to which $\k{w_{\ttg\tth}}$ and $\k{w_{\tth\ttg}}$ belong are independent of $l,m$, and hence $l,m$ can be varied without changing the eigenvalues of the respective states under the $\{\mcu_a\}$. For ii), we need only note that 
		\bea \lan w_{\ttg\tth;a} | w_{\ttg\tth} \ran & =  \lan \tte | \mcu_a^\da \mco_{\ttg,l}^2 \mco_{\tth,l+m+1}^2  | \tte \ran \\
		& =   \lan \tte | \mcu_a^\da (\proj{\tte}_l + \proj{\ttg}_l) \\ 
		& \qq \times (\proj{\tte}_{l+m+1} + \proj{\tth}_{l+m+1}) | \tte \ran \\
		& = \lan \tte | \mcu_a^\da | \tte \ran,\eea 
		which is nonzero since $\ket{\tte}$ is an eigenstate of $\mcu_a^\da$ by assumption (on account of its definite symmetry charge), and since unitarity guarantees that $\mcu_a^\da$ has no zero eigenvalues. 
	\end{proof}

	\section{The geometry and complexity of Baumslag-Solitar groups} \label{app:bs_details}
	
	In this section we state and prove some facts about the group geometry of the Baumslag-Solitar group $\bs(1,2)$ and some of its simple generalizations. The original work in which these groups were defined is Ref.~\cite{baumslag1962some}. We will always work with the presentation
	\be \bs(1,2) = \lan \tta,\ttb \, | \, \tta\ttb = \ttb\tta\tta\ran.\ee 
	Everything we say below can be readily modified to work for the generalixed groups $\bs(1,q) = \lan \tta,\ttb \, |\, \tta\ttb = \ttb\tta^q\ran$, but for concreteness we will specify to $q=2$ throughout, and as in the main text will write $\bs$ for $\bs(1,2)$. 
	A useful fact about our chosen presentation for $\bs$ is that it admits the the matrix representation \cite{clay2017office}
	\be \label{matrep} \tta = \bpm 1 & 1 \\ 0 & 1 \epm ,\qq  \ttb = \bpm 1/2 & 0 \\ 0 & 1\epm,\ee 
	which allows the word problem to be solved in linear time and is helpful in numerical studies of $\bs$'s group geometry.\footnote{For $\bs(1,q)$, the $1/2$ in $\ttb$ is replaced by $1/q$.}
	
	Our notation will carry over from the previous section on general aspects of group geometry. 
	We will also introduce the notation $n_b(w)$ to denote the net number of $b$s that appear in a word $w$, namely 
	\be n_b(w) \triangleq \sum_{i=1}^{|w|} (\d_{w_i,b} - \d_{w_i,b\inv}).\ee

	\ss*{Worst-case complexity}
	
	We will start by analyzing the worst-case time and space complexity of words in $\bs$, which we do by computing the Dehn and expansion length functions. 
	
	To orient ourselves, it is helpful to recall that the Cayley graph of $\bs$ is homeomorphic to the product of the real line with a 3-regular tree (for $\bs(1,q)$ it is a $(q+1)$-tree), with the depth of a given word $w$ along the tree being controlled by the relative number of $b$s and $b\inv$s that $w$ contains (see Fig.~\ref{fig:bs_geometry} of the main text).  
	We will refer to this tree as the {\it b-tree}, and will use terminology whereby multiplying by $b$ ($b\inv$) moves one ``up'' (``down'') on the sheet of the $b$-tree one is currently at, while multiplying by $a$ ($a\inv$) moves one  ``right'' (``left'') within the sheet (along the {\it a-axis}). Because of the hierarchical nature of the graph, multiplying by $b$ --- namely moving deeper into the $b$-tree --- moves one to ``larger scales'', while multiplying by $b\inv$ does the opposite. 
	
	\begin{figure}
		\centering
		\includegraphics[scale=0.4]{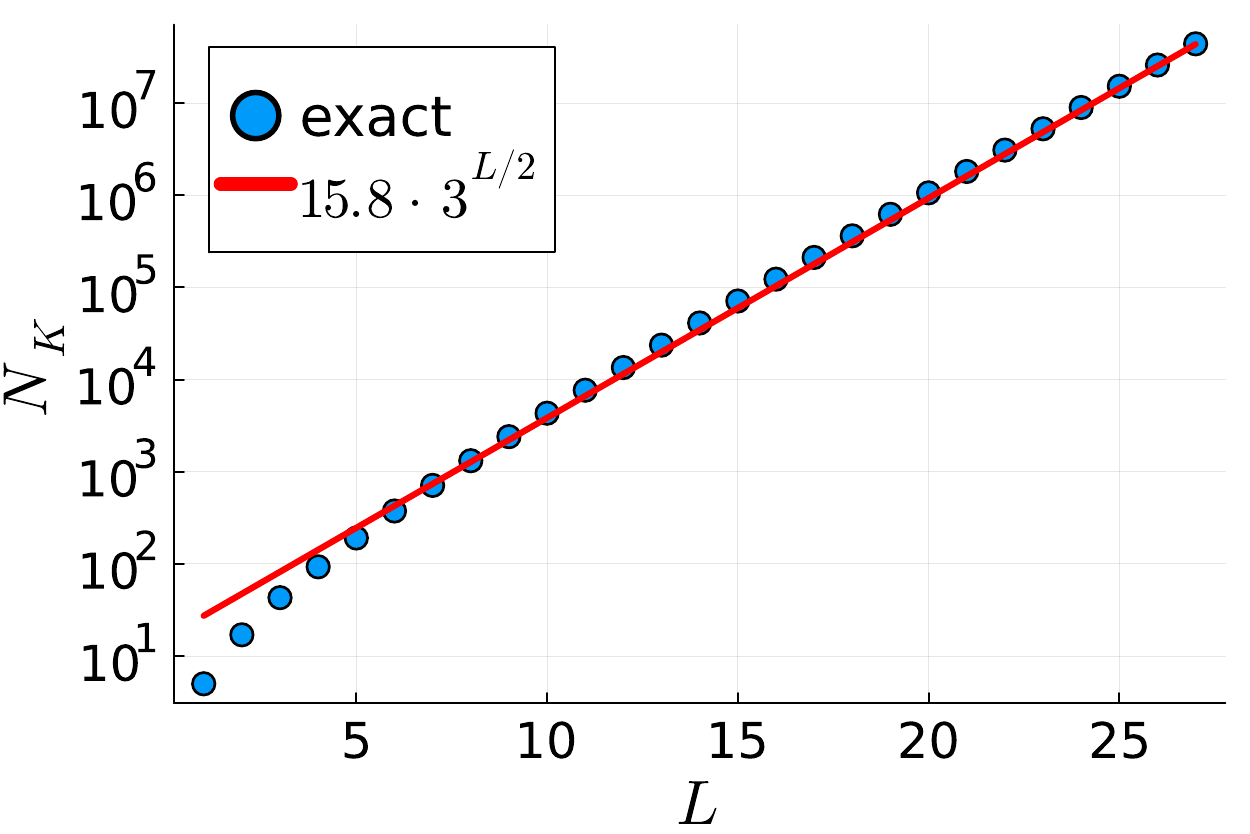}
		\caption{The number $N_K(L)$ of Krylov sectors, equal to the number of distinct $\bs$ group elements whose geodesic length is less than or equal to $L$.
		}
		\label{fig:bsgrowth}
	\end{figure}
	
	Because of the tree structure, it is clear that the group volume
	\be N_K(L) \sim \l^L\ee 
	grows exponentially with $L$ for some constant $\l>1$. A naive guess for the value of $\l$ is as follows. First, we realize that the exponential growth comes from the tree structure, and that motion by one node on this tree is always possible though the use of at most two group generators (since to perform an arbitrary move on the tree one must multiply by either $b$, $b\inv$, or $ab$). Since the number of points at depths $d \leq L$ of a 3-regular tree goes like $3^L$, we therefore estimate $N_K(L) \sim 3^{L/2}$, giving $\l \approx \sqrt3$. This estimate is actually extremely close to the numerically computed scaling, as we show in Fig.~\ref{fig:bsgrowth}.

	\sss*{Dehn function} 
	
	We now turn to computing the Dehn function and expansion length of $\bs$. The Dehn function must be large, since it is easy to construct words with $\dehn(w) \sim 2^{|w|}$ \cite{gersten1992dehn}: 
	
	\bpro \label{prop:wbig}
	Define $w_n \triangleq b^{-n} a b^{n}$, so that $w_n \sim a^{2^n}$.\footnote{$w_n$ is nearly a geodesic for $a^{2^n}$; the true geodesic with our presentation is instead $b^{-(n-1)} a^2 b^{n-1}$, which is one character shorter than $w_n$. } Then 
	the word 
	\be \label{wbig} w_{big} = w_n a w_n\inv a\inv\ee  
	has area 
	\be \label{bigarea} \dehn(w_{big}) = 2^{n+1}-2 \sim 2^{|w_{big}|}.\ee  
	
	\epro 
	
	A visual illustration of this statement is given in Fig.~\ref{fig:bs_geometry}. The proof, which is a condensed version of the original proof in Ref.~\cite{gersten1992dehn}, is as follows: 
	
	\bproof 
	
	To determine $\dehn(w_{big})$, we need to find the minimal number of relations needed to turn $w_{big}$ into the identity word. Geometrically, this corresponds\footnote{Up to the aforementioned $O(|w_{big}|^2)$ contributions.} to the number of 2-cells in the Cayley complex $\cg_\bs$ that form a minimal spanning surface $S_{w_{big}}$ with $w_{big}$ as its boundary. 
	Note that the loop defined by $w_{big}$ is entirely contained within two sheets of $\cg_\bs$. It is furthermore clear that if we only consider bounding surfaces $S_{w_{big}}$ contained within these two sheets, the minimal bounding surface has area 
	\be 2\sum_{k=0}^{n-1}2^k =  2^{n+1} - 2.\ee 
	Therefore we need only show that this area cannot be reduced by considering surfaces that extend into other sheets. The reader can verify that this is true, since a 2-cell of $S_{w_{big}}$ that lives on any other sheet will necessarily make an unwanted contribution to $\p S_{w_{big}}$. More formally, we can recognize that, being homeomorphic to the product of $\rr$ with a 3-tree, the Cayley 2-complex is contractible, and thus the $S_{big}$ found above is the unique minimal bounding surface. 
	\eproof

	Note that by considering a homotopy which shrinks $w_{big}$ down to the identity by first making the loop narrower along the $a$ axis before shrinking it along the $b$ axis, the length of the word does not parametrically increase. This is an intuitive explanation for the following fact: 
	\begin{fact}\label{fact:bsexp}[\cite{gersten1993isoperimetric}]
		$\bs$ has linear expansion length, $\el(L) \sim L$. 
	\end{fact} 
	
	The previous proposition immediately implies the existence of exponentially many (in $L$) elements of $K_{e, L}$ with exponentially large (also in $L$) area. The Dehn function thus grows at least as $\dehn(L) \gtrsim 2^L$. We now show that this bound is in fact tight. The existing proof of this fact in the mathematical literature appears to be to note that $\bs$ is an ``asychronously automatic'' group \cite{epstein1992word}, and that any such group has $\dehn(L) \lesssim 2^L$ \cite{baumslag1991automatic}. In the following, we give a more elementary proof:
	\bpro 
	$\bs$ has exponential Dehn function: 
	\be \dehn(L) \sim 2^L.\ee 
	\epro 
	\bproof 
	The construction above tells us that $\dehn(L) \gtrsim 2^L$. A matching upper bound can be proven using an algorithm which converts an input word to a particular standard form.
	
	This is done
	by simply moving all occurrences of $b$ in $w$ to the left and all occurrences of $b\inv$ to the right, duplicating $a$ and $a\inv$s along the way as needed and (optionally, for our present purposes) eliminating $bb\inv$ pairs as they are encountered. 
	
	It is clear that at most an exponential in $L$ number of additional $a$s are generated during this process of shuffling the $\ttb$s and $\ttb\inv$s around (as usual, this means exponential up to a polynomial factor, here linear in $L$), so that the resulting word is 
	\be \label{bs_canform} w' = \ttb^k w_a \ttb^{-l},\ee 
	where $w_a$ is a word containing only $\tte$, $\tta,\tta\inv$ and with length $|w_a| \lesssim 2^L$. Since $w\sim e$, we know that $k=l$ and that $w_a\sim \tte$. Thus an additional $2^L$ relations suffice to reduce $w_a$ to $\tte$ and hence $w'$ to $\tte$ -- thus, $2^L$ is also an upper bound on $\dehn(L)$. 
	\eproof 
	
	We saw in Fact~\ref{fact:bounds_on_mcke} that $\bs$ has linear expansion length, meaning that there exists constants $C,D$ such that $\el(w) \leq C L + D$ for all $w \in K_{e, L}$. The argument in Ref.~\cite{gersten1993isoperimetric} however does not tell us whether $C=1$ or $C>1$, which is need for determining whether or not $\bs$ exhibits fragile fragmentation. The following proposition answers this question in the affirmative: 
	\begin{proposition}\label{prop:bs_linear_exp}
		There exists constants $\a,D$ with $\a>0$ such that 
		\be \label{elbsapp} \el(L) \geq (1+\a)L + D. \ee 
	\end{proposition}
	\begin{proof}
		It is enough to show that $\el(w) \geq (1+\a)L+D$ for a particular length-$L$ word $w$, which we choose to be $w_{big}$ in \eqref{wbig} with $n$ such that $L = 4(n +1)$. By the contractibility of the Cayley 2-complex $\cg_\bs$, for all $0\leq m\leq 2^n$, any homotopy from $w_{big}$ to the identity must pass through the vertex $\tta^m$ of $\cg_\bs$. Therefore 
		\be \el(w_{big}) \geq \max_{0\leq m \leq 2^n} 2|\tta^m|,\ee
		where $|\tta^m|$ is the geodesic distance (not word length) of $\tta^m$. 
		
		We now argue that there exists an $2^{n-1} < m < 2^n$, an $\a>0$, and a constant $c_1$ such that $|\tta^m| > (2+2\a)n+c_1$. By the contractibility of $\cg_\bs$, any  geodesic of $\tta^m$ will be contained within a single sheet. Furthermore, since $m$ is exponentially large in $n$ in the worst case, the geodesic should reach a height of at least $n-c_2$ on the sheet, where $c_2$ is a sufficiently large $O(1)$ constant.  If the geodesic does not reach such a height, then it must make a much larger number of steps in the $\tta$ direction resulting in a larger perimeter. Letting $n_{|\ttb|}$ ($n_{|\tta|}$) be the number of $\ttb,\ttb\inv$ ($\tta,\tta\inv$) characters  that appear in the geodesic, this shows that $n_{|\ttb|} \geq 2(n-c_2)$, and so $\a \geq 0$. 
		
		To compute the perimeter, we have to determine the number of steps the smallest length word makes in the $\tta$ direction in order to reach $\tta^m$.  Since at its highest point $n_{|\ttb|}/2$ steps down along the $\ttb$ direction are needed, we can intersperse any of these $n_{|\ttb|}/2 \geq n - c_2$ steps with at least one step along the $\tta$ direction.  Note that if two consecutive steps along the $\tta$ direction are taken, these can be pulled past the previous $\ttb$ step to form a single $\tta$ step, thereby reducing the total length of the loop.  Therefore, the minimal length loop has $\ttb$ steps interspersed with at most a single $\tta$ step.  Since there are exponentially many (in $n$) choices of $m$, and since each $\tta^m$ has a unique geodesic (i.e. minimal sequence of $\tta$ and $\ttb$ steps), a simple application of the pigeonhole principle guarantees that $\exists m$ such that $\geq 0.98 n$ steps are along $\tta$.  This can be verified because the total number of placements of less than $0.98 n$ steps along $\tta$ in a sequence of $n_{|\ttb|} \geq (n-c_2)$ steps along $\ttb$ is $\ll 2^n$.  Therefore, there is a $c_4$ such that $\max_{0\leq m \leq 2^n} 2|\tta^m| \geq 2(2 + 0.98)n + c_4 = 5.96 n + c_4$, meaning that \eqref{elbsapp} is satisfied with $\a = 0.49$. 

	\end{proof}
	
	While this demonstrates that words in $K_{e, L}$ must expand by an amount proportional to $L$ before being mapped to the identity word, this statement is only meaningful in the large $L$ limit, and in practice $\el(L)/L$ can be very close to $1$ for modest choices of $L$ (as was seen in the numerics of Sec.~\ref{sec:numerics}). 
	
	\ss*{Average-case complexity}\label{app:bs_average} 
	
	We now turn to examining the {\it average}-case complexity of words in $\bs$ (or more precisely, the complexity of typical words in $\bs$). Unfortunately, only a few results on average-case Dehn functions are known (one example for nilpotent groups is Ref.~\cite{young2008averaged}), and average-case complexity has not been studied for groups with superpolynomial Dehn function. We emphasize that our analysis in this section is {\it not} rigorous and relies on several reasonable claims backed up by some numerical evidence. Letting $\typdehn(L)$ denote the Dehn function of a {\it typical} word in $K_{e, L}$, we claim
	\begin{claim}\label{claim:typdehn}
		With probability $1-\epsilon$, a randomly chosen word in $K_{e,L}$ for $\bs$ has a Dehn function 
		\be \typdehn(L) \gtrsim 2^{f(\epsilon) \sqrt L}.\ee 
		for $f(\epsilon) \to 0$ as $\epsilon \to 0$.
	\end{claim}
	{\it Rigorously } proving this claim---which we believe is within reach---constitutes an interesting direction for future research. 
	
	\sss*{Distribution of geodesic lengths for random words}
	
	$\typdehn(L)$ measures the typical area of words in $K_{e, L}$. The difficulty in computing $\typdehn(L)$ lies in the fact that it is hard to sample words from $K_{e, L}$, due to the constraint that such words form closed loops in the Cayley 2-complex $\cg_\bs$. A much easier task is to determine the geometry of typical length-$L$ paths in $\cg_\bs$, regardless of whether these paths form closed loops. Understanding this easier problem will allow us to build intuition for computing the scaling of $\typdehn(L)$. 
	
	The precise question we address is this: given a random length-$L$ word $w$, what is the geodesic distance $|g(w)|$ of $w$? We can efficiently obtain an answer numerically by randomly sampling words $w$ and approximately computing their geodesic distances. In order for this procedure to be efficient, it is helpful to realize that any word may always be brought into the following canonical form: 
	\begin{proposition}[\cite{burillo2015metric}]
		Any word in $\bs$ can always be mapped to a unique standard form 
		\be \label{wknl} w_{knl} = b^k a^n b^{-l},\qq k,n,l\in \zz, \, \, n,l\geq 0,\ee 
		where $n$ can be even only if at least one of $k,l$ are zero.\footnote{If both are nonzero then we can simplify $b^k a^{2n} b^{-l} \sim b^{k-1}a^n b^{-(l-1)}$.}
	\end{proposition}
	\begin{proof} 
		This follows from the arguments around \eqref{bs_canform} or use of the matrix representation \eqref{matrep}.
	\end{proof}

	Since $k,n,l$ are unique, the $w_{knl}$ serve as a set of canonical representatives for each $g$. To get the geodesic of an arbitrary word, we first reduce it to this form, and then make use of the fact that 
	\begin{proposition}[\cite{burillo2015metric}]
		The geodesic distance of a word $w$ is determined by the exponents $k,n,l$ appearing in its canonical form $w_{knl}$ as 
		\be \label{geoest} 
		\frac12(k+l+\log_2|n|) \leq 
		|g(w_{knl})| \leq 4(k+l+\log_2|n|+1)\ee 
		when $n\neq 0$, and $|g(w_{knl})| = |k-l|$ if $n=0$.
	\end{proposition}
	
	All that remains is to find a way of determining $k,n,l$ given an arbitrary word $w$. This is done using the matrix representation \eqref{matrep}, in which $w_{knl}$ becomes 
	\be w_{knl} = \bpm 2^{l-k} & n2^{-k} \\0 & 1 \epm.\ee 
	Therefore to find $k,n,l$ for an input string $w$, we proceed as follows. We first find the matrix corresponding to $w$ by explicit matrix multiplication, yielding a result of the form $\bpm A &  B \\ 0 & 1\epm$. We know right away that 
	\be \log_2 A = n_b(w) = k-l,\ee 
	the number of $b$s contained in $w$ minus the number of $b\inv$s. Then: 
	\begin{itemize}
		\item If $B\in \zz$, either $k=0$ or $l=0$. Which one of these scenarios holds depends on $\sgn(n_b)$: if $n_b<0$ then $k=0, l = -n_b$ and $n = B$, while if $n_b > 0$ then $l=0, k=n_b$ and $n=2^{n_b}B$. 
		\item If $B \not\in \zz$, then $k>0$, and $k$ is determined by the number of significant figures after the decimal point when $B$ is represented in binary,\footnote{There is a small subtlety here, since this approach is incorrect if $l=0$ and $n = 2^m n_o$ with $m>0, n_o \in 2\zz+1$. However in this case it is easy to show that the naive values of $(k,n,l)$ one obtains from this procedure in this case are $(k,n,l)_{\rm naive} = (k-m, n_0, -m)$. Since $l_{\rm naive} < 0$, this mistake is easy to identify and correct for.} which then allows one to determine $l$ and $n$. 
	\end{itemize}
	
	\begin{figure}
		\centering 
		\includegraphics[width=.35\tw]{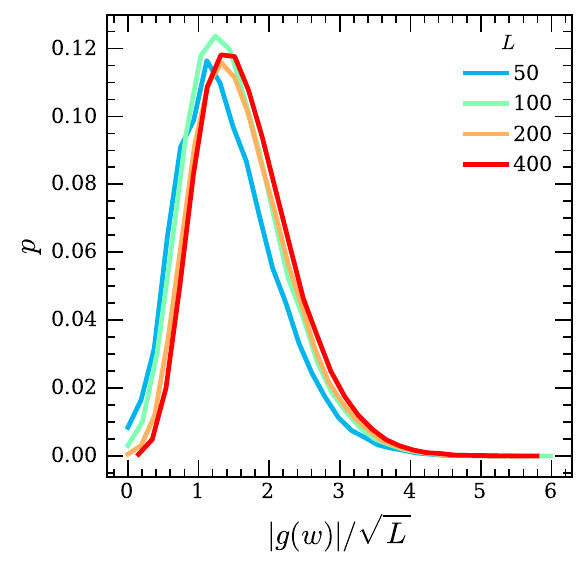}
		\caption{\label{fig:geo_dist} Histogram of the geodesic length estimate $|g(w)| = k+l+\log_2(|n|+1)$ of randomly-chosen length-$L$ words in $\bs$. The scaling collapse indicates that typical words have a geodesic length scaling as $\sqrt L$. }
	\end{figure}
	
	Numerically implementing the above procedure for several values of $L$ gives the histogram shown in Fig.~\ref{fig:geo_dist}, where for simplicity we plot an estimate of $|g(w)|$ as $k+l+\log_2(|n|+1)$, which has the same asymptotic scaling as the true value of the geodesic. 
	From this we see that the probability of getting $|g(w)| = 0$ (namely getting $w \in K_e$) is suppressed (we have already seen that ${\rm Pr}(w\in K_e) \sim e^{-L^{1/3}}$), and that the typical geodesic distance goes as $|g(w)|\sim \sqrt L$. 
	
	\begin{figure}
		\centering 
		\includegraphics[width=.35\tw]{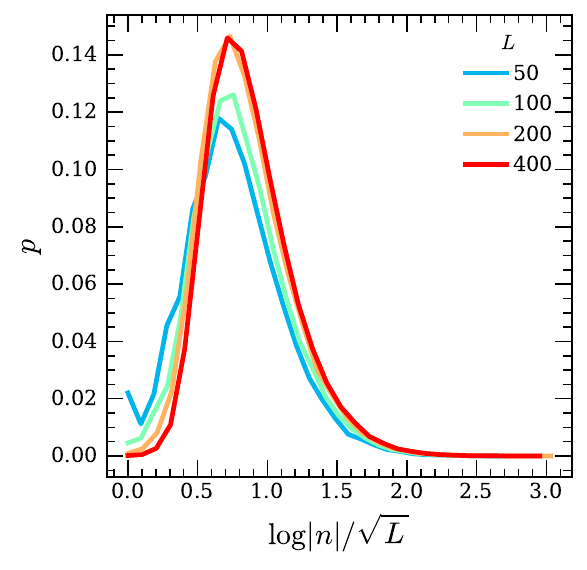}
		\caption{\label{fig:lognplot} Histogram of the distance $|n|$ that a randomly-chosen length-$L$ word proceeds along the $\tta$ axis of $\bs$'s Cayley graph. The scaling collapse indicates that typical words have $|n| \sim 2^{\sqrt L}$. }
	\end{figure}
	
	Since we are sampling random words, $n_b(w) = k-l$ will converge to a Gaussian of width $\sqrt L$; thus a typical word will reach a depth of $\sqrt L$ on the $b$-tree part of the Cayley graph. We also expect the $k+l$ contribution to the geodesic estimate \eqref{geoest} to scale as $\sqrt L$; indeed this can be numerically verified to be the case. The fact that a random word typically has $|g(w)| \sim \sqrt L$ then means that the $\log |n|$ contribution to the geodesic estimate scales as $\log|n| \lesssim \sqrt L$. This means that words which traverse {\it exponentially} far along the $a$ axis of the Cayley graph (namely those with $\log|n|\sim L$) are rare, meaning that we should not expect the Dehn time of typical words to be close to the worst-case result. 
	To determine whether $\log |n| < \sqrt L$ or $\log |n| \sim \sqrt L$, we simply make a histogram of $(\log|n| )/ \sqrt L$, with the result shown in Fig.~\ref{fig:lognplot}. 
	We see that the $\log |n|$ part adds at most an $O(1)$ contribution to the expected geodesic distance, resulting in:
	\begin{observation}\label{obs:geo_distances} 
		A typical randomly-chosen word travels to a depth of $\sim\sqrt L$ along the $b$ tree of $\cg_\bs$, and reaches a distance of $\sim 2^{\sqrt L}$ along the $a$ axis.
	\end{observation}
	
	\begin{figure}
		\centering 
		\includegraphics[scale=0.4]{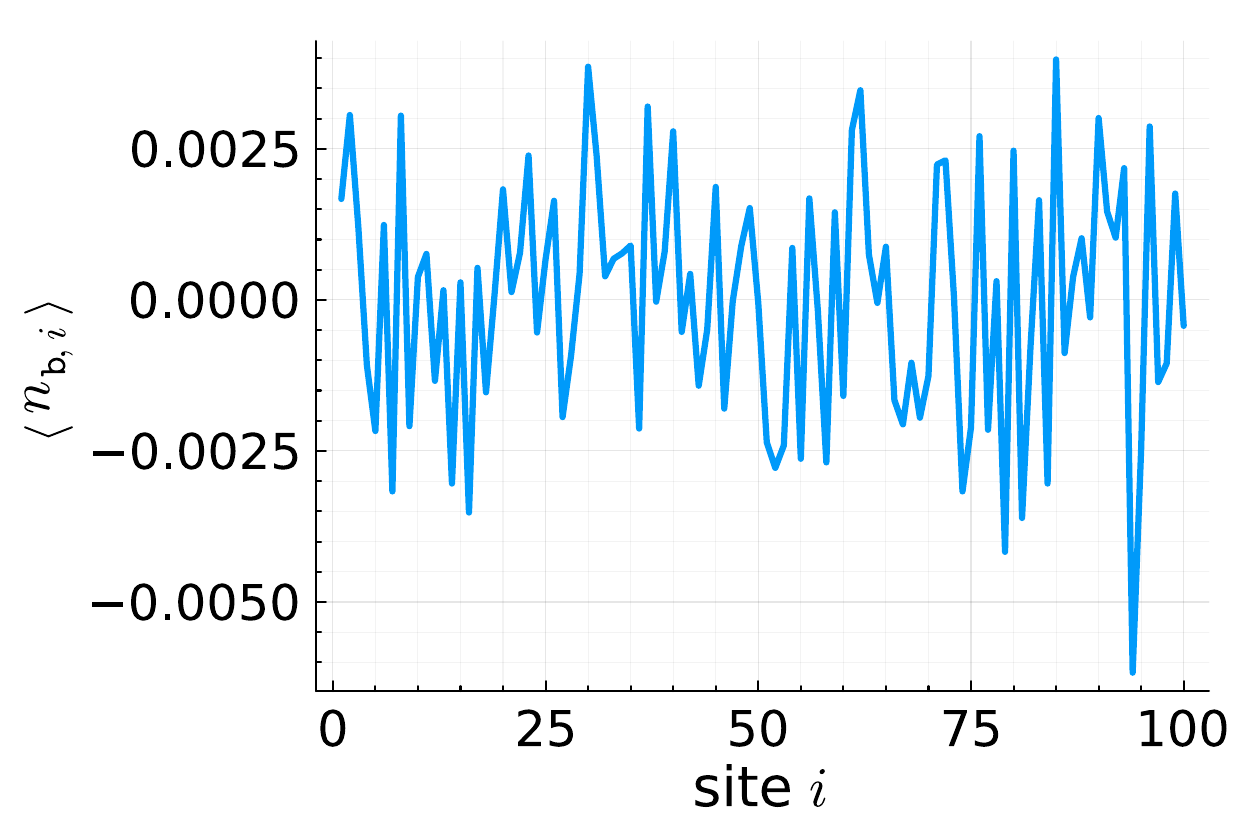}
		\caption{\label{fig:homog} $\lan w |n_{\ttb, i}|w\ran$ averaged over $10^5$ random words $w$ of length $L=100$ in the identity sector $K_{e,L}$ of $\bs$.} 
	\end{figure}

	Before moving on, we note as an aside that a randomly chosen word in $K_{e,L}$ is numerically observed to be essentially homogeneous as far as the conserved $n_\ttb$ density is concerned, meaning that $\frac1{|K_{e,L}|}\sum_{w \in K_{e,L}} \lan w | n_{\ttb,i} | w \ran \approx 0$ for all $i$ (although as far as we can tell there is no symmetry which enforces this expectation value to vanish, so we do not have a clear explanation for this). This is demonstrated in Fig.~\ref{fig:homog}, and is important for numerically determining thermalization times in the manner of Sec.~\ref{sec:numerics}.  
	
	\sss*{Dehn times of typical words} 
	
	The rough intuition leading to claim~\ref{claim:typdehn} is as follows. Let $L\in 2\nn$ for simplicity and consider a random word $w = w_L w_R$ in $K_{e, L}$, with $|w_L| = |w_R| = L/2$. Then to the extent that $w_{L,R}$ behave like random length-$L/2$ words, the midpoint of the loop defined by $w$ will be at a distance $\sim 2^{\sqrt L}$ along the $a$ axis from the origin, as follows from observation~\ref{obs:geo_distances}. Given this, we may expect that area of $w$ also scales as $2^{\sqrt L}$. The only exception is if $w_R$ follows a path which is close to $w_L\inv$ which implies that $w$ subtends little area, but we will argue that such an event is unlikely.  Intuitively, this is because the probability the  $w_R$ and $w_L\inv$ follow different branches should be large.

	We begin with an argument about certain types of words which we will argue are likely to occur as subwords of typical elements in $K_{e, L}$: 
	\begin{claim}\label{dyckclaim}
		Consider a length-$L$ word $w$ chosen randomly subject to the following constraints. First, $n_b(w) = 0$. Second, the returning walk induced on $\zz$ by restricting to the $b,b\inv$ characters of $w$ is constrained to $\zz^{\geq 0}$ (i.e. it is a {\it Dyck walk}).  Put another way, the cumulative sums $n_b(x) = \sum_{j=1}^x (\d_{w_i,b} - \d_{[w_i]_j,b} )$ are positive for all $x$. 
		
		Let $S_{Dyck}(w)$ be the set of all length-$L$ words $w' \sim w$ obeying the above constraint. For $w'$ is drawn uniformly from $S_{Dyck}(w)$, $w$ and $w'$ satisfy $d(w,w') \sim 2^{\sqrt L}$ with high probability.
	\end{claim}
	{\it Argument: }
	The Dyck walk property means that all $w'\in S_{Dyck}(w)$ are reducible to $a^{n_w}$ for some $n_w$, with $n_w$ the same for all words in $S_{Dyck}(w)$. This can be seen by an inspection of the Cayley graph, or by recalling the canonical form \eqref{wknl} (the walks corresponding to a group element are read right to left, so $b^k a^nb^{-l} \in S_{Dyck}(w)$ only if $l=k=0$). 
	
	Since $w$ was chosen randomly from the set of words obeying the Dyck walk condition, the $b$-walk defined by $w$ is (using simple concentration inequalities) likely to reach a height of $O(\sqrt L)$, namely to have $\max_x n_b(x) \sim \sqrt L$. A random $w$ fulfilling this condition can be seen to be likely to reduce to a word $a^{n_w}$ of length $n_w \sim 2^{\sqrt L}$, which can be seen by following the procedure bringing $w$ into the canonical form \eqref{wknl}. 
	
	\begin{figure}
		\centering 
		\includegraphics[width=.35\tw]{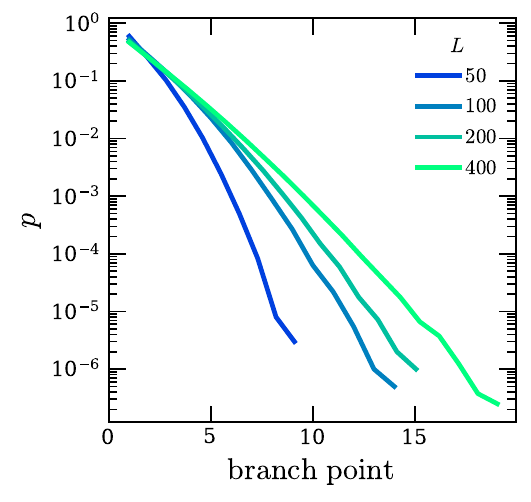}
		\caption{\label{fig:branchpoints} The probability for two returning random walks $w,w'$ on the 3-tree---with transition probabilities as in \eqref{btransprobs}---to have a given branch point $Br(w,w')$, confirming that $Br(w,w') = O(1)$ with constant probability. }
	\end{figure}
	
	Supose now that $n_w$ scales as $2^{\sqrt L}$, and consider a random element $w' \in S_{Dyck}(w)$. We claim that the distance between $w$ and $w'$ is likely to scale as $d(w,w') \sim 2^{\sqrt L}$. This is because $w,w'$ are likely to travel on different sheets of the Cayley graph for nearly all of the length of their walks. More precisely, consider the projection of $w,w'$ onto the b-tree, and define the branch point $Br(w,w')$ as the largest depth of a vertex in the tree shared by the paths undertaken by both $w$ and $w'$. We claim that $Br(w,w')$ is exponentially likely to be $O(1)$, so that $w,w'$ indeed travel on separate sheets for nearly the entirety of their trajectories. 
	
	To demonstrate this more carefully, consider how returning random walks on $\cg_\bs$ behave when projected onto the b-tree. 
	When the walk on the tree moves to larger scales of $\cg_\bs$, we will refer to it as moving ``upwards'', and when it moves to smaller scales we will say that it moves ``downwards''. At a given vertex in $\cg_\bs$, moving upwards while staying on the same sheet can be done by moving directly upwards, or by moving to the left or right by an even number of steps, and then moving up. Moving upwards onto a different sheet on the other hand is done by moving left or right by an odd number of steps before moving up. Finally, moving downwards (on the same sheet) can be done by moving an arbitrary amount either left or right, and then moving downwards. 
	
	The above reasoning shows that the probability to move downwards is equal to the probability of moving upwards, despite the fact that moving upwards can be done on either of two sheets. More precisely, let $ p^\nwarrow, p^\nearrow$ and $p^\doa$ be the probabilities of moving up on the same sheet, up on a different sheet, and down, respectively. Then 
	\bea \label{btransprobs} p^\nwarrow & = \frac45 \frac14 \sum_{k\in \zz} \sum_{l=|k|}^\infty {2l \choose l+|k|} \left(\frac{1}{4}\right)^{2l} = \frac13 \\ 
	p^\nearrow & = \frac12 - p^\nwarrow = \frac16 \\ 
	p^\doa & = \frac12.\eea
	The fact that $p^\doa= 1/2$ means that as far as motion in the tree is concerned, the walk will {\it not} move ballistically upwards or downwards, but will instead move diffusively. Using the above transition probabilities, the expected behavior of $Br(w,w')$ can be calculated analytically using generating functions. The details are rather messy however, and we content ourselves with a numerical demonstration that $Br(w,w') = O(1)$ with constant probability, see Fig.~\ref{fig:branchpoints}. 
	
	Now we return to our discussion of the distance between $w$ and $w'$.  As just argued, $Br(w,w')$ is likely to be $O(1)$. The contractibility of $\cg_\bs$ means that the minimal bounding surface linking $w$ to $w'$ must consist of all cells bounded by $w$ and the $a$ axis that lie at a depth greater than $Br(w,w')$, together with the analogous set of cells for $w'$ (a similar argument arose in the proof of proposition \ref{prop:wbig}, whereby in that case $Br(w,w')=0$). Since each of these contributions to the bounding surface consists of $\sim 2^{\sqrt L}$ cells, we thus have $d(w,w') \sim 2^{\sqrt L}$. \qed 
	
	We are now in a position to argue for the correctness of Claim~\ref{claim:typdehn}. Dyck walks do not constitute a constant fraction of all returning walks: the number of length-$L$ Dyck walks scales as $\sim L^{-3/2} 2^L$), while the number of returning walks goes as $\sim L^{-1/2}2^L$.  Thus only a fraction $\sim 1/L$ of returning walks are also Dyck walks.  Nevertheless, this still gives us enough information to argue that the {\it expected} Dehn function is $\gtrsim 2^{\sqrt{L}}$.  To show that the Dehn function is $\sim 2^{\sqrt{L}}$ with constant probability, the basic idea is to realize that a generic word in $K_e$ is likely to contain at least two subwords of size $\sim \sqrt L$ obeying the Dyck property, allowing for application of the above argument. 
	
	First, when sampling a random word in $K_e$, we can first sample uniformly from all $5^{L/2}$ words $w_L$ of length $L/2$, and then sample from all words $w_R$ of length $L/2$ such that $w_L w_R \sim e$. Since $w_L$ is chosen randomly, we expect that with unit probability in the $L\ra \infty$ limit, the walk defined by $w_L$ reaches a distance of order $\sim 2^{\sqrt L}$ along the $a$ axis of the Cayley graph (this was numerically demonstrated to be the case in Fig.~\ref{fig:lognplot}). Since $|w_L| \sim L$, reaching this distance is only possible if $w_L$ contains an excursion along a particular sheet of the $b$-tree which reaches a maximal depth of $\sim \sqrt L$. $w_L$ must therefore contain at least one subword $w_{L,D}$ which performs a Dyck walk of height $\sim \sqrt L$. 
	
	Given $w_L$, we now consider sampling words $w_R$ such that $w_L w_R \sim e$. Since $w_R$ must also traverse a distance of $\sim 2^{\sqrt L}$ along the $a$ axis, it too must contain a subword $w_{R,D}$ which performs a Dyck walk of height $\sim \sqrt L$. By the contractibility of the Cayley 2-complex, the areas that $w_{L,D}$ and $w_{R,D}$ define with respect to the $a$ axis can cancel out only if the excursions that $w_{L,D}$ and $w_{R,D}$ perform occur along the {\it same} sheets of the $b$-tree for a fraction of their respective walks close to 1. The chance for this to occur is however exponentially small in the depth of the walk (since the number of such walks on the b-tree grows exponentially in their length), which goes as $\sqrt L$. Therefore the areas contributed by $w_{L,D}$ and $w_{R,D}$ are unlikely to cancel, and thus the area of $w = w_L w_R$ will scale as $2^{\sqrt L}$ with probability approaching 1 as $L\ra\infty$. 
	
	\begin{figure}
		\centering 
		\includegraphics[width=.5\tw]{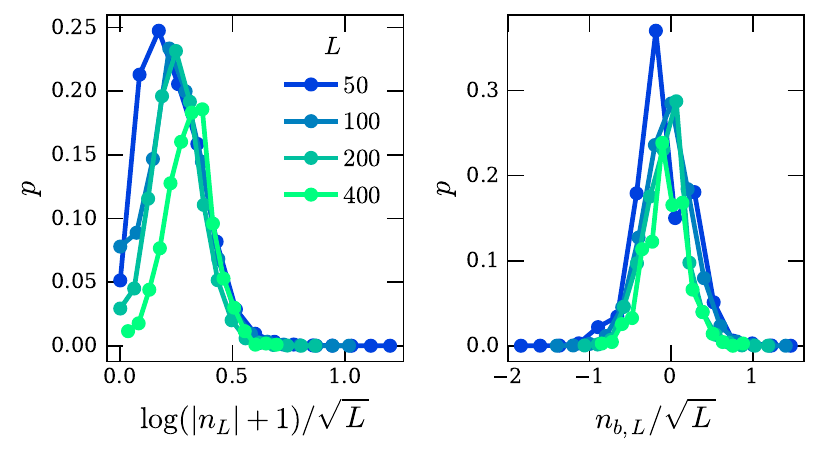}
		\caption{\label{fig:postselected_geodist} Statistical properties of words in $K_{e, L}$. For $w$ chosen randomly from $K_{e, L}$ (with $L$ even), we write $w=w_Lw_R$ with $|w_L| = |w_R| = L/2$, and then bring $w_L$ into canonical form as $w_L = b^{k_L} a^{n_L} b^{-l_L}$. {\it Left:} histograms of $\log(|n_L|+1)/\sqrt L$. We see that typical words in $K_{e, L}$ reach a distance of $\sim 2^{\sqrt L}$ along the $a$ axis of the Cayley graph at their midpoints. {\it Right:} histograms of $n_{b,L} = k_L - l_L$, showing that typical words in $K_{e, L}$ travel a depth of $\sim \sqrt L$ into the b-tree at their midpoints.  } 
	\end{figure}
	
	We end our argument for claim~\ref{claim:typdehn} by providing a degree of evidence from numerics which supports the conclusions of the above arguments. We do this by considering the distance that randomly chosen words in $K_{e, L}$ travel from the origin. More precisely, for each randomly-chosen $w \in K_{e, L}$, we subdivide $w$ into two equal-length halves as $w = w_L w_R$ and then compute the typical amount that $w_L$ extends along the $a$ axis of $\cg_\bs$. Sampling uniformly from $K_{e, L}$ is done by sampling random words and then postselecting on them being in $K_{e, L}$, which for large $L$ is numerically costly on account of the postselection succeeding with probably $\sim e^{-\a L^{1/3}}$. This limits our numerics to relatively modest system sizes, for which finite size effects are rather strong. Nevertheless, the results shown in Fig.~\ref{fig:postselected_geodist} shown that typical words in $K_{e, L}$ have $w_L$ which indeed extends a distance of $\sim 2^{\sqrt L}$ along the $a$ axis of $\cg_\bs$. As argued above, this implies the correctness of claim~\ref{claim:typdehn} (since this is false only if $w_L$, $w_R$ almost always traverse the same path of the b-tree, which we have already demonstrated is unlikely). \qed 
	
	Finally, we note that by following similar logic as in the proof of proposition~\ref{prop:sectorindep}, the same scaling as in claim~\ref{claim:typdehn} can be argued to occur for randomly chosen words, and not just those in $K_e$. More precisely, from claim~\ref{claim:typdehn} we expect that if $w$ is a random length-$L$ word and $w'$ is subsequently randomly chosen from $K_{g(w)}(L)$, then $d(w,w') \sim 2^{\sqrt L}$ with constant probability.

	\section{Iterated Baumslag-Solitar groups} \label{app:iterated_bs}
	
	In this appendix we discuss a construction for modifying the Bamsulag-Solitar group $\bs(n,m)$ so that the corresponding Dehn function scales superexponentially in $L$. From the spacetime bound \eqref{spacetimebound} this immediately allows us to construct examples of models with superlinear fragility lengths. 
	
	In general, variants of $\bs(n,m)$ can be constructed whose Dehn functions grow as fast as $\dehn(L) \sim \exp^{\lfloor \log_2 L\rfloor }(1)$ \cite{baumslag1969non}, a rapidly growing function of $L$. We will however study the simplest variants of this class where the Dehn function scales doubly-exponentially $\dehn(L) \sim 2^{2^L}$. A family of groups 
	$\bs^{(2)}(n,m,o,p)$ with this scaling are defined in their simplest presentations by three generators $\tta,\ttb,\ttc$, where $\tta,\ttb$ satisfy the relations of $\bs(n,m)$ and $\ttb,\ttc$ satisfy those of $\bs(o,p)$: 
	\be \bs^{(2)}(n,m,o,p) = \lan \tta,\ttb,\ttc \,|\, \tta^m\ttb=\ttb\tta^n , \, \ttb^o\ttc=\ttc\ttb^p\ran. \ee
	We will focus throughout on the case $n=o=1, m=p=2$, which provides the simplest nontrivial example. For notational brevity, we will write $\bs^{(2)}(1,2,1,2)$ simply as $\bs^{(2)}$.
	
	The Cayley graph of $\bs^{(2)}$ consists of an infinite hierarchy of $\bs$ Cayley graphs. Multiplying by $\ttc$ (i.e. moving along the $\ttc$ direction) increases the ``scale'' of the $\bs$ Cayley graph generated by $\tta,\ttb$, while multiplying by $\ttc\inv$ decreases the scale. Basic facts about the geometry of $\cg_{\bs^{(2)}}$ such as its growth rate and the size of the identity sector are relatively difficult to address numerically, partly because unlike $\bs$, $\bs^{(2)}$ does not admit a simple linear representation, making it difficult to compute $g(w)$ for an arbitrary word $w$.\footnote{It is in fact an outstanding open question if there exist groups with $\dehn(L) > 2^L$ for which the word problem can be solved in linear time (for instance, a matrix representation for $BS^{(2)}$ would answer this question in the affirmative). }
	
	One may anticipate that $\bs^{(2)}$ should have a superexponential Dehn function due to the hierarchical structure of its Cayley graph, since the $\ttc$ characters can be used to create exponential expansion of the $\ttb$ characters, which in turn can create doubly-exponential expansion of the $\tta$ characters. The following result makes this intuition more precise: 
	
	\bpro 
	$\bs^{(2)}$ has Dehn function 
	\be \dehn(L) \sim 2^{2^L}.\ee 
	\epro 
	
	This result appears to be well-known and was stated in Ref.~\cite{clay2017office} without an explicit proof;  we could not find a proof of this statement and thus provide one below for completeness. 
	The key result we need to complete the proof is known as {\it Britton's lemma} \cite{lyndon1977combinatorial},\footnote{We thank Tim Riley for suggesting this proof strategy.} which is stated as follows: 
	\begin{lemma}
		Let $G$ be a group with presentation $S$. Further let $G$ contain two isomorphic subgroups $H,K\subset G$, with $\phi: H \ra K$ the isomorphism between them. Define the group 
		\be G_\phi \equiv \lan S, t \, | \, t\inv H t = \phi(H) \ran\ee 
		
		and wolog express any $w$ on $\{S,t\}^*$ in the form 
		\be w = g_0 t^{\ep_1} g_1 t^{\ep_2} g_2 \cdots t^{\ep_{n-1}} g_{n-1} t^{\ep_n} g_n,\qq g_i \in G \, \, \ep_i = \pm1 .\ee 
		Britton's lemma states that if $w \sim e$ represents the identity in $G_\phi$, then there must be some $i$ such that either 
		\begin{enumerate}
			\item $n=0$ and $g_0 = e$, 
			\item $\ep_i=-1, \ep_{i+1} = 1$ and $g_i \in H$, or 
			\item $\ep_i = 1, \ep_{i+1} = -1$ and $g_i \in K$.	
		\end{enumerate} 
	\end{lemma}
	
	This gives us a way of simplifying any word $w \sim e$ representing the identity in $G_\phi$. Case 1 above is trivial. In case 2, we can replace the occurrence of $t\inv h t$ with $\phi(h)$, while in case 3 we may replace $t k t\inv$ with $\phi\inv(k)$. After doing this reduction, we will be left with a new word $w'$ which still represents $e$, and can apply the lemma again. Iterating, we are guaranteed to eliminate all $t$s from any $w\sim e$ in $\{S,t\}^*$ to obtain a new word $w_G \in \{S\}^*, w_G \sim e$. For instance, the group $\bs(1,2)$ is simply $G_{\phi}$ for the choices $G = \zz$, $H = G$, and $K = 2\zz$. 
	
	We now return to a proof of the Dehn function scaling: 
	
	\bproof We first construct a lower bound. Define the word $w_n = \ttc^{-n} \ttb \ttc^{n}$, so that $w_n \sim \ttb^{2^n}$. Then feed this word into the construction of the large-area word $w_{big}$ for $\bs$, by defining $w_n' = w_n a w_n\inv$. Then we claim the word 
	\be w_{huge} = (w_n')\inv \tta w_n' \tta\inv \ee 
	has area 
	\be \dehn(w_{huge}) \sim 2^{2^n},\ee 
	which is doubly exponential in $|w_{huge}|$. This follows by an argument similar to the one we gave for the area of $w_{big}$ in $\bs$. The tree-like structure of the sheets of the $\bs$ Cayley graph give tree-like structures both for words built from $\ttb,\ttc$ and those built from $\tta,\ttb$. Letting $\wt w_{huge} = \ttb^{-2^n} \tta \ttb^{2^n}$, this means that
	\be \dehn(w_{huge}) = \dehn(\wt w_{huge} \tta \wt w_{huge}\inv \tta\inv) + O(2^n).\ee 
	But using our results from our study of $\bs$, we know that $ \dehn(\wt w_{huge} \tta \wt w_{huge}\inv \tta\inv) \sim 2^{2^n}$. Thus $\dehn(L) \geq 2^{2^L}$ asymptotically. 
	
	We now need to provide a matching upper bound. We do this by combining Britton's lemma with our earlier result about $\bs$. Note that $\bs^{(2)}$ can be obtained from $\bs$ using just the type of extension as appears in Britton's lemma, where $G = \bs, H = \lan \ttb\ran, K = \lan \ttb^2\ran$. Then we know that if we are given $w \sim \tte$ where $|w|=L$ in $\bs^{(2)}$, we can obtain a word $w'\sim \tte$ in $\bs$ after at most $O(L)$ applications of $\ttc\inv \ttb \ttc = \ttb^2$. The maximum amount that $|w|$ can grow by under these substitutions is $O(2^L)$. Thus an upper bound on $\dehn(L)$ in $\bs^{(2)}$ can therefore be obtained by an upper bound on $\dehn(2^L)$ in $\bs$. Using our previous result on the latter, we conclude 
	\be \dehn(L) \lesssim 2^{2^L},\ee 
	and thus when combined with the lower bound, we also have that $\dehn(L) \sim 2^{2^L}$. 
	\eproof 
	
	$\bs^{(2)}$ also provides an example with a superlinear expansion length: 
	\begin{corollary}
		$\bs^{(2)}$ has exponential expansion length, 
		\be \el(L) \sim 2^L.\ee 
	\end{corollary}
	\bproof 
	From the general bound \eqref{spacetimebound} and our above result about the Dehn function of $\bs^{(2)}$, we know that $\el(L) \geq 2^L$. The upper bound follows from the above application of Britton's lemma and the fact that the expansion length of $\bs$ is only $O(L)$. 
	\eproof  
	
	It is easy to generalize the above example to construct groups with even faster growing space and time complexity: 
	\begin{corollary}
		Define the group $\bs^{(l)}$ through the presentation \cite{clay2017office}
		\be \bs^{(l)} = \lan \tta_0,\dots,\tta_l \, | \, \tta_{i-1} \tta_i = \tta_i \tta_{i-1}^2, \, i = 1 , \dots, n \ran. \ee 
		This group has Dehn function and expansion length 
		\be \dehn(L) \sim \exp^{(l)}(L), \qq \el(L) \sim \exp^{(l-1)}(L).\ee 
	\end{corollary}
	
	\section{Detailed analysis of non-group examples}\label{app:nongroup}
	In this appendix we provide some detailed analysis studying both the time and space complexity of the non-group examples presented in the main text.
	\ss*{Star-Motzkin model}
	
	Recall the Star-Motzkin model from Sec.~\ref{sec:starmotzkin}.  There are two sources of fragmentation; the first originates from the parentheses, and the second from the interaction of the parentheses with the $\ast$ character.
	
	In the following analysis, we will rely on the fact that we can write down a non-local conserved quantity under the dynamics that necessarily reflects the interaction between the Motzkin degrees of freedom with the $\ast$ degrees of freedom.  This non-local conserved quantity is  
	\begin{equation}
		Q = \sum_i 2^{\sum_{j < i} n_{(,j} - n_{),j}} n_{\ast, i}.
	\end{equation}
	where $n_{c,i} = \proj{c}_i$.  The interpretation of this operator can be understood pictorially.  A configuration of parentheses (ignoring the $\ast$ character) can be mapped onto a height configuration.  The height $h_i$ is written as $h_i = \sum_{j < i} n_{(,j} - n_{),j}$.  Bringing a $\ast$ character at height $h_i$ down to height $h_j < h_i$ by a sequence of local updates creates $2^{h_i - h_j}$ such $\ast$ characters at height $h_j$.  The charge operator simply counts the total number of $\ast$ characters if {\it all} $\ast$ characters are brought down to zero height.  In the situation where the height becomes negative, we can rescale $2^{\sum_{j < i} \hat{n}_{(,j} - \hat{n}_{),j}}$ to $2^{\sum_{j < i} \hat{n}_{(,j} - \hat{n}_{),j} - \overline{h}}$, where $\overline{h}$ is the value of the lowest height.  The interpretation of the exponential charge is then equivalent to the previously introduced definition if $\overline{h}$ is redefined to be at zero height.
	
	We now proceed to understanding how to label Krylov sectors of the dynamics.  For simplicity, we start with analyzing the Krylov sector $K_Q$ corresponding to the balanced sector for the parentheses, with total value $Q$ for the non-local conserved quantity.  We expect this model to exhibit fragile fragmentation with a linear expansion length (much like $\dyn_\bs$) and as such, we will discuss connectivity of configurations assuming they are appended to a reservoir of $0$ characters of length $\alpha L$.
	
	We claim that the dynamics is ergodic within the Krylov sector $K_Q$ so long as $\alpha = O(1)$ is sufficiently large.  The proof follows from finding a path between any two states in $K_Q$ given the desired size of the reservoir of $O(L)$.  To do this, we construct a reference state $R$, which we show can be reached from all states given the provided space: 
	\begin{equation}
		R = ( ( \cdots (\ast^{m_1} )\cdots \ast^{m_{k-2}} ) \ast^{m_{k-1}} ) \ast^{m_{k}}
	\end{equation}
	where the configuration $( ( \cdots (\ast^{m_1} )\cdots \ast^{m_{k-2}} ) \ast^{m_{k-1}} ) \ast^{m_{k}}$ holds all of the non-local conserved quantity, and the $m_i \in \{0,1\}$ are the numbers appearing in the binary representation of $Q$: 	 
	\begin{equation}
		Q = \sum_{i=1}^k m_i 2^{k-i}.
	\end{equation}
	Without any $\ast$ characters ($Q = 0$), the dynamics is entirely ergodic within $K_0$.  Therefore, the goal is to show that when $Q \neq 0$, all of the $\ast$'s can be isolated to a reference configuration $R$ on one side of the system.
	
	To this end, given a configuration $C$, we first isolate $O(\log Q)$ number of $0$ characters on one side of the system and use these to create a nest $((\cdots()\cdots))$ of parentheses.  Next, we construct an algorithm for localizing the entirety of the non-local conserved quantity into the nest, thereby reducing $C$ to the reference configuration $R$.  Consider the $\ast$ character nearest to the nest.  Labelling the nest by `$q$' (the current amount of the conserved quantity inside of it), this character is positioned as such:
	\begin{equation}
		q \, ( ( ( \cdots ( \ast \cdots
	\end{equation}
	where the number of open parentheses is $p$.  Next, we move the $\ast$ character via the following sequence.
	\begin{equation}
		(\cdots ( ( ( \ast \to (\cdots ((\ast\ast ( \to ( \cdots (\ast \ast (\ast ( \to ( \cdots \ast \ast (\ast (\ast ( \to \cdots
	\end{equation}
	Repeating, we can bring this to the canonical form
	\begin{equation}
		\ast \ast ( \ast (\ast \cdots \ast ( \ast (\cdots
	\end{equation}
	The nest may then absorb the two $\ast$ characters adjacent to it, forming $q' \,\, ( \ast (\ast \cdots \ast ( \ast (\cdots$ where $q' = q+1$.  The next step is to collapse any paired parentheses: i.e. $()$ is sent to $00$.  For instance, this kind of collapse will occur for the configuration
	\begin{equation}
		q \, ( ( ( \cdots ( \ast ) \cdots
	\end{equation}
	since the canonical form is $q' \,\, ( \ast (\ast \cdots \ast ( \ast ()\cdots \to q' \,\, ( \ast (\ast \cdots \ast ( \ast \cdots$.  
	
	After absorbing one unit of the conserved quantity in the nest and performing the collapse process, we iterate these two steps.  By construction, this algorithm will eventually localize the conserved quantity in the nest, forming the reference configuration $R$.  So long as the reservoir is large enough it is possible to transition from any configuration $C \in K$ to $R$, therefore proving ergodicity.

	This result indicates to us that the dynamics is ergodic within the $K$ sector (up to a mild form of fragile fragmentation that exist due to $\el(L) > L$).  It also indicates that the dynamics may be slow, since transporting the non-local conserved quantity out of a region appears to take a very large number of steps.  This would indicate that the relaxation times of $\dyn_{\ast M}$ mimic that of $\dyn_{\bs}$.
	\begin{figure}
		\centering
		\includegraphics[scale=0.3]{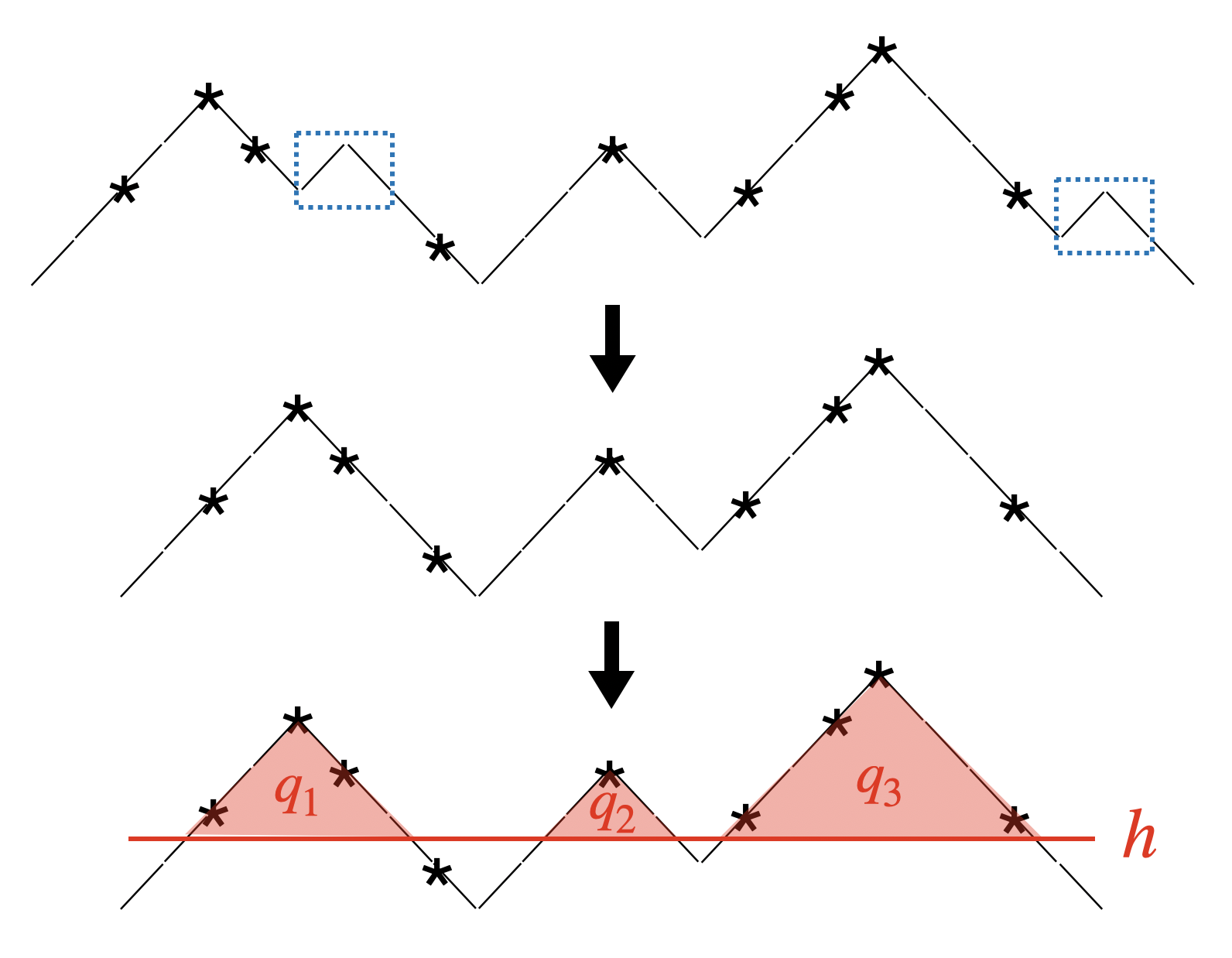}
		\caption{An illustration of the process used to define $h$-restrictions.  We first clean the configuration, as shown in the upper panel.  Then we draw a line at height $h$ and consider all contiguous regions above this line, which are shaded in red in the bottom panel.  We label the non-local charge in each of these regions by $q_i$.}
		\label{fig:hres}
	\end{figure}
	To formalize this, we define the notion of an {\it $h$-restriction} of $C$.  First, we perform a preprocessing step where we eliminate as many $()$ pairs as possible in $C$.  The new configuration $\widetilde{C}$ is what we will call a {\it clean} version of $C$.  A $h$-restriction of $\widetilde{C}$, which we denote by $S_h(\widetilde{C})$ is formed by first drawing a reference line at a height $h$ (note that the height profile of the parentheses is shifted so the minimum height is at 0). $S_h(\widetilde{C})$ denotes a set of contiguous configurations above height $h$ -- an example is denoted in Figure~\ref{fig:hres}.  We also label the total non-local conserved quantity in a contiguous configuration $c$ by $q_c = \sum_{i \in c} n_{i,\ast} 2^{h_i}$.  Though this quantity should not be associated with a `charge' corresponding to a symmetry, we will henceforth, in an abuse of notation, call it a charge. 
	
	Suppose two contiguous configurations $c$ and $c'$ have charge $q_c$ and $q_c'$.  If we want to transport an amount of charge $\Delta q$ between these two configurations, how much time will this take?  Note that because $c$ and $c'$ are supported above height $h$, charge must be pumped in or out of them at increments of $2^h$.  Therefore, the amount of time required is at least $O(\Delta q/2^h)$.  Since the value of the charge supported in $c$ or $c'$ can be exponentially large, transporting some fraction of the charge in $c$ to $c'$ can take exponentially long time, so long as $h$ is not too large.
	
	With this in mind, we can discuss how long it takes to transition between two configurations.  Denote the $h$-restriction of $C$ to be $S_h(\widetilde{C})$ and that of $C'$ to be $S_h(\widetilde{C}')$.  Label the charges of the $h$-restriction of $C$ to be (in decreasing order) $q_1 \geq q_2 \geq \cdots, \geq q_{n_C}$ and for $C'$ to be $q'_1 \geq q'_2 \geq \cdots, \geq q'_{n_{C'}}$.  Here, $n_C$ denotes the number of contiguous regions in the $h$-restriction of $C$.  If $n_C \neq n_{C'}$, then some number of contiguous regions need to be created\footnote{These new regions can either be created from scratch or can be created by borrowing a part of an existing contiguous region.  However, one can show that in both cases, the amount of time required to pump charge $q$ into this new region will be at least $q/2^h$.  Thus, in what follows we can assume wolog that new contiguous configurations are constructed from scratch.}.  In this case, assuming that wolog $n_C \geq n_{C'}$, we construct the two vectors
	\begin{equation}
		\vec{q} = \langle q_1, q_2, \cdots,q_{n_C}\rangle \hspace{0.5cm} \vec{q}\,' = \langle q_1', q_2', \cdots,q_{n_{C'}}',0,\cdots,0\rangle
	\end{equation}
	where the number of $0$'s in $\vec{q}\,'$ is $n_{C} - n_{C'}$ indicating a number of yet-to-be created contiguous configurations.  The number of charge that needs to be transferred in and out of these contiguous configurations is at least $\Delta q = \norm{\vec{q} - \vec{q}\,'}_1$.  The amount of time required to do this is therefore
	\begin{equation}
		t_{hit} \geq 2^{-h} \norm{\vec{q} - \vec{q}\,'}_1.
	\end{equation}
	
	As a result of this bound, there is a simple method for checking whether the time to go between two configurations is very long.  Given two configurations $C$ and $C'$, we first construct clean versions and successively raise the value of $h$ until their $h$-restricted charge vectors are significantly different in 1-norm.  At this point, so long as $h$ is not too large, we know that it will take a long time to traverse between these configurations.
	
	Note that the hitting times strongly depend on the total value of $Q$.  In particular, we have the obvious bound $Q \leq L \cdot 2^{\max_i h_i}$.  For a randomly chosen height profile, $\max_i h_i = O(\sqrt{L})$ and thus we expect $Q \sim \exp(\sqrt{L})$.  As a result, we may expect that for generic sectors the hitting time could scale like $\sim \exp(\sqrt{L})$, the same scaling as the one argued for in typical Dehn function of the Baumslag-Solitar group (See Eqn.~\eqref{avgdehnbs} and App.~\ref{app:bs_average}).
	
	\ss*{Chiral star-Motzkin}
	We now provide a more in depth analysis of the chiral star-Motzkin model discussed in Sec.~\ref{sec:chstarmotzkin}. Note that, like in the star-Motzkin dynamics, the chiral star-Motzkin dynamics features a non-local conserved quantity:
	\begin{equation}
		Q_R = \sum_i 2^{\sum_{j < i} n_{(,j} - n_{),j}} n_{\triangleright, i}.
	\end{equation}
	To understand why large spatial resources are needed, we first consider the following warm-up example.  Suppose we have the configuration
	\begin{equation}
		C = ((\cdots (\triangleright)  \cdots ) \triangleright) ((\cdots (\triangleright)  \cdots ))
	\end{equation}
	and we want to convert it to the configuration
	\begin{equation}
		C' = ((\cdots (\triangleright)  \cdots )) ((\cdots (\triangleright)  \cdots )\triangleright).
	\end{equation}
	In essence, we want to move a single unit of charge from one of the nests to the other one.  We can move the $\triangleright$ out of the first nest yielding
	\begin{equation}
		((\cdots (\triangleright)  \cdots )) \triangleright\triangleright ((\cdots (\triangleright)  \cdots )).
	\end{equation}
	but unlike in the star-Motzkin model, we cannot move this into the other cluster.  The only way to proceed is to collapse the entire cluster, giving
	\begin{equation}
		(\cdots (\triangleright)  \cdots) \triangleright^{2^h+2} (\cdots () \cdots) \to (\cdots (\triangleright)  \cdots ) \triangleright^{2^h+2}
	\end{equation}
	where $h$ is the height of the nest that was collapsed.  Upon doing this, we may reinsert an empty nest of parentheses forming
	\begin{equation}
		(\cdots (\triangleright)  \cdots ) (\cdots ()  \cdots ) \triangleright^{2^h+2}.
	\end{equation}
	which when $2^h$ of the $\triangleright$s are used to populate the center of the right next with an $\triangleright$, gives $C'$.
	
	Let us more rigorously understand when a large amount of spatial resources are needed.  We work in an intrinsic Krylov sector with fully matched parenthesis ($m = n = 0$) and with a fixed value of the non-local conserved quantity $Q_R = Q$.  Consider an $h$-restriction of configurations $C$ and $C'$.  Label the charges of the contiguous regions (in decreasing order) for $C$ and $C'$ with $\vec{q}$ and $\vec{q}\,'$; this notation was introduced in the discussion of the star-Motzkin model.  We need to transfer an amount of charge to convert between charge configurations $\vec{q}$ and $\vec{q}\,'$; define
	\begin{equation}
		\Delta q_{\max} = \max_{i} |q_i - q_i'|
	\end{equation}
	This is (a lower bound on) the maximum amount of charge that has to be transferred out of a single contiguous region.  Some of this charge can be deposited below height $h$.  However, the maximum amount of charge below height $h$ is $L 2^h$.  Therefore, if $\Delta q_{\max} - L 2^h$ is large, then this remaining amount of charge must be transferred to different contiguous region.  As there are at most $L$ contiguous regions, at some point in time during the charge transfer process, an amount of charge has to be inserted in a contiguous region of charge at least $\Delta q_{\max}/L - 2^h$.  This requires an amount of space $O(\Delta q_{\max}2^{-h}/L)$ since at minimum the entire region needs to be collapsed down to height $h$ before inserting the charge.  Therefore, if $\Delta q_{\max} = \epsilon Q$ then this space can be very large.  This is in fact an extremely loose bound, but sufficient for our purpose of showing large space complexity.
	
	Therefore, as in the quasi-fragmented example, given two configurations $C$ and $C'$, one first constructs clean versions of these configurations and then selects a height $h$ such that the $h$-restriction of $C$ and $C'$ has large $\Delta q_{\max}$.  If this is the case, the space complexity scales linearly in $\Delta q_{\max}2^{-h}$ up to polynomial factors in $L$. 
	
	\ss*{Non-groups and thermalization} 
	
	Since we could also construct group based examples with large time and space complexities, it is an interesting question to ask whether there are qualitatively new features that non-group based constraints provide to the dynamics.
	
	We answer this question by examining the structure of reduced density matrices of subsystems under the dynamics.  Recall that under {\it group dynamics} $\dyn_G$, reduced density matrices of subsystems, defined as $\rho_A = \Tr_{A^c}\left(\dyn_G^{\dagger} \, \rho_0 \,  \dyn_G\right)$, will have non-zero values along their diagonals. To explain why, consider a decomposition of the system $S = A A^c$.  Start with an initial product state
	\begin{equation}
		\ket{\psi_0} = \ket{u}_A \otimes \ket{v}_{A^c},
	\end{equation}
	and suppose $v$ has $m$ zeroes (and can be converted under the dynamics to some canonical form $0^m v'$).  Under the dynamics, we can perform a sequences of transitions that converts $\ket{u}_A \otimes \ket{0^m v'}_{A^c}$ to $\ket{0^{|A|}}_A \otimes \ket{0^{m-|A|} uv'}_{A^c}$ and subsequently
	\begin{equation}
		\ket{0^{|A|}}_A \otimes \ket{0^{m-|A|} uv'}_{A^c} \to \ket{w}_A \otimes \ket{w^{-1}0^{m-|A|-|w|} uv'}_{A^c}.
	\end{equation}
	If the number of zeroes in the initial state is large enough, then any word $w$ can be produced in $A$, therefore implying that $\rho_A$ will have all diagonal elements nonzero.  Note that this crucially relies on the existence of inverses, hence why it is special to a group structure.
	
	However, this property {\it no longer} applies for dynamics which are not based on groups.  Instead, we argue that for certain non-group examples, it is not possible to attain all words $u$ in subsystem $A$.  To see this, let us consider the chiral star-Motzkin model from the previous subsection.  Consider the initial word in the subsystem to be
	\begin{equation}
		u = ( ( ( \cdots (\triangleright) \cdots ) ) )
	\end{equation}
	and the entire initial state $\ket{u} \otimes \ket{v}$ to be in $K_Q$ for some large value of $Q$.  Let us assume that $v$ is generic enough that a constant fraction of it is filled with '$0$' characters.  We define
	\begin{equation}
		n_{\ast, A} = \sum_{i \in A} \ketbra{\ast}{\ast}_i
	\end{equation}
	and track the probability distribution $p(n_{\ast, A})$ over time, where
	\begin{equation}
		p(n, t) = \mel{\psi (t)}{\mathcal{P}_{n}}{\psi(t)}
	\end{equation}
	where $\mathcal{P}_{n}$ projects onto configurations with $n_{\ast, A} = n$. Let us first suppose that $|A| \ll \log L$.  Then, we contend that $p(n, t) > 0$ for all $0 \leq n \leq |A|$.  To show this is simple: in order to allow for $0 \leq n \leq |A|$ we need to be able to annihilate all of the parentheses in $A$.  Since this can be done in roughly $2^{|A|/2}$ space and $|A| \ll \log L$, configurations with all possibly densities $n_{\ast, A}|$ can be reached.
	
	However, when $|A| \gg \log L$ it is no longer that case that all configurations with densities $0 \leq n_{\ast, A} \leq |A|$ are achievable.  To see this, consider sectors with $Q \gg \exp(|A|)$.  Note that in order to have configurations with $n_{\ast, A} = |A|/2 + q$, we need to annihilate at least $q$ pairs of parentheses in $A$.   If we set the height of the $h$-restriction to be $h = |A|/2 - q$ and the value of $Q \gg \exp(|A|)$, there will be at least two contiguous regions.  To transport out an amount of charge $2^q$ from the contiguous region in $A$ to a contiguous region outside $A$ requires at least $2^q/\text{poly}(L)$ space.  If $q \gg \log L$ then this is not possible.  Therefore, we find that for subsystems of size $\gg \log L$, {\it no configurations} with charge $n_{\ast, A} = |A|/2 + q$ are reachable with $q \gg \log L$.  This implies an unusual property that the reduced density matrix $\rho_A$ will not have full rank unless the size of the subsystem is smaller than $\log L$.  In this example, the consequences are easily observable since the value of $n_{\ast, A}$ is a local operator.

\end{document}